\documentclass[11pt]{article}
\usepackage{graphicx}    
\usepackage{tocloft}
\usepackage{amsmath,amssymb,amsthm,amsfonts}
\usepackage{latexsym,bbm,graphicx,float,mathtools}

\usepackage{verbatim}

\usepackage{hyperref}  
\hypersetup{
	colorlinks,
	citecolor=black,
	filecolor=black,
	linkcolor=black,
	urlcolor=black,
	linktoc=all
}

\newtheorem{theorem}{Theorem}
\newtheorem{lemma}{Lemma}[section]
\newtheorem{claim}[lemma]{Claim}

\newtheorem{corollary}[lemma]{Corollary}

\newtheorem{remark}[theorem]{Remark}

\newtheorem{definition}[lemma]{Definition}

\theoremstyle{definition}

\theoremstyle{plain}

\newcommand{\ignore}[1]{}

\newcommand{\Ex}{\mathop{{\bf E}\/}}
\renewcommand{\Pr}{\operatorname{{\bf Pr}}}
\newcommand{\Prx}{\mathop{{\bf Pr}\/}}

\newcommand{\Ber}{\operatorname{Ber}}

\newcommand{\poly}{\mathrm{poly}}

\newcommand{\dist}{\mathrm{dist}}

\newcommand{\Alg}{\mathrm{Alg}}

\newcommand{\N}{\mathbbm N}
\newcommand{\Z}{\mathbbm Z}

\newcommand{\eps}{\varepsilon}
\renewcommand{\epsilon}{\eps}

\newcommand{\ol}[1]{\overline{#1}}
\newcommand{\oA}{\ol{A}}

\newcommand{\oS}{\ol{S}}

\newcommand{\bA}{\mathbf{A}}
\newcommand{\bB}{\mathbf{B}}
\newcommand{\bC}{\mathbf{C}}

\newcommand{\bG}{\mathbf{G}}
\newcommand{\bH}{\mathbf{H}}

\newcommand{\bL}{\mathbf{L}}
\newcommand{\bM}{\mathbf{M}}

\newcommand{\bP}{\mathbf{P}}
\newcommand{\bQ}{\mathbf{Q}}

\newcommand{\bS}{\mathbf{S}}
\newcommand{\bT}{\mathbf{T}}

\newcommand{\bV}{\mathbf{V}}
\newcommand{\bW}{\mathbf{W}}
\newcommand{\bX}{\mathbf{X}}
\newcommand{\bY}{\mathbf{Y}}
\newcommand{\bZ}{\mathbf{Z}}

\newcommand{\boldf}{\boldsymbol{f}} 

\newcommand{\bb}{\boldsymbol{b}}
\newcommand{\bd}{\boldsymbol{d}}
\newcommand{\bg}{\boldsymbol{g}}
\newcommand{\bh}{\boldsymbol{h}}
\newcommand{\bi}{\boldsymbol{i}}
\newcommand{\bj}{\boldsymbol{j}}

\newcommand{\bm}{\boldsymbol{m}}

\newcommand{\boldr}{\boldsymbol{r}} 

\newcommand{\bv}{\boldsymbol{v}}

\newcommand{\bx}{\boldsymbol{x}}
\newcommand{\by}{\boldsymbol{y}}

\newcommand{\bell}{\boldsymbol{\ell}}

\newcommand{\bGamma}{{\mathbf{\Gamma}} }

\newcommand{\eqdef}{\stackrel{\rm def}{=}}

\newcommand{\calD}{\mathcal{D}}
\newcommand{\calE}{\mathcal{E}}

\newcommand{\calG}{\mathcal{G}}
\newcommand{\calH}{\mathcal{H}}

\newcommand{\calN}{\mathcal{N}}

\newcommand{\calP}{\mathcal{P}}

\newcommand{\calU}{\mathcal{U}}
\newcommand{\calV}{\mathcal{V}}

\newcommand{\calY}{\mathcal{Y}}

\newcommand{\Dyes}{\calD_{\text{yes}}}
\newcommand{\Dno}{\calD_{\text{no}}}

\def\FullBox{\hbox{\vrule width 8pt height 8pt depth 0pt}}
\newcommand{\QED}{\;\;\;\FullBox}
\renewenvironment{proof}{\noindent{\bf Proof:~~}}{\hfill\QED}
\newenvironment{proofof}[1]{\noindent{\bf Proof of {#1}:~~}}{\hfill\(\QED\)}

\def\authornameAL{Amit Levi}
\def\authoraffiAL{University of Waterloo. Email: \href{mailto: amit.levi@uwaterloo.ca}{amit.levi@uwaterloo.ca}.}

\def\authoraffiEW{Columbia University. Email: \href{mailto: eaw@cs.columbia.edu}{eaw@cs.columbia.edu}.}
\def\Dyes{\mathcal{D}_{\text{yes}}}
\def\Dno{\mathcal{D}_{\text{no}}}
\def\bM{\mathbf{M}}
\def\bA{\mathbf{A}}
\def\bH{\mathbf{H}}
\def\Junta{\mathrm{Junta}}

\def\obM{\overline{\bM}}
\def\obA{\overline{\bA}}

\def\calEy{\calE_{\text{yes}}}
\def\calEn{\calE_{\text{no}}}

\title{Lower Bounds for Tolerant Junta and Unateness Testing\\ via Rejection Sampling of Graphs}

\author{
\authornameAL\thanks{\authoraffiAL}
\and
Erik Waingarten\thanks{\authoraffiEW}}
\begin{document}         
\maketitle

\begin{abstract}
We introduce a new model for testing graph properties which we call the \emph{rejection sampling model}. We show that testing bipartiteness of $n$-nodes graphs using rejection sampling queries requires complexity $\widetilde{\Omega}(n^2)$. Via reductions from the rejection sampling model, we  give three new lower bounds for tolerant testing of Boolean functions of the form $f\colon\{0,1\}^n\to \{0,1\}$:
\begin{itemize}
\item Tolerant $k$-junta testing with \emph{non-adaptive} queries requires $\widetilde{\Omega}(k^2)$ queries.
\item Tolerant unateness testing requires $\widetilde{\Omega}(n)$ queries.
\item Tolerant unateness testing with \emph{non-adaptive} queries requires $\widetilde{\Omega}(n^{3/2})$ queries.
\end{itemize}
Given the $\widetilde{O}(k^{3/2})$-query non-adaptive junta tester of Blais \cite{B08}, we conclude that non-adaptive tolerant junta testing requires more queries than non-tolerant junta testing. In addition, given the $\widetilde{O}(n^{3/4})$-query unateness tester of Chen, Waingarten, and Xie \cite{CWX17b} and the $\widetilde{O}(n)$-query non-adaptive unateness tester of Baleshzar, Chakrabarty, Pallavoor, Raskhodnikova, and Seshadhri \cite{BCPRS17}, we conclude that tolerant unateness testing requires more queries than non-tolerant unateness testing, in both adaptive and non-adaptive settings. These lower bounds provide the first separation between tolerant and non-tolerant testing for a natural property of Boolean functions.
\end{abstract}

\thispagestyle{empty}

\newpage
\thispagestyle{empty}
\setcounter{tocdepth}{2}
\tableofcontents
\thispagestyle{empty}
\newpage
\setcounter{page}{1}

\newcommand{\calbE}{\boldsymbol{\mathcal{E}}}

\section{Introduction}

Over the past decades, property testing has emerged as an
important line of research in sublinear time algorithms. 
{The} goal is to understand randomized algorithms for approximate decision making, where the algorithm needs to decide (with high probability) whether a huge object has some property by making a few queries to the object. Many different types of objects and properties
have been studied from this property testing perspective
(see the surveys by Ron \cite{R08,R10} and the recent textbook by Goldreich \cite{G17} for overviews of
contemporary property testing research). This paper deals with property testing of Boolean functions $f \colon \{0, 1\}^n \to \{0, 1\}$ and property testing of graphs with vertex set $[n]$. 

In this paper we describe a new model of graph property testing, which we call the \emph{rejection sampling model}. For $n \in \N$ and a subset $\calP$ of graphs on the vertex set $[n]$, we say a graph $G$  on vertex set $[n]$ has property $\calP$ if $G \in \calP$ and say $G$ is $\eps$-far from having property $\calP$ if all graphs $H \in \calP$ differ on at least $\eps n^2$ edges\footnote{The distance definition can be modified accordingly when one considers bounded degree or sparse graphs. }. The problem of $\eps$-testing $\calP$ with \emph{rejection sampling queries} is the following task:
\begin{quote}
	Given some $\eps > 0$ and access to an unknown graph $G=([n],E)$, output ``accept" with probability at least $\frac{2}{3}$ if $G$ has property $\calP$, and output ``reject'' with probability at least $\frac{2}{3}$ if $G$ is $\eps$-far from having property $\calP$. The access to $G$ is given by the following oracle queries: given a query set $L \subseteq [n]$, the oracle samples an edge $(\bi, \bj) \sim E$ uniformly at random and returns $\{ \bi, \bj\} \cap L$. 
\end{quote}
We measure the complexity of algorithms with rejection sampling queries by considering the sizes of the queries. The complexity of an algorithm making queries $L_1, \dots, L_t \subset [n]$ is  $\sum_{i=1}^t |L_i|$.

The rejection sampling model allows us to study testers which rely on random sampling of edges, while providing the flexibility of making lower-cost queries. This type of query access strikes a delicate balance between simplicity and generality: queries are constrained enough for us to show high lower bounds, and at the same time, the flexibility of making queries allows us to reduce the rejection sampling model to Boolean function testing problems. Specifically, we reduce to tolerant junta testing and tolerant unateness testing (see Subsection~\ref{sec:intro-apps}). 

Our main result in the rejection sampling model is regarding \emph{non-adaptive} algorithms. These algorithms need to fix their queries in advance and are not allowed to depend on answers to previous queries (in the latter case we say that the algorithm is \emph{adaptive}).
We show a lower bound on the complexity of testing whether an unknown graph $G$ is bipartite using non-adaptive queries.
\begin{theorem}\label{thm:bipartite}
	There exists a constant $\eps > 0$ such that any non-adaptive $\eps$-tester for bipartiteness in the rejection sampling model has cost $\widetilde{\Omega}(n^2)$\footnote{We use the notations $\widetilde{O},\widetilde{\Omega}$  to hide polylogarithmic dependencies on the argument, i.e. for expressions of the form $O(f\log^c f)$ and $\Omega(f /\log^c f)$ respectively (for some absolute constant $c$).}.
\end{theorem}
More specifically, Theorem~\ref{thm:bipartite} follows from applying Yao's principle to the following lemma.
\begin{lemma}\label{lem:graphs}
	Let $\calG_1$ be the uniform distribution over the union of two disjoint cliques of size $n/2$, and let $\calG_2$ be the uniform distribution over complete bipartite graphs with each part of size $n/2$. Any deterministic non-adaptive algorithm that can distinguish between $\calG_1$ and $\calG_2$ with constant probability using rejection sampling queries, must have complexity $\widetilde{\Omega}(n^2)$. 
\end{lemma}

We discuss a number of applications of the rejection sampling model (specifically, of Lemma~\ref{lem:graphs}) in the next subsection. In particular, we obtain new lower bounds in the \emph{tolerant testing framework} introduced by Parnas, Ron, and Rubinfeld in \cite{parnas2006tolerant} for {two} well-studied properties
of Boolean functions (specifically, $k$-juntas and unateness; see the next subsection for definitions of these properties). These lower bounds are obtained by a reduction from the rejection sampling model; we show that too-good-to-be-true Boolean function testers for these properties imply the existence of rejection sampling algorithms which distinguish $\calG_1$ and $\calG_2$ with $\tilde{o}(n^2)$ complexity. Therefore, we may view the rejection sampling model as a useful abstraction in studying the hard instances of tolerant testing $k$-juntas and unateness.


\subsection{Applications to Tolerant Testing: Juntas and Unateness}\label{sec:intro-apps}

Given $n \in \N$ and a subset $\calP$ of $n$-variable Boolean functions, a Boolean function $f \colon \{0, 1\}^n \to \{0, 1\}$ has property $\calP$ if $f \in \calP$. The distance between Boolean functions $f, g\colon \{0, 1\}^n \to \{0, 1\}$ is $\dist(f, g) = \Prx_{\bx \sim \{0,1\}^n}[f(\bx) \neq g(\bx)]$. The distance of $f$ to the property $\calP$ is $\dist(f, \calP) = \min_{g \in \calP} \dist(f, g)$. We  say that $f$ is $\eps$-close to $\calP$ if $\dist(f, \calP) \leq \eps$ and $f$ is $\eps$-far from $\calP$ if $\dist(f, \calP) > \eps$. The problem of \emph{tolerant property testing} \cite{parnas2006tolerant} of $\calP$ asks for query-efficient randomized algorithms for the following task:
\begin{quote}
	Given parameters $0 \leq \eps_0 < \eps_1 < 1$ and black-box query access to a Boolean function $f \colon \{0, 1\}^{n} \to \{0, 1\}$, accept with probability at least $\frac{2}{3}$ if $f$ is $\eps_0$-close to $\calP$ and reject with probability at least $\frac{2}{3}$ if $f$ is $\eps_1$-far from $\calP$.
\end{quote}
An algorithm which performs the above task is an $(\eps_0, \eps_1)$-tolerant tester for $\calP$. A $(0,\eps_1)$-tolerant tester is a \emph{standard} property tester or a \emph{non-tolerant} tester. As noted in \cite{parnas2006tolerant}, tolerant testing is not only a natural generalization, but is also very often the desirable attribute of testing algorithms. This motivates the high level question: how does the requirement of being tolerant affect the complexity of testing the properties studied?  We make progress on this question by showing query-complexity separations for two well-studied properties of Boolean functions: $k$-juntas, and unate functions.
\begin{itemize}
	\item ($k$-junta) A function $f\colon\{0,1\}^n\to\{0,1\}$ is a \textit{$k$-junta} if it depends on at most $k$ of its variables, i.e., there exists $k$ distinct indices $i_1, \dots i_k \in [n]$ and a $k$-variable function $g \colon \{0, 1\}^k \to \{0, 1\}$ where $f(x) = g(x_{i_1}, \dots, x_{i_k})$ for all $x \in \{0, 1\}^n$.
	\item (unateness) A function $f\colon\{0,1\}^n\to\{0,1\}$ is \emph{unate} if $f$ is either non-increasing or non-decreasing in every variable. Namely, there exists a string $r \in \{0, 1\}^n$ such that the function $f(x \oplus r)$ is monotone with respect to the bit-wise partial order on $\{0, 1\}^n$.
\end{itemize}

The next theorem concerns non-adaptive tolerant testers for $k$-juntas.

\begin{theorem}\label{thm:junta}
	For any $\alpha < 1$, there exists constants $0 < \eps_0 < \eps_1 < 1$ such that for any $k = k(n) \leq \alpha n$, any non-adaptive $(\eps_0, \eps_1)$-tolerant $k$-junta tester must make $\widetilde{\Omega}(k^2)$ queries.
\end{theorem}

We give a noteworthy consequences of the Theorem~\ref{thm:junta}. In \cite{B08}, Blais gave a non-adaptive $\widetilde{O}(k^{3/2})$-query tester for (non-tolerant) testing of $k$-juntas, which was shown to be optimal for non-adaptive algorithms by Chen, Servedio, Tan, Waingarten and Xie in \cite{CSTWX17}. Combined with Theorem~\ref{thm:junta}, this shows a polynomial separation in the query complexity of non-adaptive tolerant junta testing and non-adaptive junta testing.

The next two theorems concern tolerant testers for unateness.
\begin{theorem}\label{thm:unate}
	There exists constants $0 < \eps_0 < \eps_1 < 1$ such that any (possibly adaptive) $(\eps_0, \eps_1)$-tolerant unateness tester must make $\widetilde{\Omega}(n)$ queries.
\end{theorem}

\begin{theorem}\label{thm:unate-non}
	There exists constant $0 < \eps_0 < \eps_1 < 1$ such that any non-adaptive $(\eps_0, \eps_1)$-tolerant unateness tester must make $\widetilde{\Omega}(n^{3/2})$ queries.
\end{theorem}

A similar separation in tolerant and non-tolerant testing occurs for the property of unateness as a consequence of Theorem~\ref{thm:unate} and Theorem~\ref{thm:unate-non}. Recently, in \cite{BCPRS17}, Baleshzar, Chakrabarty, Pallavoor, Raskhodnikova, and Seshadhri gave a non-adaptive $\widetilde{O}(n)$-query tester for (non-tolerant) unateness testing, and Chen, Waingarten and Xie~\cite{CWX17} gave an (adaptive) $\widetilde{O}(n^{3/4})$-query tester for (non-tolerant) unateness testing. We thus, conclude that by Theorem~\ref{thm:unate} and Theorem~\ref{thm:unate-non}, tolerant unateness testing is polynomially harder than (non-tolerant) unateness testing, in both adaptive and non-adaptive settings.

\subsection{Related Work}

The properties of $k$-juntas and unateness have received much attention in property testing research (\cite{FKRSS:04, CG04, B08, B09, BGSMW13, STW15, CSTWX17, BCELR16} study $k$-juntas, and \cite{GGLRS00, KS16, CS16b, BCPRS17, CWX17, CWX17b} study unateness). We briefly review the current state of affairs in (non-tolerant) $k$-junta testing and unateness testing,
and then discuss tolerant testing of Boolean functions and the rejection sampling model.

\paragraph{Testing $k$-juntas.} The problem of testing $k$-juntas, introduced by Fischer, Kindler, Ron, Safra, and Samorodnitsky~ \cite{FKRSS:04}, is now well understood up to poly-logarithmic factors. Chockler and Gutfreund \cite{CG04} show that any tester for $k$-juntas requires $\Omega(k)$ queries (for a constant $\eps_1$). Blais~\cite{B09} gave a junta tester that uses $O(k\log k+k/\eps_1)$ queries, matching the bound of \cite{CG04} up to a factor of $O(\log k)$ for constant $\eps_1$. When restricted to non-adaptive algorithms, \cite{FKRSS:04} gave a non-adaptive tester making $\widetilde{O}(k^2/\eps_1)$ queries, which was subsequently improved in~\cite{B08} to $\widetilde{O}(k^{3/2})/\eps_1$. In terms of lower bounds, 

Buhrman, Garcia-Soriano, Matsliah, and de Wolf~\cite{BGSMW13} gave a $\Omega(k \log k)$ lower bound for $\eps = \Omega(1)$, and Servedio, Tan, and Wright \cite{STW15} gave a lower bound which showed a separation between adaptive and non-adaptive algorithms for $\eps_1 = \frac{1}{\log k}$. These results were recently improved in ~\cite{CSTWX17} to $\widetilde{\Omega}(k^{3/2}/\eps_1)$, settling the non-adaptive query complexity of the problem up to poly-logarithmic factors.

\paragraph{Testing unateness.} The problem of testing unateness was introduced alongside the problem of testing monotonicity in Goldreich, Goldwasser, Lehman, Ron, and Samorodnitsky \cite{GGLRS00}, where they gave the first $O(n^{3/2}/\eps_1)$-query non-adaptive tester. Khot and Shinkar \cite{KS16} gave the first improvement by giving a $\widetilde{O}(n/\eps_1)$-query adaptive algorithm. A non-adaptive algorithm with $\widetilde{O}(n/\eps_1)$ queries was given in \cite{CS16, BCPRS17}. Recently, \cite{CWX17,BCPRS17b} show that $\widetilde{\Omega}(n)$ queries are necessary for non-adaptive one-sided testers. Subsequently, \cite{CWX17b} gave an adaptive algorithm testing unateness with query complexity $\widetilde{O}(n^{3/4}/\eps_1^2)$. The current best lower bound for general adaptive testers appears in \cite{CWX17}, where it was shown that any adaptive two-sided tester must use $\widetilde{\Omega}(n^{2/3})$ queries.

\paragraph{Tolerant testing.} Once we consider tolerant testing, i.e., the case $\epsilon_0>0$, the picture is not as clear. In the paper introducing tolerant testing, \cite{parnas2006tolerant} observed that standard algorithms whose queries are uniform (but not necessarily independent) are inherently tolerant to some extent.  Nevertheless, achieving  $(\eps_0,\eps_1)$-tolerant testers for constants $0<\eps_0< \eps_1$, can require applying different methods and techniques (see e.g, \cite{GR05,parnas2006tolerant,FN07,ACCL07,KS09,MR09,FR10,CGR13,BRY14b,BMR16, Tel16}).

By applying the observation from \cite{parnas2006tolerant} to the unateness tester in \cite{BCPRS17}, the tester accepts functions which are $O(\eps_1/n)$-close to unate with constant probability. We similarly obtain weak guarantees for tolerant testing of $k$-juntas. Diakonikolas, Lee, Matulef, Onak, Rubinfeld, Servedio, and Wan \cite{diakonikolas2007testing} observed that one of the (non-adaptive) junta testers from~\cite{FKRSS:04} accepts functions that are $\poly(\eps_1,1/k)$-close to $k$-juntas. Chakraborty, Fischer, Garcia-Sor\'{i}ano, and Matsliah~\cite{CFGM:12}  noted that the analysis of the junta tester of Blais~\cite{B09}  implicitly implies  an  $\exp(k/\eps_1)$-query complexity tolerant tester which accepts functions that are $\eps_1/c$-close to some $k$-junta (for some constant $c > 1$) and rejects functions that are $\eps_1$-far from every $k$-junta. Recently, Blais, Canonne, Eden, Levi and Ron \cite{BCELR16} showed that when required to distinguish between the cases that $f$ is $\eps_1/10$-close to a $k$-junta, or is $\eps_1$-far from a $2k$-junta, $\poly(k,1/\eps_1)$ queries suffice. 

For general properties of Boolean functions, tolerant testing could be much harder than standard testing. Fischer and Fortnow \cite{FF06} used PCPs in order to construct a property of Boolean functions $\calP$ which is $(0, \eps_1)$-testable with a constant number of queries (depending on $\eps_1$), but any $(1/4, \eps_1)$-tolerant test for $\calP$ requires $n^{c}$ queries for some $c > 0$. While \cite{FF06} presents a strong separation between tolerant and non-tolerant testing, the complexity of tolerant testing of many natural properties remains open. We currently neither have a $\poly(k, \frac{1}{\eps_1})$-query tester which $(\eps_0, \eps_1)$-tests $k$-juntas, nor a $\poly(n, \frac{1}{\eps_1})$-query tester that $(\eps_0, \eps_1)$-tests unateness or monotonicity when $\eps_0 = \Theta(\eps_1)$.

\paragraph{Testing graphs with rejection sampling queries.} Even though the problem of testing graphs with rejection sampling queries has not been previously studied, the model shares characteristics with previous studied frameworks. 
These include sample-based testing studied by Goldreich, Goldwasser, and Ron in \cite{GGR98, GR16}, where the oracle receives random samples from the input. One crucial difference between rejection sampling algorithms (which always query $[n]$) and sample-based testers is the fact that rejection sampling algorithms only receive \emph{positive} examples (in the form of edges), as opposed to random positions in the adjacency matrix (which may be a \emph{negative} example indicated the non-existence of an edge).

The rejection sampling model for graph testing also bears some resemblance to the conditional sampling framework for distribution testing introduced in Canonne, Ron, and Servedio, as well as Chakraborty, Fischer, Goldhirsh, and Matsliah  \cite{CRS15, CFGM16}, where the algorithm specifies a query set and receives a sample conditioned on it lying in the query set. 

\subsection{Techniques and High Level Overview}\label{sec:high-level}

We first give an overview of how the lower bound in the rejection sampling model (Lemma~\ref{lem:graphs}) implies lower bounds for tolerant testing of $k$-juntas and unateness, and then we give an overview of how Lemma~\ref{lem:graphs} is proved.

\paragraph{Reducing Boolean Function Testing to Rejection Sampling} This work should be considered alongside some recent works showing lower bounds for testing the properties of monotonicity, unateness, and juntas in the standard property testing model \cite{BB16, CWX17, CSTWX17}. The lower bounds in \cite{BB16, CWX17} and \cite{CSTWX17} may be reinterpreted as following the same general paradigm. We discuss this general view next, followed by an overview of this work. At a high level, one may view the lower bounds from \cite{BB16, CWX17, CSTWX17} as proceeding in three steps:
\begin{enumerate}
	\item First, design a randomized indexing function $\bGamma \colon \{0, 1\}^n \to [N]$ that partitions the Boolean cube $\{0, 1\}^n$ into roughly equal parts in a way compatible with the property (either monotonicity, unateness, or junta). We want to ensure that algorithms that make few queries cannot learn too much about $\bGamma$, and that queries falling in the same part are close in Hamming distance.
	\item Second, define two distributions over sub-functions $\bh_i \colon \{0, 1\}^n \to \{0, 1\}$ for each $i \in [N]$. The hard functions are defined by $\boldf(x) = \bh_{\bGamma(x)}(x)$, so that one distribution corresponds to functions with the property, and the other distribution corresponds to functions far from the property.
	\item Third, show that any testing algorithm for the property is actually solving some algorithmic task (determined by the distributions of $\bh_i$) which is hard when queries are close in Hamming distance.
\end{enumerate}
Belovs and Blais \cite{BB16} used a construction of Talagrand \cite{T96}, known as the Talagrand function, to implement a randomized partition in a monotone fashion. The Talagrand function is a randomized DNF of $2^{\sqrt{n}}$ monotone terms of size $\sqrt{n}$, and one may define $\bGamma \colon \{0, 1\}^n \to [2^{\sqrt{n}}]$ to output the index of the first term of a Talagrand function which satisfies input $x \in \{0,1\}^n$. One can show that any two queries $z, z' \in \{0, 1\}^n$ which are semi-balanced\footnote{We will say $z \in \{0, 1\}^n$ is \emph{semi-balanced} if $|z| \approx \frac{n}{2} \pm \sqrt{n}$.} with Hamming distance more than $\widetilde{\Omega}(n^{3/4})$ will fall in different parts with high probability. The sub-functions $\bh_i \colon \{0, 1\}^n \to \{0, 1\}$ are then given by random dictators or random anti-dictators, so the algorithmic task is simple: determine whether the distribution over functions $\bh_i$ is supported on dictators or anti-dictators when queries in the same part are at distance at most $\widetilde{O}(n^{3/4})$ from each other. An argument in the spirit of the one-sided error monotonicity lower bound from \cite{FLNRRS02} gives an $\Omega(n^{1/4})$ lower bound for monotonicity testing. \cite{CWX17} further refined the idea by designing improved randomized partitions $\bGamma \colon \{0, 1\}^n \to [N]$, which they called two-level Talagrand functions. The improved construction $\bGamma$ partitions  $\{0, 1\}^n$ in a monotone fashion, but has the property that queries $z, z' \in \{0,1\}^n$ which are semi-balanced with Hamming distance $\widetilde{\Omega}(n^{2/3})$ fall into different parts with high probability, thus bringing the lower bound to $\widetilde{\Omega}(n^{1/3})$ using the same algorithmic task as \cite{BB16}.

Higher lower bounds for unateness are possible because the unateness property allows for reductions to harder algorithmic tasks. Specifically, \cite{CWX17} consider the following algorithmic task: there are two classes of distributions supported on $[n] \times \{ +, - \}$, and the task is to distinguish two classes with random samples. One class of distributions consists of the uniform distribution $\mu$ over $[n] \times \{+, -\}$, the other class of distributions is uniform over the support, but each $\mu$ satisfies the property that each $j \in [n]$ has either $\mu(j, +) = 0$ or $\mu(j, -) = 0$. Each sub-function $\bh_i$ is specified by a random sample of $\mu$, where $\bh_i$ is a dictator in variable $j$ if $(j, +)$ was sampled, and an anti-dictator in variable $j$ if $(j, -)$ was sampled. The first key observation is that the distance of the functions $\boldf(x) = \bh_{\bGamma(x)}(x)$ from unateness,
depends on whether $\mu$ comes from the first or second case. The second key observation is that multiple random samples are required to distinguish the two classes of distributions.\footnote{For example, in order to distinguish whether a distribution $\mu$ belongs to the first or second class with one-sided error, an algorithm must observe two samples $(j, +)$ and $(j, -)$ from $\mu$, which would indicate that $\mu$ is uniform over the whole set $[n] \times \{+, -\}$. In fact, the adaptive algorithm for unateness testing in \cite{CWX17b} can be interpreted as one based on solving this algorithmic task with a ``rejection sampling''-style oracle.}

For the case of $k$-juntas, \cite{CSTWX17} used a simple indexing function $\bGamma \colon \{0,1\}^{n}\to [2^{n/2}]$ that partitions $\{0, 1\}^n$ according to projections on randomly chosen $\frac{n}{2}$ variables. The second and third step also follows the above strategy. In their case, they define the $\mathsf{SSSQ}$ and $\mathsf{SSEQ}$ (for Set-Size-Set-Queries and Set-Size-Element-Queries) problems as the hard algorithmic task, which give the lower bounds. 

Our lower bounds for tolerant testing follow the same paradigm. For the randomized indexing function, we use the construction from \cite{CSTWX17} for the junta lower bound and a Talagrand-based construction (similar to \cite{CWX17}, but somewhat simpler) for the unateness lower bounds. The hard algorithmic task we embed is distinguishing between the distributions $\calG_1$ and $\calG_2$ with access to a rejection sampling oracle.

At a high level, our reductions show that the class of functions which are close to $k$-juntas and the class of functions which are close to unate have much richer structure than $k$-juntas and unate functions. In particular, the distance of the functions drawn from our hard distributions to $k$-junta and unateness will depend on a global parameter of an underlying graph used to define the functions\footnote{The relevant graph parameter in $k$-juntas and unateness will be different. Luckily, both graph parameters will have gaps in their value depending on the distribution the graphs were drawn from (either $\calG_1$ or $\calG_2$).  This allows us to reuse the work of proving Lemma~\ref{lem:graphs} to obtain Theorem~\ref{thm:junta}, Theorem~\ref{thm:unate}, and Theorem~\ref{thm:unate-non}.}. Thus, tolerant testing algorithms for $k$-juntas and unateness must explore the relationships between different variables to gain some information about the underlying graph. This lies in stark contrast to the algorithms of \cite{B08}, \cite{CWX17b}, and \cite{BCPRS17} which test $k$-juntas (non-adaptively) and unateness, since these three algorithms treat the variables independently.

The distributions and the reductions themselves are quite involved, so we defer a high level overview of the reductions to those corresponding sections (Sections~\ref{sec:junta} and~\ref{sec:unate}).

\paragraph{Distinguishing $\calG_1$ and $\calG_2$ with Rejection Sampling Queries}
In order to prove Lemma~\ref{lem:graphs}, one needs to rule out any deterministic non-adaptive algorithm which distinguishes between $\calG_1$ and $\calG_2$ with rejection sampling queries of complexity $\widetilde{o}(n^2)$. In order to keep the discussion at a high level, we identify three possible ``strategies'' for determining whether an underlying graph is a complete bipartite graph, or a union of two disjoint cliques:
\begin{enumerate}
	\item One approach is for the algorithm to sample edges and consider the subgraph obtained from edges returned by the oracle. For instance, the algorithm may make all rejection sampling queries to be $[n]$. These queries are expensive in the rejection sampling model, but they guarantee that an edge from the graph will be observed. If the algorithm is lucky, and there exists a triangle in the subgraph observed, the graph must not be bipartite, so it must come from $\calG_2$.
	\item Another sensible approach is for the algorithm to forget about the structure of the graph, and simply view the distribution on the edges generated by the randomness in the rejection sampling oracle as a distribution testing problem. Suppose for simplicity that the algorithm makes rejection sampling queries $[n]$. Then, the corresponding distributions supported on edges from $\calG_1$ and $\calG_2$ will be $\Omega(1)$-far from each other, so a distribution testing algorithm can be used.
	\item A third, more subtle, approach is for the algorithm to use the fact that $\calG_1$ and $\calG_2$ correspond to a complete bipartite graph and the union of two cliques, and extract knowledge about the non-existence of edges when making queries which return either $\emptyset$ or a single vertex. More specifically, an algorithm may query a random subset $L \subset [n]$ of size $\frac{n}{2}$. The subset $L$ will be split among the two sides of the graph (in the case of $\calG_1$ and $\calG_2$), and when an edge sampled by the oracle is incident on only one vertex of $L$, the rejection sampling oracle will return this one vertex. At this point, the algorithm may extract some information about how $L$ is divided in the underlying graph, and eventually distinguish between $\calG_1$ and $\calG_2$.
\end{enumerate}
The three strategies mentioned above all fail to give $\widetilde{o}(n^2)$ rejection sampling algorithms. The first approach fails because with a budget of $\widetilde{o}(n^2)$, rejection sampling algorithms will observe subgraphs which consist of various trees of size at most $\log n$, thus we will not observe cycles. The second approach fails since the distributions are supported on $\Omega(n^2)$ edges, so distribution testing algorithms will require $\Omega(n)$ edges (which costs $\Omega(n^2)$) to distinguish between $\calG_1$ and $\calG_2$. Finally, the third approach fails since algorithms will only observe $o(n)$ responses from the oracle corresponding to lone vertices which will be split roughly evenly among the unknown parts of the graph, so these observations will not be enough to distinguish between $\calG_1$ and $\calG_2$.

Our lower bound rules out the three strategies sketched above when the complexity is $\widetilde{o}(n^2)$, and shows that if the above three strategies do not work (in any possible combination with each other as well), then no non-adaptive algorithm of complexity $\widetilde{o}(n^2)$ will work. The main technical challenge is to show that the above strategies are the \emph{only} possible strategies to distinguish $\calG_1$ and $\calG_2$. In Section~\ref{sec:LowerBound}, we give a more detailed, yet still high-level discussion of the proof of Lemma~\ref{lem:graphs}.

Finally, the analysis of Lemma~\ref{lem:graphs} is tight; there is a non-adaptive  rejection sampling algorithm which distinguishes $\calG_1$ and $\calG_2$ with complexity $\widetilde{O}(n^2)$. The algorithm (based on the first approach mentioned above) is simple: make $\widetilde{O}(n)$ queries $L= [n]$, and if we observe an odd-length cycle, we output ``$\calG_1$'', otherwise, output ``$\calG_2$''. 

\section{Preliminaries}
We use boldfaced letters such as $\bA, \bM$ to denote random variables. Given a string $x\in\{0, 1\}^n$ and $j\in[n]$, we write $x^{(j)}$ to denote the string obtained from $x$ by flipping the $j$-th coordinate. An edge along the $j$-th direction in $\{0,1\}^n$ is a pair $(x,y)$ of strings with $y = x^{(j)}$. In addition, for $\alpha\in\{0,1\}$ we use the notation $x^{(j\rightarrow \alpha)}$ to denote the string $x$ where the $j$th coordinate is set to $\alpha$.
Given $x\in\{0,1\}^n$ and $S\subseteq[n]$, we use $x|_S\in\{0,1\}^S$ to denote the projection of $x$ on $S$. For a distribution $\calD$ we write $\bd \sim \calD$ to denote an element $d$ drawn according to the distribution. We sometimes write $a \approx  b\pm c$ to denote $b - c \leq a \leq b+c$. 

Throughout this paper, we extensively use a generalization of Chernoff bounds for \textit{negatively correlated} random variables. 
\begin{definition} Let $\bX_1,\ldots,\bX_n \in\{0, 1\}$ be random variables. We say that $\bX_1,\ldots,\bX_n$ are \emph{negatively correlated} if for all $I\subset[n]$ the following hold:
	\begin{align*} 
	\Prx\left[\forall i\in I\;:\; \bX_i=0\right]&\le \prod_{i\in I}\Prx\left[\bX_i=0\right]\\
	\Prx\left[\forall i\in I\;:\; \bX_i=1\right]&\le \prod_{i\in I}\Prx\left[\bX_i=1\right]\;.
	\end{align*}
\end{definition}

\begin{theorem} [Theorem $1.16$ from \cite{doerr2011analyzing}]Let $\bX_1,\ldots, \bX_n$ be negatively correlated binary random variables. Let $a_1,\ldots,a_n\in[0,1]$ and $\bX=\sum_{i=1}^na_i\bX_i$. Then,  for $\delta\in[0,1]$,
	\begin{align*}
	\Prx\left[ \bX\ge(1+\delta)\Ex \left[\bX\right] \right]&\le \exp(-\delta^2\Ex [\bX]/2)\\
	\Prx\left[ \bX\le(1-\delta)\Ex \left[\bX\right] \right]&\le \exp(-\delta^2\Ex [\bX]/3)\;.
	\end{align*}
\end{theorem}

In addition, some of our proofs will use \emph{hyper-geometric} random variables. Consider a population of size $N$ that consists of $K$ objects of a special type. Suppose $n$ objects are picked without replacement. Let $\bX$ be a random variable that counts the number of special objects picked in the sample. Then, we say that $\bX$ is a hyper-geometric random variable, and we denote $\bX\sim \mathrm{HG}(N,K,n)$. These hyper-geometric random variables enjoy tight concentration inequities (which are similar to Chernoff type bounds).
\begin{theorem}[\cite{Hoef63}]\label{thm:HG-Chernoff}
	Let $\bX\sim\mathrm{HG}(N,K,n)$ and  $\mu=K/N$. Then for any $t>0$
	\begin{align*}
	\Prx\left[ \bX\le (\mu-t) n\right]&\le \exp(-2t^2 n)\\
	\Prx\left[ \bX\ge (\mu+t) n\right]&\le \exp(-2t^2 n)\;.
	\end{align*}
\end{theorem}

\def\cost{\mathrm{cost}}

\section{The Rejection Sampling Model}
In this section, we define the rejection sampling model and the distributions over graphs we will use throughout this work. We define the rejection sampling model tailored to our specific application of proving Lemma~\ref{lem:graphs}.

\begin{definition}\label{def:rejec}
	Consider two distributions, $\calG_1$ and $\calG_2$ supported on graphs with vertex set $[n]$. The problem of distinguishing $\calG_1$ and $\calG_2$ with a rejection sampling oracle aims to distinguish between the following two cases with a specific kind of query:
	\begin{itemize}
		\item Cases: We have an unknown graph $\bG \sim \calG_1$ or $\bG \sim \calG_2$. 
		\item Rejection Sampling Oracle: Each query is a subset $L\subset [n]$; an oracle samples an edge $(\bj_1, \bj_2)$ from $\bG$ uniformly at random, and the oracle returns $\bv = \{ \bj_1, \bj_2 \} \cap L$. The complexity of a query $L$ is given by $|L|$. 
	\end{itemize} 
\end{definition}

We say a non-adaptive algorithm $\text{Alg}$ for this problem is a sequence of query sets $L_1, \dots, L_q \subset [n]$, as well as a function $\Alg \colon \left( [n] \cup \left([n] \times [n]\right) \cup \{ \emptyset \}\right)^q \to \{ \text{``$\calG_1$''}, \text{``$\calG_2$''} \}$. The algorithm sends each query to the oracle, and for each query $L_i$, the oracle responds $\bv_i \in [n] \cup \left([n] \times [n]\right) \cup \{ \emptyset \}$, which is either a single element of $[n]$, an edge in $\bG$, or $\emptyset$. The algorithm succeeds if:
\[ \Prx_{\substack{\bG \sim \calG_1, \\ \bv_1, \dots, \bv_q}} \left[\Alg(\bv_1, \dots, \bv_q) \text{ outputs ``$\calG_1$''} \right] - \Prx_{\substack{\bG \sim \calG_2, \\ \bv_1, \dots, \bv_q}} \left[ \Alg(\bv_1, \dots, \bv_q) \text{ outputs ``$\calG_1$''}\right] \geq \frac{1}{3}. \]
The complexity of $\text{Alg}$ is measured by the sum of the complexity of the queries, so we let $\cost(\text{Alg}) = \sum_{i=1}^q |L_i|$.

While our interest in this work is primarily on lower bounds for the rejection sampling model, an interesting direction is to explore upper bounds of various natural graph properties with rejection sampling queries. Our specific applications only require ruling out non-adaptive algorithms, but one may define adaptive algorithms in the rejection sampling model and study the power of adaptivity in this setting as well.

\ignore{Next, we define the \emph{adaptive} variant of the above. Let $\text{Alg}'$ be an adaptive algorithm making $q$ queries to the rejection sampling oracle. Then we view $\text{Alg}'$ as a (possibly probabilistic) mapping from the query-answer history $\langle (\bL_1,\bv_1),\ldots,(\bL_j,\bv_j) \rangle$ to $\bL_{j+1}$, for every $j<q$, and to $\{ \text{accept}, \text{reject} \}$ when $j=q$. Similarly to the above, we say that an adaptive algorithm $\text{Alg}'$ succeeds if:
	
	\[ \Prx_{\substack{\bG \sim \calG_1}} \left[\text{Alg}' \text{ output ``accept''} \right] - \Prx_{\substack{\bG \sim \calG_2}} \left[ \text{Alg}' \text{ outputs ``accept''}\right] \geq \frac{1}{3}. \]
	The cost of $\text{Alg}'$ is measured in the same way as the non-adaptive case.}

\subsection{The Distributions $\calG_1$ and $\calG_2$}

Let $\calG_1$ and $\calG_2$ be two distributions supported on graphs with vertex set $[n]$ defined as follows. Let $\bA \subset [n]$ be a uniform random subset of size $\frac{n}{2}$. 
\begin{align*}
\calG_1 &= \left\{ K_{\bA} \cup K_{\obA} : \bA \subset [n] \text{ random subset size $\frac{n}{2}$} \right\} \\
\calG_2&= \left\{ K_{\bA, \obA} : \bA \subset [n] \text{ random subset size $\frac{n}{2}$} \right\}, 
\end{align*}
where for a subset $A$, $K_{A}$ is the complete graph on vertices in $A$ and $K_{A, \oA}$ is the complete bipartite graph whose sides are $A$ and $\oA$.

\section{Tolerant Junta Testing}\label{sec:junta}
In this section, we will prove that distinguishing the two distributions $\calG_1$ and $\calG_2$ using a rejection sampling oracle reduces to distinguishing two distributions $\Dyes$ and $\Dno$ over Boolean functions, where $\Dyes$ is supported on functions that are close to $k$-juntas and $\Dno$ is supported on functions that are far from any $k$-junta with high probability. 

\subsection{High Level Overview}

We start by providing some intuition of how our constructions and reduction implement the plan set forth in Subsection~\ref{sec:high-level} for the property of being a $k$-junta. We define two distributions supported on Boolean functions, $\Dyes$ and $\Dno$, so that functions in $\Dyes$ are $\eps_0$-close to being $k$-juntas and functions in $\Dno$ are $\eps_1$-far from being $k$-juntas (where $\eps_0$ and $\eps_1$ are appropriately defined constants and $k = \frac{3n}{4}$). 

As mentioned in the introduction, our distributions are based on the indexing function used in \cite{CSTWX17}. We draw a uniform random subset $\bM \subset [n]$ of size $n/2$ and our function $\bGamma = \Gamma_{\bM} \colon \{0, 1\}^n \to [2^{n/2}]$ projects the points onto the variables in $\bM$. Thus, it remains to define the sequence of functions $\bH = (\bh_i \colon \{0, 1\}^n \to \{0,1\} : i \in [2^{n/2}])$.

We will sample a graph $\bG \sim \calG_1$ (in the case of $\Dyes$), and a graph $\bG \sim \calG_2$ (in the case of $\Dno$) supported on vertices in $\obM$. Each function $\bh_i \colon \{0,1\}^n \to \{0, 1\}$ is given by first sampling an edge $(\bj_1, \bj_2) \sim \bG$ and letting $\bh_i$ be a parity (or a negated parity) of the variables $x_{\bj_1}$ and $x_{\bj_2}$. Thus, a function $\boldf$ from $\Dyes$ or $\Dno$ will have all variables being relevant, however, we will see that functions in $\Dyes$ have a group of $\frac{n}{4}$ variables which can be eliminated efficiently\footnote{We say that a variable is eliminated if we change the function to remove the dependence of the variable.}.

We think of the sub-functions $\bh_i$ defined with respect to edges from $\bG$ as implementing a sort of \emph{gadget}: the gadget defined with respect to an edge $(j_1, j_2)$ will have the property that if $\boldf$ eliminates the variable $j_1$, it will be ``encouraged'' to eliminate variable $j_2$ as well. In fact, each time an edge $(\bj_1, \bj_2) \sim \bG$ is used to define a sub-function $\bh_i$, any $k$-junta $g \colon \{0, 1\}^n \to \{0, 1\}$ where variable $\bj_1$ or $\bj_2$ is irrelevant will have to change half of the corresponding part indexed by $\bGamma$. Intuitively, a function $\boldf \sim \Dyes$ or $\Dno$ (which originally depends on all $n$ variables) wants to eliminate its dependence of $n-k$ variables in order to become a $k$-junta. When $\boldf$ picks a variable $j\in \obM$ to eliminate (since variables in $\bM$ are too expensive), it must change points in parts where the edge sampled is incident on $j$. The key observation is that when $\boldf$ needs to eliminate multiple variables, if $\boldf$ picks the variables $j_1$ and $j_2$ to eliminate, whenever a part samples the edge $(j_1, j_2)$, the function changes the points in one part and eliminates two variables. Thus, $\boldf$ eliminates two variables by changing the same number of points when there are edges between $j_1$ and $j_2$. 

At a high level, the gadgets encourage the function $\boldf$ to remove the dependence of variables within a group of edges, i.e., the closest $k$-junta will correspond to a function $g$ which eliminates groups of variables with edges within each other and few outgoing edges. More specifically, if we wants to eliminate $\frac{n}{4}$ variables from $\boldf$, we must find a bisection of the graph $\bG$ whose cut value is small; in the case of $\calG_1$, one of the cliques will have cut value 0, whereas any bisection of a graph from $\calG_2$ will have a high cut value, which makes functions in $\Dyes$ closer to $\frac{3n}{4}$-juntas than functions in $\Dno$.

The reduction from rejection sampling is straight-forward. We consider all queries which are indexed to the same part, and if two queries indexed to the same part differ on a variable $j$, then we the algorithm ``explores'' direction $j$. Each part $i \in [2^{n/2}]$ where some query falls in has a corresponding rejection sampling query $L_i$, which queries the variables explored by the Boolean function testing algorithm.

\subsection{The Distributions $\Dyes$ and $\Dno$}\label{sec:distributions}
The goal of this subsection is to define the two distributions $\Dyes$ and $\Dno$, supported over Boolean functions with $n$ variables. Functions $f \in \Dyes$ will be \emph{close} to being a $k$-junta (for $k = \frac{3n}{4}$) with high probability, and functions $f \sim \Dno$ will be \emph{far} from any $k$-junta with high probability. 

\paragraph{Distribution $\Dyes$}  A function $f$ from $\Dyes$ is generated from a tuple of three random variables, $(\bM, \bA, \bH)$, and we set $f = f_{\bM, \bA, \bH}$. The tuple is drawn according to the following randomized procedure:
\begin{enumerate}
	\item Sample a uniformly random subset $\bM\subset [n]$ of size $m\eqdef\frac{n}{2}$. Let $N = 2^m$ and $\Gamma_\bM : \{0,1\}^{n}\to \left[N \right]$ be the function that maps $x\in\{0,1\}^{n}$ to a number encoded by $x|_{\bM} \in [N]$.
	\item Sample $\bA \subset \obM$ of size $\frac{n}{4}$ uniformly at random, and consider the graph $\bG$ defined on vertices $[\obM]$ with $\bG = K_{\bA} \cup K_{\obA}$, i.e., $\bG$ is a uniformly random graph drawn according to $\calG_1$.
	\item Define a sequence of $N$ functions $\bH=\{\bh_i \colon \{0, 1\}^{n} \to \{0, 1\} : i\in\left[N\right]\}$ drawn from a distribution $\calE(\bG)$. For each $i\in\{1,\ldots,N/2\}$, we let $\bh_i(x)=\bigoplus_{\ell\in\bM} x_\ell$.
	
	For each $i\in \{N/2+1,\ldots, N\}$, we will generate $\bh_i$ independently by sampling an edge $(\bj_1, \bj_2) \sim \bG$ uniformly at random, as well as a uniform random bit $\boldr \sim \{0, 1\}$. We let\label{Dyes:pickpair}
	\[ \bh_i(x) = x_{\bj_1} \oplus x_{\bj_2} \oplus \boldr. \]
	\item Using $\bM,\bA$ and $\bH$, define $f_{\bM, \bA, \bH}=\bh_{\Gamma_\bM(x)}(x)$ for each $x\in\{0,1\}^{n}$.
\end{enumerate}

\paragraph{Distribution $\Dno$} A function $f$ drawn from $\Dno$ is also generated by first drawing the tuple $(\bM, \bA, \bH)$ and setting $f = f_{\bM, \bA, \bH}$. Both $\bM$ and $\bA$ are drawn using the same procedure; the only difference is that the graph $\bG = K_{\bA, \obA}$, i.e., $\bG$ is a uniformly random graph drawn according to $\calG_2$. Then $\bH \sim \calE(\bG)$ is sampled from the modified graph $\bG$.

\ignore{For the remainder of the section, we will commonly refer to $\calEy(M, A)$ and $\calEn(M, A)$ as simply $\calEy$ and $\calEn$, respectively, since the values of $M$ and $A$ will be clear from the context.} We let 
\[ k \eqdef \frac{3n}{4} \qquad \eps_0 \eqdef\frac{1}{8} \qquad \eps_1 \eqdef \frac{3}{16}. \]

\newcommand{\oM}{\overline{M}}
\newcommand{\vol}{\mathrm{vol}}

Consider a fixed subset $M \subset [n]$ which satisfies $|M| = \frac{n}{2}$, and a fixed subset $A \subset \oM$ which satisfies $|A| = \frac{n}{4}$. Let $G$ be a graph defined over vertices in $\oM$, and for any subsets $S_1, S_2 \subset \oM$, let
\[ E_G(S_1, S_2) = \left| \left\{ (j_1, j_2) \in G : j_1 \in S_1, j_2 \in S_2 \right\} \right|, \]
be the number of edges between sets $S_1$ and $S_2$. Additionally, we let
\begin{align}
\chi(G) &= \min \left\{ \dfrac{E_G(S, S) + E_G(S, \oS)}{E_G(\oM, \oM)} : S \subset \oM, |S| \geq \frac{n}{4} \right\}  \label{eq:S-min}
\end{align}
be the minimum fraction of edges adjacent to a set $S$ of size at least $\frac{n}{4}$. 
The following lemma relates the distance of a function $\boldf = f_{M, A, \bH}$ where $\bH\sim \calE(G)$ to being a $k$-junta to $\chi(G)$. We then apply this lemma to the graph in $\Dyes$ and $\Dno$ to show that functions in $\Dyes$ are $\eps_0$-close to being $k$-juntas, and functions in $\Dno$ are $\eps_1$-far from being $k$-juntas.

\begin{lemma}
	Let $G$ be any graph defined over vertices in $A$. If $\boldf = f_{M, A, \bH}$, where $\bH \sim \calE(G)$, then
	\[ \frac{1}{4} \cdot \chi(G) - o(1) \leq \dist(\boldf, \text{$k$-}\Junta) \leq \frac{1}{4} \cdot \chi(G) + o(1) \]
	with probability at least $1 - o(1)$.
\end{lemma}

\begin{proof}
	We first show that $\dist(\boldf, \text{$k$-}\Junta) \leq \frac{1}{4} \cdot \chi(G) + o(1)$. Let $S \subset \oM$ with $|S| \geq \frac{n}{4}$ be the subset achieving the minimum in (\ref{eq:S-min}), and consider the indicator random variables $\bX_i$ for $i \in \{ N/2+1, \dots, N\}$ defined as:
	\[ \bX_i = \left\{ \begin{array}{cc} 1 & \bh_i(x) = x_{j_1} \oplus x_{j_2} \oplus r \text{ with }j_1 \in S \text{ or } j_2 \in S \\
	0 & \text{otherwise}	\end{array} \right. ,\]
	and note that the variables $\bX_i$ are independent and equal $1$ with probability $\chi(G)$. Consider the function $\bg \colon \{0, 1\}^{n}\to \{0, 1\}$ is defined as:
	\[ \bg(x) = \left\{ \begin{array}{cc} \bh_{\Gamma_M(x)}(x) & \bX_{\Gamma_M(x)} = 0 \\
	0 		 & \text{otherwise} \end{array}\right. .\]
	Note that the function $\bg$ is a $k$-junta, since $\bg$ only depends on variables in $[n] \setminus S$, and $|S| \geq \frac{n}{4}$. In addition, we have that:
	\[ \dist(\boldf, \text{$k$-}\Junta) \leq \dist(\boldf, \bg) = \frac{1}{2^{n}} \sum_{i =N/2+1}^N \frac{2^{n-m}}{2} \cdot \bX_i = \frac{1}{2 \cdot 2^{m}} \sum_{i=N/2+1}^N \bX_i, \]
	and by a Chernoff bound, we obtain the desired upper bound. 
	
	For the lower bound, let $T \subset [n]$ of size $\frac{n}{4}$. We divide the proof into two cases: 1) $M \cap T \neq \emptyset$, and 2) $M \cap T = \emptyset$.
	
	We handle the first case first, and let $j \in M \cap T$. 
	\begin{itemize}
		\item Suppose $j$ is the highest order bit of $M$, so that $\Gamma_M(x^{(j\to0)}) \in \{1, \dots, N/2\}$ and $\Gamma_{M}(x^{(j\to1)}) \in \{N/2+1, \dots, N \}$. For $y \in \{0, 1\}^{M \setminus \{ j\}}$ and $\alpha \in \{0, 1\}$, let $X_{y, \alpha} = \{ x \in \{0, 1\}^{n} : x_{|M \setminus \{j\}} = y, x_j = \alpha\}$, $X_y = X_{y, 0} \cup X_{y, 1}$. For every $x \in X_{y}$,
		\[ \boldf(x) = \left\{ \begin{array}{cc} \bigoplus_{i \in M} x_i & x_j = 0 \\
		x_{\bj_1} \oplus x_{\bj_2} \oplus \boldr & x_j = 1 \end{array} \right. ,\] 
		for some $\bj_1, \bj_2 \in \oM$ and $\boldr \in \{0, 1\}$. Thus, for at least half of all points in $x \in X_{y, 0}$, $\boldf(x) \neq \boldf(x^{(j)})$. Therefore, for any function $g \colon \{0, 1\}^{n} \to \{0, 1\}$ which does not depend on $j$, for each $x \in X_{y, 0}$ where $\boldf(x) \neq \boldf(x^{(j)})$, either $\boldf(x) \neq g(x)$, or $\boldf(x^{(j)}) \neq g(x^{(j)})$, thus, 
		\[ \dist(\boldf, g) \geq \frac{1}{2^{n}} \sum_{y \in \{0, 1\}^{M \setminus \{j \}}} \frac{1}{2} \cdot |X_{y, 0}| \geq \frac{1}{4}. \]
		\item Suppose $j$ is not the highest order bit of $M$. Then, if $\Gamma_M(x) \in \{1, \dots, N/2\}$, then $\Gamma_{M}(x^{(j)}) \in \{1, \dots, N/2\}$. We note that for each $y \in \{0, 1\}^{M \setminus \{j\}}$ and $x \in X_{y, 0}$ with $\Gamma_{M}(x) \in \{1, \dots, 2^{m-1}\}$, $\boldf(x) \neq \boldf(x^{(i)})$. Thus again, for any $g \colon \{0, 1\}^{n }\to \{0, 1\}$ which does not depend on $j$,  $\dist(\boldf, g) \geq \frac{1}{4}$, since half of all points $x \in \{0, 1\}^{n}$ satisfy $\Gamma_{M}(x) \in \{1, \dots, N/2\}$.
	\end{itemize}
	
	Therefore, we may assume that $T \subset \oM$. Again, consider the indicator random variables $\bX_i$ for $i \in \{N/2+1, \dots, N\}$ given by
	\[ \bX_i = \left\{ \begin{array}{cc} 1 & \bh_i(x) = x_{j_1} \oplus x_{j_2} \oplus r \text{ with $j_1 \in T$ or $j_2 \in T$} \\
	0 & \text{otherwise} \end{array} \right. ,\]
	and by the definition of $\chi(G)$, we have that $\bX_i = 1$ with probability at least $\chi(G)$. Suppose $x \in \{0, 1\}^{n}$ with $\Gamma_M(x) = i$ and $\bX_i = 1$ with $\bh_i(x) = x_{j_1} \oplus x_{j_2} \oplus r$ with $j_1 \in T$, then $\boldf(x) \neq \boldf(x^{(j_1)})$, which means that any function $g \colon \{0, 1\}^{n} \to \{0,1\}$ which does not depend on variables in $T$, either $g(x) \neq \boldf(x)$ or $g(x^{(j_1)}) \neq \boldf(x^{(j_1)})$, thus, for all such functions $g$,
	\[ \dist(\boldf, g) \geq \frac{1}{4 \cdot 2^{m-1}} \sum_{i =N/2 + 1}^{N} \bX_i \geq \frac{1}{4} \cdot \chi(G) - \frac{1}{n}\]
	with probability $1 - \exp\left(-\Omega(\frac{N}{n^2}) \right)$ by a Chernoff bound. Thus, we union bound over at most $2^{n/2}$ possible subsets $T \subset \oM$ with $|T| \geq \frac{n}{4}$ to conclude that $\dist(\boldf, \text{$k$-}\Junta) \geq \frac{1}{4} \cdot \chi(G) - \frac{1}{n}$ with probability $1 - o(1)$. 
\end{proof}

\begin{figure}\label{fig:graph-cuts}
	\centering
	\begin{picture}(250, 150)
	\put(0,10){\includegraphics[width=0.6\linewidth]{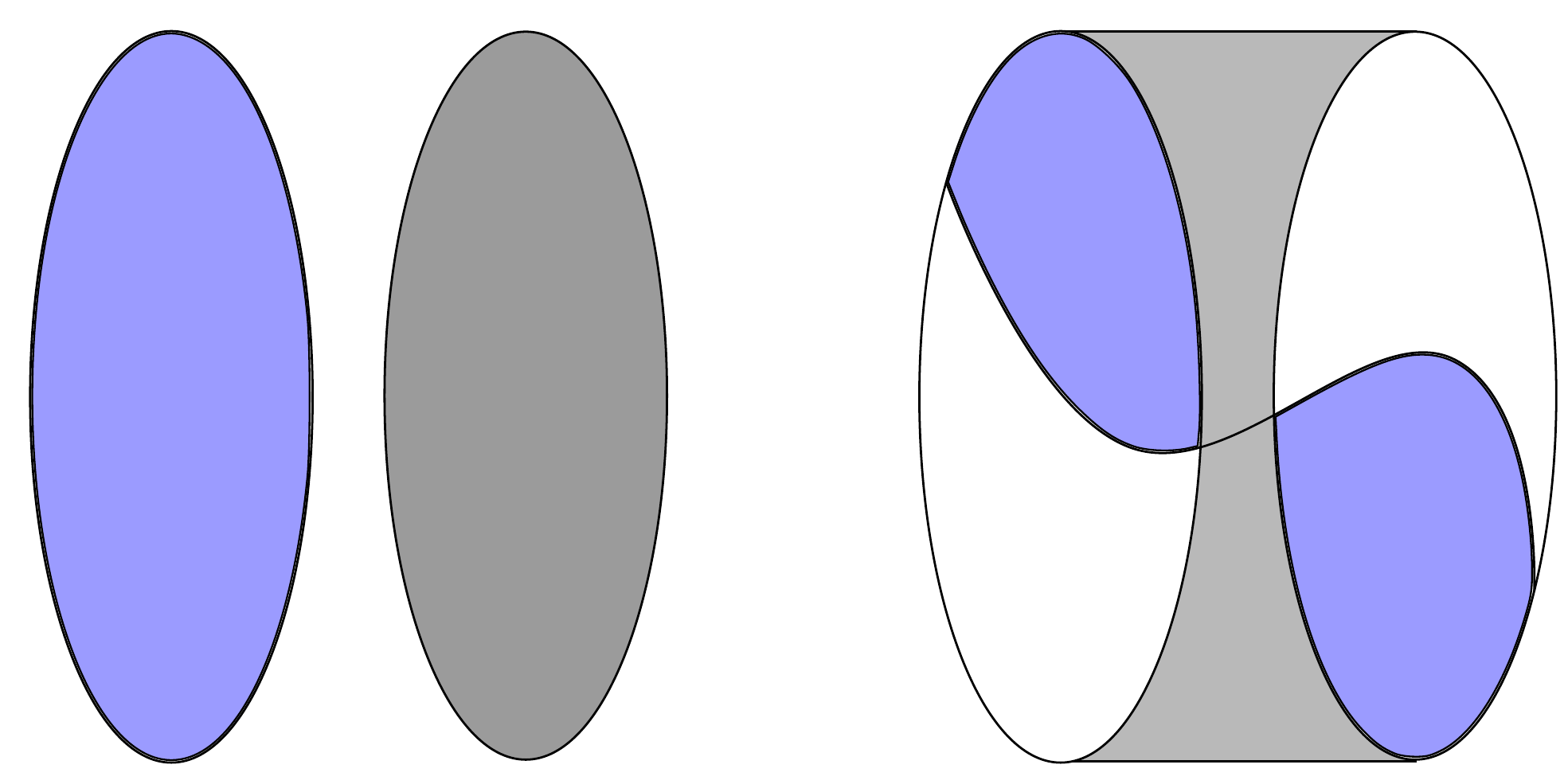}}
	\put(55, 150){$\Dyes$}
	\put(215, 150){$\Dno$}
	\put(25, 0){$A$}
	\put(90, 0){$\oA$}
	\put(185, 0){$A$}
	\put(250, 0){$\oA$}
	\put(195, 100){$\alpha$}
	\put(255, 50){$\beta$}
	\end{picture}
	\caption{Example of graphs $\bG$ from $\Dyes$ and $\Dno$. On the left, the graph $\bG$ is the union of two cliques of size $\frac{n}{4}$, corresponding to $\Dyes$. We note that $\chi(G) = \frac{1}{2}$, since if we let $S = \bA$ (pictured as the blue set), we see that $S$ contains half of the edges. On the right, the graph $\bG$ is the complete bipartite graph with side sizes $\frac{n}{4}$, corresponding to $\Dno$. We note that $\chi(G) = \frac{3}{4}$: consider any set $S \subset \oM$ of size at least $\frac{n}{4}$ pictured in the blue region, and let $\alpha = |S \cap A|$ and $\beta = |S \cap \oA|$, where $\alpha + \beta \geq \frac{n}{4}$, so $E(S, S) + E(S, \oS) \geq (\frac{n}{4})^2 - \alpha \beta \geq (\frac{n}{4})^2(1 - \frac{1}{4})$.}
\end{figure}

\begin{corollary}
	We have that $\boldf \sim \Dyes$ has $\dist(\boldf, \text{$k$-}\Junta) \leq \eps_0 + o(1)$ with probability $1 - o(1)$, and that $\boldf \sim \Dno$ has $\dist(\boldf, \text{$k$-}\Junta) \geq \eps_1 - o(1)$ with probability $1 - o(1)$.
\end{corollary}

\begin{proof}
	For the upper bound in $\Dyes$, when $G = K_{A} \cup K_{\oA}$, we have $\chi(G) \leq \frac{1}{2}$. For the lower bound in $\Dno$, when $G = K_{A, \oA}$, $\chi(G) \geq \frac{3}{4}$ (see Figure~\ref{fig:graph-cuts}). 
	\end{proof}

\subsection{Reducing from Rejection Sampling}\label{subsec:ReductionJuntas}
In this subsection, we will prove that  distinguishing the two distributions $\calG_1$ and $\calG_2$ using rejection sampling oracle reduces to distinguishing the two distributions $\Dyes$ and $\Dno$.

\begin{lemma}\label{lem:juntarejustion} 
	Suppose there exists a deterministic non-adaptive algorithm $\Alg$ making $q$ queries to Boolean functions $f\colon \{0, 1\}^{2n} \to \{0, 1\}$. Then, there exists a deterministic non-adaptive algorithm $\Alg'$ making rejection sampling queries to an $n$-vertex graph such that:
	\begin{align*}
	\Prx_{\boldf \sim \Dyes}[\Alg(\boldf)  \text{ ``accepts''}] &= \Prx_{\bG \sim \calG_1}[\Alg'(\bG)\text{ outputs ``$\calG_1$''}], \qquad\text{and} \\
	\Prx_{\boldf \sim\Dno}[\Alg(\boldf) \text{ ``accepts''}] &= \Prx_{\bG \sim \calG_2}[\Alg'(\bG)\text{ outputs ``$\calG_1$''}].
	\end{align*}
	and has $\cost(\Alg') = O (q  \log n)$ with probability $1 - o(1)$ over the randomness in $\Alg'$.   
\end{lemma}
\begin{proof} 
	Consider an algorithm $\Alg$ making $q$ queries to a Boolean function $\boldf = f_{\bM, \bA, \bH} \colon \{0, 1\}^{2n} \to \{0, 1\}$ (sampled from either $\Dyes$ or $\Dno$).
	First, note that $\bM$ and $\bA$ is distributed in the same way in $\Dyes$ and $\Dno$. Therefore, a rejection sampling algorithm may generate $\bM$ and $\bA$ and utilize its randomness from the rejection sampling oracle to simulate $\bH$. 
	
	Specifically, given the queries $z_1,\ldots,z_1\in\{0,1\}^{2n}$ of $\Alg$, we will partition them into sets $\bQ_1,\ldots,\bQ_t$, such that for all $z,z'\in \bQ_i$, we have that $z|_{\bM}=z'|_{\bM}$. Given the above partition, we define our queries to the rejection sampling oracle $\bL_1,\ldots,\bL_t\subset \obM$ such that for every $i\in[t]$ we let \[ \bL_i\eqdef \{j\in\obM : \exists z,z'\in \bQ_i, (z)_j\neq (z')_j\} 
	\;.\]
	Since $|\obM| = n$, we may associate each element of $\obM$ with an integer in $[n]$ and view the graphs in $\calG_1$ and $\calG_2$ as having vertex set $\obM$. In short, we let $\bL_i$ is the set of indices with two queries in $\bQ_i$ disagreeing in that index.  Next, we claim that the cost of $\Alg'$ is at most $O(q\log n) $ with probability $1-o(1)$. 
	
	Consider the bad event which occurs if there exist two queries $z,z'\in\{0,1\}^{2n}$ such that $z|_{\bM}=z'|_{\bM}$ and  $\|z-z'\|> 100\log(2n)$. Note that for any two queries $z,z'$ such that $\|z-z'\|> 100\log(2n)$, the probability that $z|_{\bM}=z'|_{\bM}$ over the choice of $\bM$ is at most $2^{-100\log(2n)} \ll \frac{1}{q^2}$, and thus we may use a union bound over all pairs of queries to get that the bad event occurs with probability $o(1)$. Therefore, we get that for any $i\in[t]$ and two queries $z,z'\in\bQ_i$ we have that $\|z-z'\|\le 100\log(2n)$ with probability $1-o(1)$, which implies that the cost of $\Alg'$ is $O(q\log n)$ with probability $1-o(1)$.
	
	Now, given the responses to the queries $\bL_1,\ldots,\bL_t\subset[\obM]$, as well as the values of $\bM,\bA$, we will be able to simulate all the randomness in the construction of the two distributions $\Dyes$ and $\Dno$. More formally, $\Alg'$ works in the following way.
	\begin{enumerate}
		
		\item $\Alg'$ makes set queries $\bL_1, \dots, \bL_t$. 
		\item Once $\Alg'$ receives the responses $\bv_1, \dots, \bv_t \in \obM \cup \left(\obM \times \obM \right)\cup \{\emptyset\}$ from the oracle, it will generate a Boolean string $(\boldr_1, \dots, \boldr_q) \in \{0, 1\}^q$ which is distributed exactly as $(f_{\bM, \bA, \bH}(z_1), \dots, f_{\bM, \bA, \bH}(z_q))$, where $f_{\bM, \bA, \bH} \sim \Dyes$ if $\bG \sim \calG_1$ and $f_{\bM, \bA, \bH} \sim \Dno$ if $\bG \sim \calG_2$.
		\item Then if $\Alg(\boldr_1, \dots, \boldr_q)$ outputs ``accept'', then $\Alg'$ should output ``$\calG_1$'', if $\Alg(\boldr_1, \dots, \boldr_q)$ outputs ``reject'', then $\Alg'$ should output ``$\calG_2$''. 
	\end{enumerate}
	Next, we will describe how to generate $(\boldr_1,\ldots,\boldr_q)\in\{0,1\}^q$. We start with setting some notations. 
	For $i \in [t]$, we denote $\bQ_i = \{z^i_1, \dots, z^i_{|\bQ_i|}\}$ and $\boldr_1^i, \dots, \boldr_{|\bQ_i|}^i$.
	
	We aim to show that the random variables $( f_{\bM, \bA, \bH}(z^i_\ell) : \ell \in [|\bQ_i|], i \in [t])$ when $f_{\bM, \bA, \bH} \sim \Dyes$ is distributed exactly the same as $( \boldr_{\ell}^i : \ell \in [|\bQ_i|], i \in [t])$ when $\bG \sim \calG_1$ and $\bv_1, \dots, \bv_t$ are sampled by the oracle (the complement case where $f_{\bM, \bA, \bH} \sim \Dno$ and $\bG\sim \calG_2$ is similar).
	
	We will proceed in $t$ stages, each in stage $i \in [t]$, we will set the values of $\boldr_1^i, \dots, \boldr_{|\bQ_i|}^i$ which will correspond to $f_{\bM, \bA, \bH}(z_1^{i}), \dots, f_{\bM, \bA, \bH}(z_{|\bQ_i|}^i)$. 
	
	If $\bQ_i$ contains strings $z$ such that $\Gamma_\bM(z)\in \{1,\ldots,2^{n-1}\}$ then we let $\boldr_1^{i},\ldots, \boldr_{|\bQ_i|}^i$ be given by $\boldr_{\ell}^i =\bigoplus_{j\in \bM}(z_{\ell}^{i})_j$ for $\ell \in [|\bQ_i|]$.
	Otherwise $\Gamma_\bM(z) \in \{2^{n-1}+1, \dots, 2^n\}$, the algorithm will use the response $\bv_i$ to generate the values $\boldr_1^{i}, \dots, \boldr_{|\bQ_i|}^i$: $\Alg'$ samples a random bit $\boldr^{i} \sim \{0, 1\}$ uniformly and generates $\boldr_1^i, \dots, \boldr_{|\bQ_i|}^i$ according to three cases, corresponding to the three cases $\bv_i$ can be in:
	\begin{itemize}
		\item If $\bv_i = \emptyset$, then $\boldr_1^i = \dots = \boldr_{|\bQ_i|}^i = \boldr^i$. 
		\item If $\bv_i = \{ j \} \subset \obM$, for each $\ell \in [|\bQ_i|]$, $\boldr_\ell^i = \boldr^i$ if $(z_\ell^i)_j = 0$, and $\boldr_\ell^i = 1- \boldr^i$ if $(z_\ell^i)_j = 1$. 
		\item If $\bv_i = \{ j_1, j_2 \} \subset \obM$, for each $\ell \in [|\bQ_i|]$, $\boldr_\ell^i = \boldr^i$ if $(z_\ell^i)_{j_1} \oplus (z_\ell^i)_{j_2} = 0$, and $\boldr_\ell^i = 1- \boldr^i$ if $(z_\ell^i)_{j_1} \oplus (z_\ell^i)_{j_2} = 1$. 
	\end{itemize}
	
	We conclude with the following claim which is immediate from the definition of $\Dyes$, $\Dno$, $\calG_1$ and $\calG_2$, and the corresponding proof simply unravels the definitions of these distributions.
	
	\begin{claim}\label{cl:eq-dist}
		If $\bG \sim \calG_1$, then $(\boldr_1, \dots, \boldr_q)$ is distributed exactly as $(f_{\bM, \bA, \bH}(z_1), \dots, f_{\bM, \bA, \bH}(z_q))$ when $f_{\bM, \bA, \bH} \sim \Dyes$, and if $\bG \sim \calG_2$, then $(\boldr_1, \dots, \boldr_q)$ is distributed exactly as $(f_{\bM, \bA, \bH}(z_1), \dots, f_{\bM, \bA, \bH}(z_q))$ when $f_{\bM, \bA, \bH} \sim \Dno$. 
	\end{claim}
	\begin{proof}
		We give the formal proof for $\Dyes$ and $\calG_1$, as the case with $\Dno$ and $\calG_2$ is the same argumentation. 
		Recall from the definition of $\Dyes$, that $\bM$ and $\bA$ are uniform random sets of size $n$ and $\frac{n}{2}$ respectively.
		Conditioned on $\bM$ and $\bA$, each sub-function $\bh_i$ is picked independently. Thus, we have
		\begin{align*}
		&\Prx_{f_{\bM, \bA, \bH} \sim \Dyes}\left[ \forall i \in [t], \forall \ell \in [|\bQ_i|], f_{\bM, \bA, \bH}(z_{\ell}^i) = y_{\ell}^i\right]\\
		&\qquad\qquad=\binom{2n}{n}^{-1}\binom{n}{n/2}^{-1}\sum_{M\subset [2n]}\sum_{A\subset\oM} \prod_{i=1}^{t}\Prx_{\bh_i}\left[ \forall \ell \in [|Q_i|], \bh_i(z^{i}_\ell) = y^i_\ell \mid \bM=M, \bA =A\right].
		\end{align*}
		We now turn to the graph problem. Recall from the definition of $\bG \sim \calG_1$, that conditioned on $\bM$ and $\bA$, the responses of the oracle, $\bv_1, \dots, \bv_t$ are independent, and $\boldr^1, \dots, \boldr^t$ are independent. Thus, we may write:
		\begin{align*}
		\Prx_{\substack{\bM,\bA, \bv_1, \dots, \bv_t \\ \boldr^1,\dots,\boldr^t}}\left[\forall j \in [q], \forall \ell \in [|Q_i|], \boldr^i_{\ell} = y_{\ell}^i \right] &= \binom{2n}{n}^{-1}\binom{n}{n/2} ^{-1}\sum_{M}\sum_{A} \prod_{i=1}^t \Prx_{\bv_i, \boldr^i}\left[\forall \ell \in [|Q_i|], \boldr^i_\ell = y_{\ell}^i \right].
		\end{align*}
		Therefore, it suffices to show that for any $M\subset[2n]$ of size $n$, $A \subset \oM$ of size $\frac{n}{2}$ and any $i \in [t]$, the random variable $(\bh_i(z_{\ell}^{i}) : \ell \in [|Q_i|])$ with $\bh_i$ from $\Dyes$ with sets $M$ and $A$ is distributed as $(\boldr_1^i, \dots, \boldr_{|Q_i|}^i)$ with oracle response $\bv_i$ and bit $\boldr^i$. 
		
		Let $(\bj_1, \bj_2)$ be a uniform random edge from $K_{A} \cup K_{\overline{A}}$, and we let $\bh_i \colon \{0, 1\}^{2n} \to \{0, 1\}$ be given by:
		\[ \bh_i(x) = \left\{ \begin{array}{cc} x_{\bj_1} \oplus x_{\bj_2} & \text{with probability $\tfrac{1}{2}$} \\
		\neg x_{\bj_1} \oplus x_{\bj_2} & \text{with probability $\tfrac{1}{2}$} \end{array}\right. \]
		Assume that $\bv_i = L_i \cap \{\bj_1,\bj_2\}=\emptyset$. Then $\bh_i(z^i_1) = \dots =\bh_i(z_{|Q_i|}^i)$ is given by a uniform random bit. Similarly, given these values of $\bv_i = \emptyset$, $\boldr^i_1 = \dots = \boldr^i_{|Q_i|}$ is also given by a uniform random bit.
		
		Now, assume that $L_i\cap \{\bj_1,\bj_2\}=\{\bj\}$. Then, for any two queries $z,z'\in Q_i$ such that $(z)_{\bj}\neq (z')_{\bj}$ we must have that $\bh_i(z)\neq\bh_i(z')$, but after this condition is set, the value of any particular $\bh_i(z)$ is a uniform random bit. Likewise, these constraints are set by the procedure generating $\boldr_1^i, \dots, \boldr_{|Q_i|}^i$, and each $\boldr_{\ell}^i$ is a uniform random bit. 
		
		Finally, assume that $L_i\cap \{\bj_1,\bj_2\}=\{\bj_1,\bj_2\}$. Then, for any two queries $z,z'\in Q_i$ such that $(z)_{\bj_1} \oplus (z)_{\bj_2}\neq (z')_{\bj_1} \oplus (z')_{\bj_2}$ we have that $\bh_i(z)\neq\bh_i(z')$, and each value of $\bh_i(z)$ is a uniform random bit. Finally, these constraints are also set forth in the definition of $\boldr^i_1,\ldots,\boldr^i_{|Q_i|}$.
	\end{proof}
	
\end{proof}

Therefore, we conclude with the following corollary.

\begin{corollary}\label{cor:reduct}
	Suppose $\Alg$ is a deterministic non-adaptive algorithm which distinguishes $\Dyes$ and $\Dno$ supported on Boolean functions of $2n$ variables with query complexity $q$, then there exists a non-adaptive algorithm $\Alg'$ for distinguishing between $\calG_1$ and $\calG_2$ supported on graphs with $n$ vertices such that with probability $1-o(1)$ over the randomness of $\Alg'$ it holds that $\cost(\Alg') = O(q \log n)$.
\end{corollary}

\begin{proof}
	We have:
	\begin{align*} 
	&\Prx_{\bG \sim \calG_1}[\Alg'(\bG) \text{ outputs ``$\calG_1$''}] - \Prx_{\bG \sim \calG_2}[\Alg'(\bG) \text{ outputs ``$\calG_1$''}] \\
	&\qquad \qquad = \Prx_{f_{\bM, \bA, \bH}\sim\Dyes}[\Alg(\boldf) \text{ ``accepts''}] - \Prx_{f_{\bM, \bA, \bH}\sim\Dno}[\Alg(\boldf) \text{ ``accepts''}] \geq \frac{1}{3} - o(1). 
	\end{align*}
	We also have that with probability at least $1-o(1)$, for each $i \in [t]$, if $Q_i = \{ z^i_1, \dots, z^i_{|Q_i|} \}$, then $|L_i| \leq \sum_{j=2}^{|Q_i|} \|z^i_1 -z^i_j \|_1 \leq |Q_i| \cdot 100 \log(2n)$. Therefore, $\cost(\Alg') = \sum_{i=1}^t |L_i| = O(q\log n)$ with probability at least $1-o(1)$. 
\end{proof}

\newcommand{\Unate}{\mathrm{Unate}}
\newcommand{\lsim}{\lesssim}
\newcommand{\gsim}{\gtrsim}

\section{Tolerant Unateness Testing}\label{sec:unate}

In this section, we show how to reduce distinguishing distributions $\calG_1$ and $\calG_2$ to distinguishing between Boolean functions which are close to unate and Boolean functions which are far from unate. We start with a high level overview of the constructions and reduction, and then proceed to give formal definitions and the reductions for adaptive and non-adaptive tolerant testing.

\subsection{High Level Overview} We now describe how our constructions and reduction implement the plan set forth in Subsection~\ref{sec:high-level} for the property of unateness. Similarly to Section~\ref{sec:junta}, we define two distributions $\Dyes$ and $\Dno$ supported on Boolean functions, so that functions in $\Dyes$ are $\eps_0$-close to being unate, and functions in $\Dno$ are $\eps_1$-far from being unate (where $\eps_0$ and $\eps_1$ are appropriately defined constants).

We will use a randomized indexing function $\bGamma \colon \{0, 1\}^n \to [N]$ based on the Talagrand-style constructions from \cite{BB16, CWX17} to partition $\{0, 1\}^n$ in a unate fashion. Again, we will then use a graph $\bG \sim \calG_1$ or $\calG_2$ to define the sequence of sub-function $\bH = (\bh_i \colon \{0,1\}^n \to \{0, 1\} : i \in [N])$. The sub-functions $\bh_i$ will be given by a parity (or negated parity) of three variables: two variables will correspond to the end points of an edge sampled $(\bj_1, \bj_2) \sim \calG$, the third variable will be one of two pre-specified variables, which we call $m_1$ and $m_2$. Consider for simplicity the case when $\bh_i(x) = x_{\bj_1} \oplus x_{\bj_2} \oplus x_{m_1}$, and assume that we require that variable $m_1$ is non-decreasing.

Similarly to Section~\ref{sec:junta}, the functions $\bh_i$ are thought of as gadgets. We will have that if $\bh_i$ is defined with respect to an edge $(j_1, j_2)$ and $m_1$, then the function $\boldf$ will be ``encouraged'' to make variables $j_1$ and $j_2$ have opposite directions, i.e., either $j_1$ is non-increasing and $j_2$ is non-decreasing, or $j_1$ is non-decreasing and $j_2$ is non-increasing. In order to see why the three variable parity implements this gadget, we turn our attention to Figure~\ref{fig:crossing-cut-1} and Figure~\ref{fig:not-crossing-cut-1}.

Intuitively, the function $\boldf$ needs to change some of its inputs to be unate, and it must choose whether the variables $j_1$ and $j_2$ will be monotone (non-decreasing) or anti-monotone (non-increasing). Suppose $\boldf$ decides that the variable $j_1$ should be monotone and $j_2$ be anti-monotone, and $m_1$ will always be monotone (since it will be too expensive to make it anti-monotone). Then, when $\bh_i(x) = x_{j_1} \oplus x_{j_2} \oplus x_{m_1}$, $\bh_i$ will have some \emph{violating edges}, i.e., edges in direction $j_1$ which are decreasing, or edges in direction $j_2$ which are increasing, or edges in direction $m_1$ which are decreasing (see Figure~\ref{fig:crossing-cut-1}, where these violating edges are marked in red). In this case, there exists a way that $\boldf$ may change $\frac{1}{4}$-th fraction of the points and remove all violating edges (again, this procedure is shown in Figure~\ref{fig:crossing-cut-1}). 

In contrast, suppose that $\boldf$ decides that the variables $j_1$ and $j_2$ both should be monotone. Then, when $\bh_i(x) = x_{j_1} \oplus x_{j_2} \oplus x_{m_1}$, the violating edges (shown in Figure~\ref{fig:not-crossing-cut-1}) form vertex-disjoint cycles of length $6$ in $\{0, 1\}^n$, thus, the function $\boldf$ will have to change $\frac{3}{8}$-th fraction of the points in order to remove all violating edges. In other words, when there is an edge $(j_1, j_2)$ sampled in $\bh_i$, the function $\boldf$ is ``encouraged'' to make $j_1$ and $j_2$ have opposite directions, and ``discouraged'' to make $j_1$ and $j_2$ have the same direction. The other cases are presented in Figures~\ref{fig:crossing-cut-2}, ~\ref{fig:crossing-cut-2-2}, and ~\ref{fig:not-crossing-cut-2}.

In order for $\boldf$ to become unate, it must first choose whether each variable will be monotone or anti-monotone. $\boldf$ will choose all variables in $\bM$ to be monotone, the variable $m_1$ to be monotone, and $m_2$ to be anti-monotone, but will have to make a choice for each variable in $\obM$, corresponding to each vertex of the graph $\bG$. As discussed above, for each edge $(j_1, j_2)$ in the graph, $\boldf$ is encouraged to make these orientations opposite from each other, so $\boldf$ will want to look for the maximum cut on the graph, whose value will be different in $\calG_1$ and $\calG_2$. 

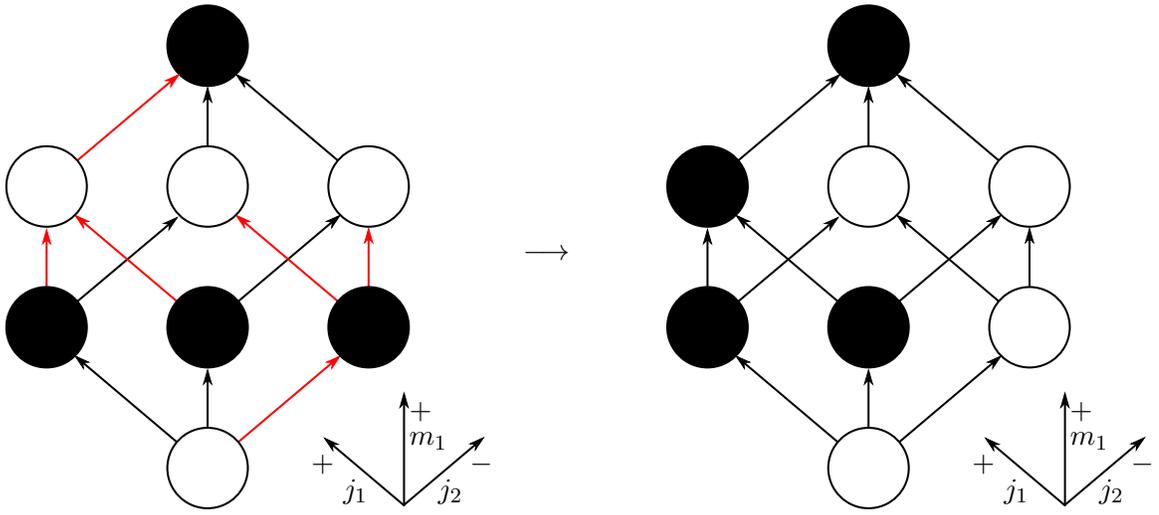
\begin{figure}\label{fig:crossing-cut-1}
	\centering
	\begin{picture}(400, 160)
	\put(0,0){\includegraphics[width=0.4\linewidth]{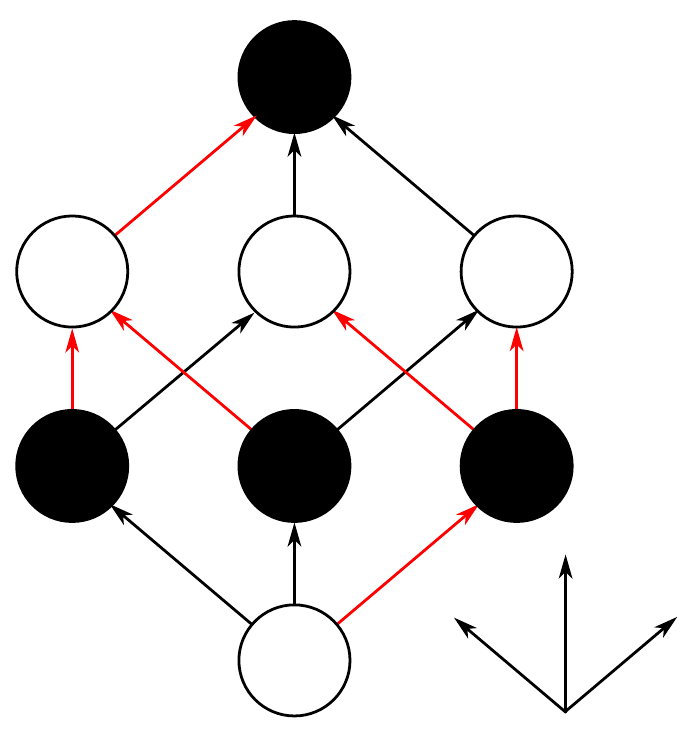}}
	\put(132, 10){$j_1$}
	\put(120, 20){$+$}
	\put(168, 10){$j_2$}
	\put(180, 20){$-$}
	\put(157, 30){$m_1$}
	\put(157, 40){$+$}
	\put(200, 100){$\longrightarrow$}
	\put(250, 0){\includegraphics[width=0.4\linewidth]{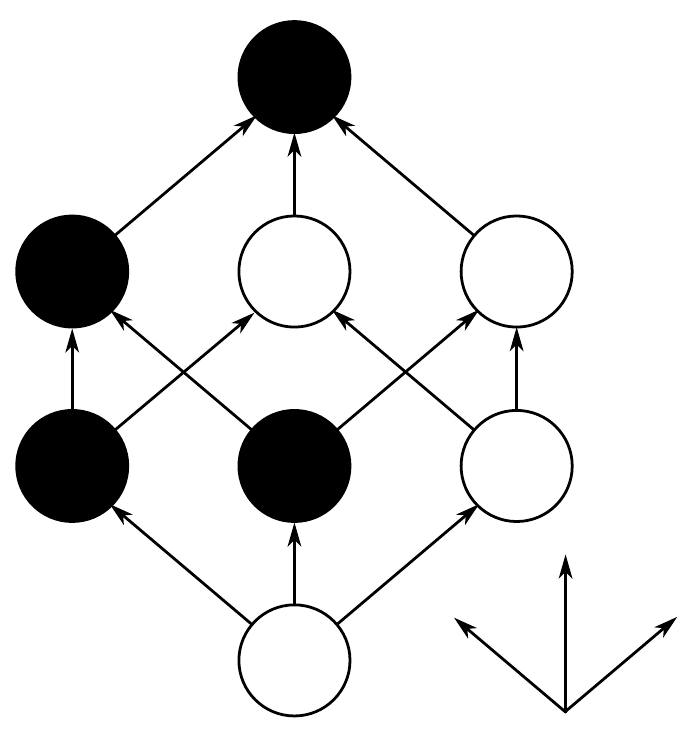}}
	\put(382, 10){$j_1$}
	\put(370, 20){$+$}
	\put(418, 10){$j_2$}
	\put(430, 20){$-$}
	\put(407, 30){$m_1$}
	\put(407, 40){$+$}
	\end{picture}
	\caption{Example of a function $\bh_i \colon \{0, 1\}^n \to \{0, 1\}$ with $\bh_i(x) = x_{j_1} \oplus x_{j_2} \oplus x_{m_1}$ with variable $j_1$ (which ought to be monotone), $j_2$ (which ought to be anti-monotone), and $m_1$ (which is always monotone). The image on the left-hand side represents $\bh_i$, and the red edges correspond to violating edges for variables $j_1, j_2$ and $m_1$. In other words, the red edges correspond to anti-monotone edges in variables $j_1$, monotone edges in variables $j_2$, and anti-monotone edges in direction $m_1$. On the right-hand side, we show how such a function can being ``fixed'' into a function $\bh_i' \colon \{0, 1\}^n \to \{0, 1\}$ by changing $\frac{1}{4}$-fraction of the points.}
\end{figure}

\begin{figure}\label{fig:not-crossing-cut-1}
	\centering
	\begin{picture}(400, 160)
	\put(0,0){\includegraphics[width=0.4\linewidth]{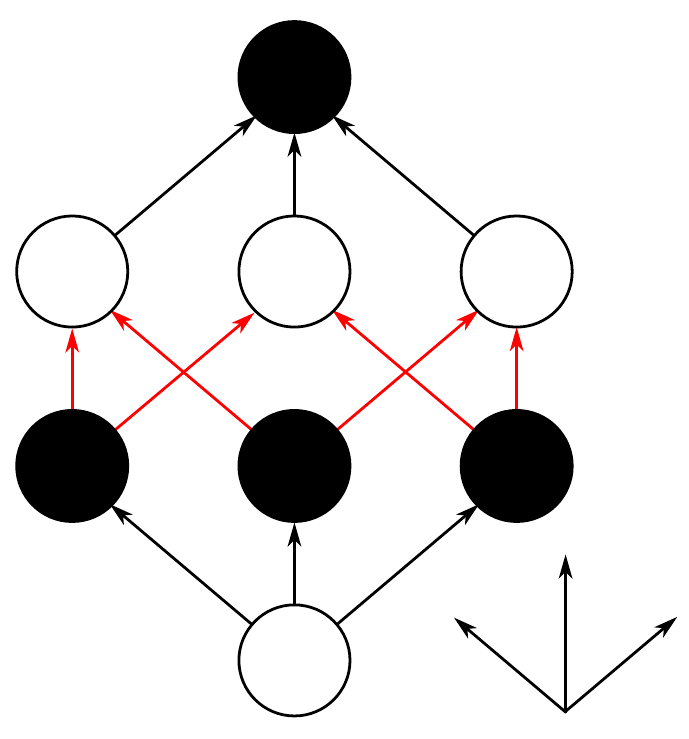}}
	\put(132, 10){$j_1$}
	\put(120, 20){$+$}
	\put(168, 10){$j_2$}
	\put(180, 20){$+$}
	\put(157, 30){$m_1$}
	\put(157, 40){$+$}
	\put(200, 100){$\longrightarrow$}
	\put(250, 0){\includegraphics[width=0.4\linewidth]{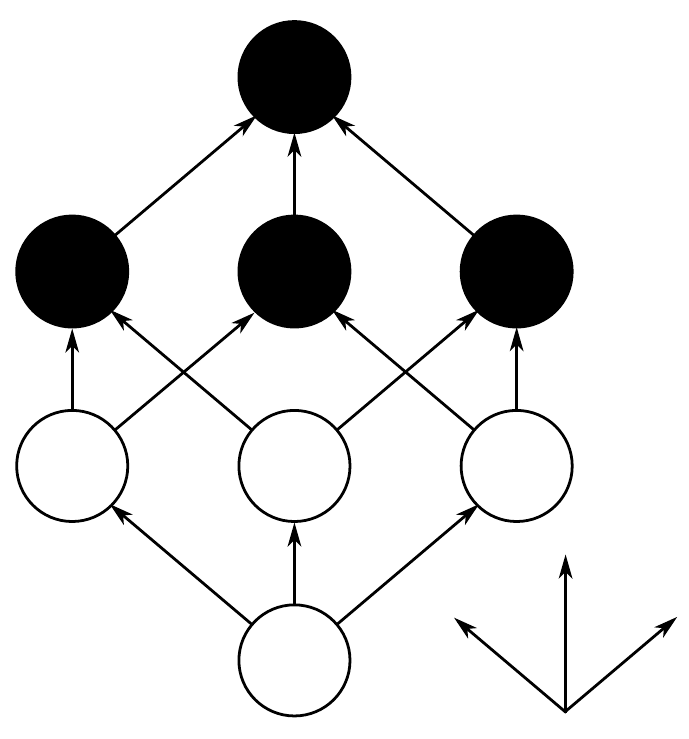}}
	\put(382, 10){$j_1$}
	\put(370, 20){$+$}
	\put(418, 10){$j_2$}
	\put(430, 20){$+$}
	\put(407, 30){$m_1$}
	\put(407, 40){$+$}
	\end{picture}
	\caption{Example of a function $\bh_i \colon \{0, 1\}^n \to \{0, 1\}$ with $\bh_i(x) = x_{j_1} \oplus x_{j_2} \oplus x_{m_1}$ with variables $j_1$ and $j_2$  (which ought to be monotone), and $m_1$ (which ought to be monotone). On the left side, we indicate the violating edges with red arrows, and note that the functions in the left and right differ by $\frac{3}{8}$-fraction of the points. We also note that any function $\bh_i' \colon \{0, 1\}^n \to \{0, 1\}$ which has $j_1$, $j_2$ and $m_1$ monotone must differ from $\bh_i$ on at least $\frac{3}{8}$-fraction of the points because the violating edges of $\bh_i$ form a cycle of length $6$.}
\end{figure}
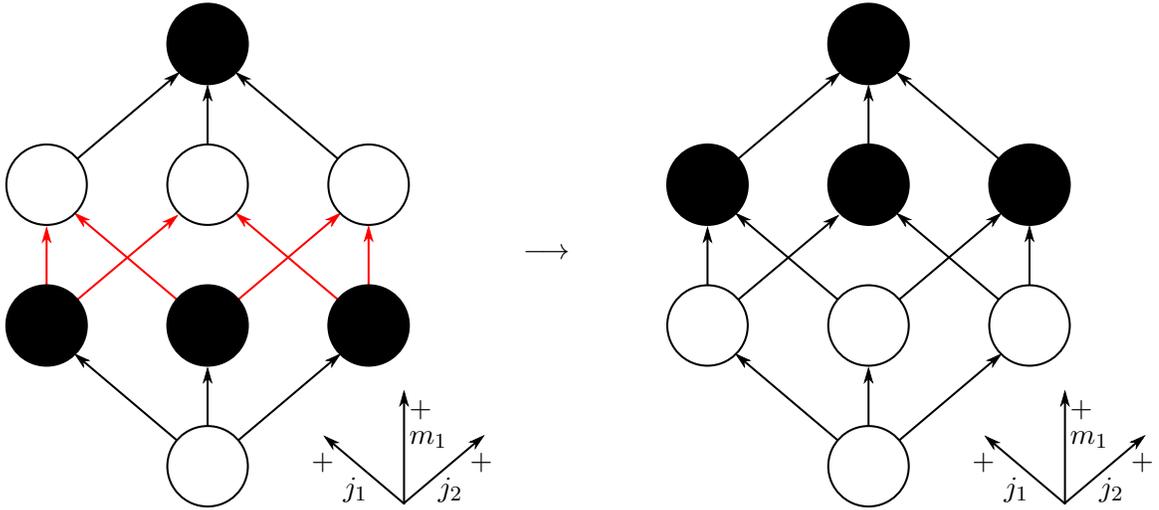

Similarly to the case in Section~\ref{sec:junta}, the reduction will follow by defining the rejection sampling queries $L_i$ corresponding to variables explored in sub-function $\bh_i$. The unate indexing functions $\bGamma$ are not as strong as the indexing functions from the Section~\ref{sec:junta}, so for each query in the Boolean function testing algorithm, our reduction will lose some cost in the rejection sampling algorithm. In particular, the adaptive reduction loses $n$ cost for each Boolean function query, since adaptive algorithms can efficiently explore variables with a binary search; this gives the $\widetilde{\Omega}(n)$ lower bound for tolerant unateness testing. The non-adaptive reduction loses $O(\sqrt{n} \log n)$ cost for each Boolean function query since queries falling in the same part may be $\Omega(\sqrt{n})$ away from each other (the same scenario occurs in the non-adaptive monotonicity lower bound of \cite{CWX17}). The non-adaptive reduction is more complicated than the adaptive reduction since it is not exactly a black-box reduction (we require a lemma from Section~\ref{sec:LowerBound}). This gives the $\widetilde{\Omega}(n^{3/2})$ lower bound for non-adaptive tolerant unateness testing. 

\subsection{The Distributions $\Dyes$ and $\Dno$} We now turn to describing a pair of distributions $\Dyes$ and $\Dno$ supported on Boolean functions $f \colon \{0, 1\}^n \to \{0, 1\}$. These distributions will have the property that for some constants $\eps_0$ and $\eps_1$ with $0 < \eps_0 < \eps_1$, 
\[ \Prx_{\boldf \sim \Dyes}[\dist(\boldf, \Unate) \leq \eps_0] = 1 - o(1) \qquad\text{and}\qquad \Prx_{\boldf \sim \Dno}[\dist(\boldf, \Unate) \geq \eps_1] = 1 - o(1). \]

We first define a function $\boldf \sim \Dno$, where we fix the parameter:
\[ N = 2^{\sqrt{n}}.\]
\begin{enumerate}
	\item Sample some set $\bM \subset [n]$ of size $|\bM| = \frac{n}{2}$ uniformly at random and let $\bm_1, \bm_2 \sim \bM$ be two distinct indices.
	\item We let $\bT \sim \calE(\bM \setminus \{ \bm_1, \bm_2\})$ (which we describe next). $\bT$ is a sequence of terms $(\bT_i : i \in [N])$ which is used to defined a multiplexer map $\Gamma_{\bT} \colon \{0, 1\}^n \to [N] \cup \{ 0^*, 1^*\}$. 
	\item We sample $\bA \subset \obM$ of size $|\bA| =\frac{n}{2}$ and define a graph as:
	\[ \bG = K_{\bA} \cup K_{\obA}. \]
	\item We now define the distribution over sub-functions $\bH = (\bh_i : i \in [N]) \sim \calH(\bm_1, \bm_2, \bG)$. For each function $\bh_i \colon \{0, 1\}^{n} \to \{0, 1\}$, we generate $\bh_i$ independently:
	\begin{itemize}
		\item When $i \leq 3N / 4$, we sample $\bj \sim \{ \bm_1, \bm_2\}$ and we let:
		\[ \bh_i(x) = \left\{ \begin{array}{cc} x_{\bj} & \bj = \bm_1 \\
		\neg x_{\bj} & \bj = \bm_2 \end{array} \right. . \]
		\item Otherwise, if $i > 3N / 4$, we sample an edge $(\bj_1, \bj_2) \sim \bG$ and an index $\bj_3 \sim \{ \bm_1, \bm_2\}$ we let:
		\[ \bh_i(x) =\left\{\begin{array}{cc} x_{\bj_1} \oplus x_{\bj_2} \oplus x_{\bj_3} & \bj_3 = \bm_1 \\
		\neg x_{\bj_1} \oplus x_{\bj_2} \oplus x_{\bj_3} & \bj_3 = \bm_2 \end{array} \right. .\]
	\end{itemize}
\end{enumerate}
The function $\boldf \colon \{0, 1\}^n \to \{0, 1\}$ is given by $\boldf(x) = f_{\bT, \bA, \bH}(x)$ where:
\begin{align} f_{\bT, \bA, \bH}(x) &=\left\{ \begin{array}{ll} 1 & |x_{|\bM}| > \frac{n}{4} + \sqrt{n} \\
0 & |x_{|\bM}| < \frac{n}{4} - \sqrt{n} \\
1 & \Gamma_{\bT}(x) = 1^* \\
0 & \Gamma_{\bT}(x) = 0^* \\
\bh_{\Gamma_{\bT}(x)}(x)& \textrm{otherwise} \end{array} \right. . \label{eq:unate-dist} 
\end{align}
We now turn to define the distribution $\calE(M)$ supported on terms $\bT$, as well as the multiplexer map $\Gamma_{\bT} \colon \{0, 1\}^n \to [N]$. As mentioned above, $\bT \sim \calE(M)$ will be a sequence of $N$ terms $(\bT_i : i \in [N])$, where each $\bT_i$ is given by a DNF term:
\[ \bT_i(x) = \bigwedge_{j \in \bT_i} x_j, \]
where the set $\bT_i \subset M$ is a uniformly random $\sqrt{n}$-element subset. Given the sequence of terms $\bT$, we let:
\[ \Gamma_{\bT}(x) = \left\{ \begin{array}{ll} 0^* & \forall i \in [N], \bT_i(x) = 0 \\
1^* & \exists i_1 \neq i_2 \in [N], \bT_{i_1}(x) = \bT_{i_2}(x) = 1 \\
i     & \bT_i(x) = 1 \text{ for a unique $i \in [N]$} \end{array} \right. .\]
It remains to define the distribution $\Dyes$ supported on Boolean functions. The function $\boldf \sim \Dyes$ will be defined almost exactly the same. We still have $\boldf = f_{\bT, \bA, \bH}$ as defined above, however, the graph $\bG$ will be different. In particular, we will let:
\[ \bG = K_{\bA, \obA}. \]

Fix any set $M \subset [n]$ of size $\frac{n}{2}$ and let $m_1, m_2 \in M$ be two distinct indices and $M' = M \setminus \{m_1, m_2\}$. For any $\bT \sim \calE(M')$, let $\bX\subset\{0,1\}^n$ be the subset of points indexed to some subfunction $\bh_i$:
\[ \bX\eqdef\left\{ x \in \{0, 1\}^n : |x_{|M}|\in [n/4-\sqrt{n},n/4 +\sqrt n]\;\;\text{and  }\Gamma_T(x) \in [N] \right\},\]
and define $\gamma \in (0, 1)$ be the parameter:
\[ \gamma \eqdef \Ex_{\bT \sim \calE(M')}\left[ \dfrac{|\bX|}{2^n} \right]. \]

\begin{claim}\label{cl:size-X}
	With probability at least $1-\exp\left(-\Omega(N/n^2)\right)$ over the draw $\bT \sim \calE(M)$ the set $\bX$ has size $|\bX| = 2^n \gamma (1 \pm \frac{1}{n})$, where $\gamma = \Omega(1)$.
\end{claim}

\begin{proof} 
	Note that:
	\[ \Ex_{\bT\sim\calE(M)}\left[|\bX|\right]=\sum_{\stackrel{x\in\{0,1\}^n :}{n/4-\sqrt{n}\le|x_{|M}|\le n/4+\sqrt{n}}}\Prx_{\bT\sim\calE(M)}[x\in \bX]\;.\]
	Fix $x\in\{0,1\}^n$ such that $n/4-\sqrt{n}\le|x_{|M}|\le n/4+\sqrt{n}$. We can view the probability on the right hand side as a sequence of $N$ disjoint events. Every event $j\in[N]$ correspond to the case where $x$ satisfies the unique term $\bT_j$. The probability of each such event is: 
	\begin{align*} 
	\Prx_{\bT \sim \calE(M)}[\Gamma_{\bT}(x) = i] &\geq \left(\frac{1}{(n/2-2)^{\sqrt{n}}} \prod_{k=0}^{\sqrt{n}-1} (|x_{|M}| - k - 2) \right) \cdot \left(1-\left(\frac{|x_{|M}|}{n/2-2}\right)^{\sqrt{n}}\right)^{N-1} \\
	& \ge \left(\frac{n/4-2\sqrt{n}}{{n/2}}\right)^{\sqrt{n}}\cdot \left(1-\left(\frac{n/4+\sqrt{n}}{n/2-2}\right)^{\sqrt{n}}\right)^{N-1} = \Omega(1/N).
	\end{align*}
	Therefore, the probability that $x\in \bX$ is at least $\Omega(1)$. Summing up all the $x$ with $|x_{|M}| \approx \frac{n}{4} \pm \sqrt{n}$ gives $\Ex_{\bT\sim\calE(M)}[|\bX|]=\Omega(2^n)$, so $\gamma = \Omega(1)$. In order to show that the random variable $|\bX|$ is concentrated around the mean, let $\Omega$ be the space of all possible $\sqrt{n}$-sized terms with variables in $M \setminus \{m_1, m_2\}$, and let $c \colon \Omega^{N} \to \Z^{\geq 0}$ be the function on the independent terms which computes the size of $\bX$:
	\[ c(\bT_1, \dots, \bT_N) = |\bX|. \]
	For every $j\in[N]$ and $T_1,\ldots,T_N, T_j' \in \Omega$ 
	\[\left| c(T_1,\ldots,T_j',\ldots,T_N)-c(T_1,\ldots,T_j,\ldots,T_N)\right| \le \frac{2^{n}}{N}\;,\] 
	so by McDiarmid's inequality:
	\[ \Prx_{\bT\sim\calE(M)}\left[ \left| |\bX| - \gamma 2^n \right| \geq 2^{n}/n \right]\le \exp\left(-\frac{\Omega(2^{2n}/n^2)}{\sum_{i=1}^N 2^{2n}/N^2}\right)=\exp\left(-\Omega(N/n^2)\right)\;.\]
\end{proof}

In addition, let $X_i \subset X$ be the subset of points $x \in X$ with $\Gamma_{T}(x) = i$, and note that the subsets $X_1, \dots, X_{N}$ partition $X$, where each $|X_i| \leq 2^{n-\sqrt{n}}$. With a similar argument as Claim~\ref{cl:size-X}, we conclude that with probability $1 - o(1)$ over the draw of $\bT \sim \calE(M)$, we have:
\begin{align} \sum_{i=1}^{3N/4} |\bX_i| = 2^n\cdot\frac{3\gamma}{4} \left( 1\pm \frac{1}{n}\right) \qquad \text{and}\qquad \sum_{i=3N/4+1}^N |\bX_i| = 2^n \cdot \frac{\gamma}{4} \left( 1 \pm \frac{1}{n}\right).  \label{eq:size-X-i}
\end{align}
Thus, we only consider functions $\boldf \sim \Dyes$ (or $\sim \Dno$) where the sets $M$, and $T$ satisfy (\ref{eq:size-X-i}).

We consider any set $A \subset \oM$ of size $\frac{n}{4}$. Now, consider any graph $G$ defined over vertices in $\oM$, and we let:
\[ \chi(G) = \min \left\{ \dfrac{E_G(S, S) + E_{G}(\oS, \oS)}{E_{G}(\oM, \oM)} : S \subset \oM \right\}. \]
In other words, we note that $\chi(G)$ is one minus the fractional value of the maximum cut, and the value of $\chi(G)$ is minimized for the set $S$ achieving the maximum cut of $G$. The following lemma relates the distance to unateness of a function $\boldf = f_{T, A, \bH}$ with $\bH \sim \calH(m_1, m_2, G)$, where $G$ is an underlying graph defined on vertices in $\oM$.

\begin{lemma}
	Let $G$ be any graph defined over vertices in $\oM$. If $\boldf = f_{T, A, \bH}$ where $\bH \sim \calH(m_1, m_2, G)$, then
	\[ \frac{\gamma}{16} \left(1 + \frac{1}{2}\cdot \chi(G)\right) - o(1) \leq \dist(\boldf, \Unate) \leq \frac{\gamma}{16} \left( 1+ \frac{1}{2} \cdot \chi(G) \right) + o(1). \]
	with probability $1 - o(1)$.
\end{lemma}

\begin{proof}
	We first show that $\dist(\boldf, \Unate) \leq \frac{\gamma}{16} \left( 1 + \frac{1}{2} \cdot \chi(G)\right) + o(1)$ with high probability. Consider the set $S \subset \oM$ which achieves the minimum of $\chi(G)$, i.e.,
	\[ \chi(G) = \dfrac{E(S, S) + E(\oS, \oS)}{E(\oM, \oM)}, \]
	and let $\bg \colon \{0, 1\}^n \to \{0, 1\}$ be the unate function which makes variables in $M$ monotone, $m_1$ monotone, $m_2$ anti-monotone, $S$ monotone, and $\oM \setminus S$ anti-monotone. We defined $\bg$ as follows:
	\begin{align*}
	\bg(x) &= \left\{ \begin{array}{cc} 1 & |x_{|M}| > \frac{n}{4} + \sqrt{n} \\
	0 & |x_{|M}| < \frac{n}{4} - \sqrt{n} \\
	1 & \Gamma_{T}(x) = 1^* \\
	0 & \Gamma_{T}(x) = 0^* \\
	\bh_{\Gamma_{T}(x)}'(x) & \text{otherwise} \end{array} \right. ,
	\end{align*}
	where we define $\bh_{i}' \colon \{0, 1\}^n \to \{0, 1\}$ as a Boolean function which depends on $\bh_i$. In particular, if $i \leq 3N/4$, we let $\bh_i' = \bh_i$. Otherwise, suppose $\bh_i$ is defined with respect to $(j_1, j_2, j_3)$. There are two cases:
	\begin{itemize}
		\item (Directions of $j_1$ and $j_2$ disagree) If $j_1 \in S$ and $j_2 \notin S$, or $j_1 \notin S$ and $j_2 \in S$, then we let $\bh_i'$ be the function on variables $x_{j_1}, x_{j_2}$ and $x_{j_3}$ with $\dist(\bh_i, \bh_i') = \frac{1}{4}$ (see Figure~\ref{fig:crossing-cut-1} for an example with $j_3 = m_1$ which needs to be monotone, $j_1 \in S$ and $j_2 \in \oS$; Figure~\ref{fig:crossing-cut-2} and Figure~\ref{fig:crossing-cut-2-2} give the symmetric constructions when $j_1$ and $j_2$ are flipped, and when variable $m_2$ is used instead of $m_1$, respectively). 
		\item (Directions of $j_1$ and $j_2$ agree) If $j_1 \in S$ and $j_2 \in S$, or $j_1 \notin S$ and $j_2 \notin S$, then we let $\bh_i'$ be the function on variables $x_{j_1}, x_{j_2}$ and $x_{j_3}$ with $\dist(\bh_i, \bh_i') = \frac{3}{8}$ (see Figure~\ref{fig:not-crossing-cut-1} for an example with $j_3 = m_1$ which needs to be monotone, $j_1 \in S$ and $j_2 \in S$; Figure~\ref{fig:not-crossing-cut-2} gives the violating edges of the symmetric examples when variable $m_2$ is used, and either both $j_1$ and $j_2$ are monotone, or both anti-monotone).
	\end{itemize}
	Therefore, we define the indicator random variable $\bC_i$ for each $i \in \{ 3N/4 + 1, \dots, N \}$ by 
	\[ \bC_i = \left\{ \begin{array}{cc} 1 & (\bj_1, \bj_2)\text{ from $\bh_i$ is not cut by $S$} \\
	0 & \text{ otherwise } \end{array} \right.,  \]
	and we note that all $\bC_i$ are independent and $\Prx_{\bH}[\bC_i] = \chi(G)$. By the two cases displayed above, we have that:
	\[ \dist(\boldf, \bg) = \dfrac{1}{2^n} \sum_{i=3N/4+1}^{N} |X_i| \left(\frac{1}{4} + \bC_i \cdot \frac{1}{8} \right) \leq \frac{\gamma}{16} \left( 1 + \frac{1}{2} \cdot \chi(G)\right) + o(1/n), \]
	with probability at least $1 - \exp\left(-\Omega(N/n^2) \right)$ over the draw of all $\bC_i$.
	
	For the lower bound, consider any function $g \colon \{0, 1\}^n \to \{0, 1\}$ which is unate. Suppose variable $x_{m_1}$ is anti-monotone in $g$, then let $\bC_i$ for $i \in [3N/4]$ be the indicator random variable
	\[ \bC_i = \left\{ \begin{array}{cc} 1 & \bh_i(x) = x_{m_1} \\
	0 & \bh_i(x) = \neg x_{m_2} \end{array} \right. .\]
	We note that if $\bC_i = 1$, then $\boldf$ and $g$ differ on at least $|X_i|/2$ from $X_i$. Thus, we have $\dist(\boldf, g) \geq \frac{3\gamma}{8}\left(1 - \frac{1}{n}\right) - o(1)$ with high probability over the draw of $\bC_i$. Likewise, we may say that if $x_{m_2}$ is monotone, then $\dist(\boldf, g) \geq \frac{3\gamma}{8} \left(1 - \frac{1}{n} \right) - o(1)$. Thus, we may consider functions $g \colon \{0, 1\}^n \to \{0, 1\}$ with $x_{m_1}$ being monotone and $x_{m_2}$ being anti-monotone. In this case, consider a set $S \subset \oM$, then if $g$ is any unate function with variables in $S$ being monotone and variables in $\oM \setminus S$ being anti-monotone, then we note that for each $i \in \{ 3N/4+1, \dots, N\}$, if $\bh_i$ sampled an edge $(j_1, j_2)$ which is cut by $S$, then $X_i$ must differ on $\frac{1}{4}$th of the points in $X_i$ (see Figure~\ref{fig:crossing-cut-1} for an example of the violating edges if $j_1$ and $j_2$ are oriented in opposite directions). On the other hand, if $(j_1, j_2)$ is not cut by $S$, then $X_i$ must differ on $\frac{3}{8}$ths of the points in $X_i$ (see Figure~\ref{fig:not-crossing-cut-1} to see how the violating edges require $\frac{3}{8}$ths of the points being different). Thus, if we let the indicator random variable $\bC_i$ be 
	\[ \bC_i = \left\{ \begin{array}{cc} 1 & (\bj_1, \bj_2) \text{ from $\bh_i$ is not cut by $S$} \\
	0 & \text{ otherwise} \end{array} \right. ,\]
	we may write:
	\[ \dist(\boldf, g) \geq \frac{1}{2^n} \sum_{i=3N/4+1}^{N} |X_i| \left( \frac{1}{4} + \frac{1}{8}\cdot \bC_i \right)\geq \frac{\gamma}{16} \left( 1 + \frac{1}{2} \cdot \chi(G)\right) + O(1/n), \]
	with probability $1 - \exp\left( -\Omega(N/n^2)\right)$ over the draw of $\bC_i$, since $\Prx[\bC_i = 1] \geq \chi(G)$. Thus, we may union bound over all $2^{n/2}$ subsets $S \subset \oM$ to conclude the claim.
\end{proof}

\begin{figure}\label{fig:crossing-cut-2}
	\centering
	\begin{picture}(400, 160)
	\put(0,0){\includegraphics[width=0.4\linewidth]{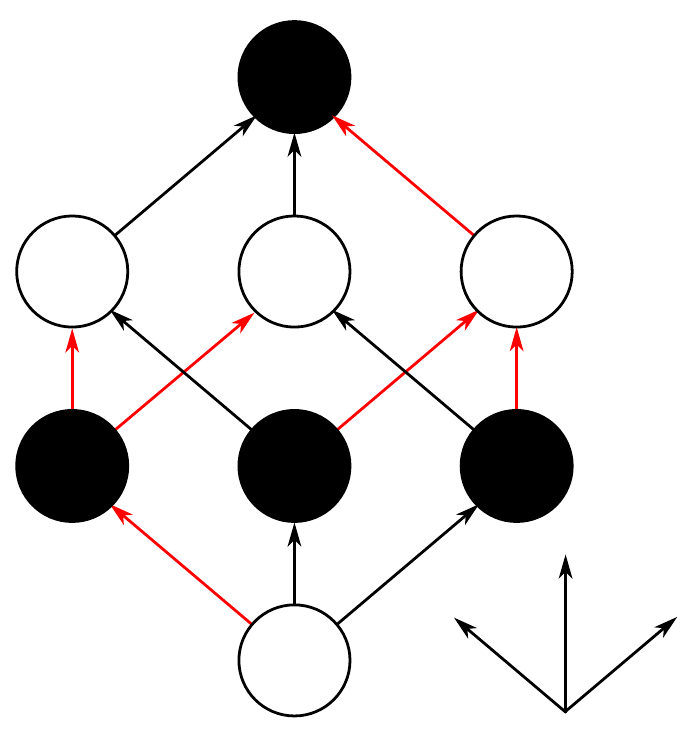}}
	\put(132, 10){$j_1$}
	\put(120, 20){$-$}
	\put(168, 10){$j_2$}
	\put(180, 20){$+$}
	\put(157, 30){$m_1$}
	\put(157, 40){$+$}
	\put(200, 100){$\longrightarrow$}
	\put(250, 0){\includegraphics[width=0.4\linewidth]{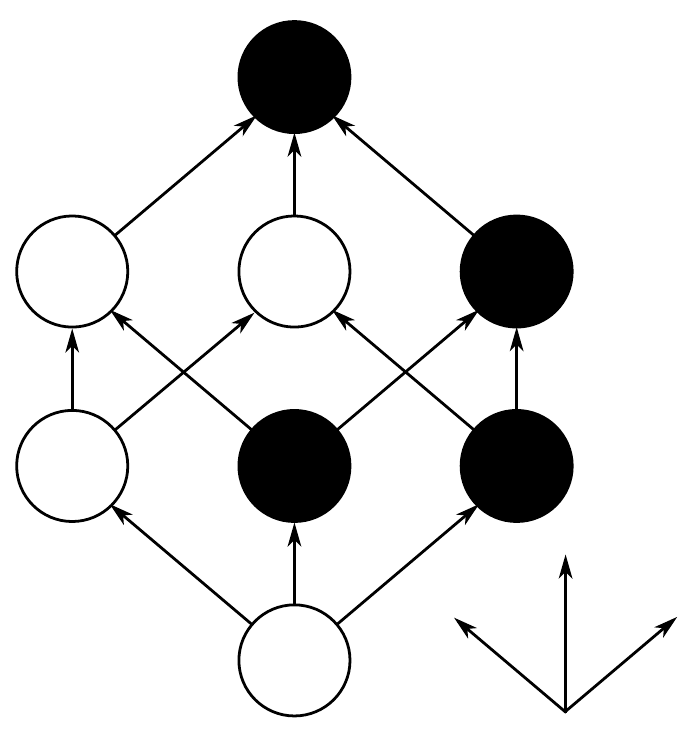}}
	\put(382, 10){$j_1$}
	\put(370, 20){$-$}
	\put(418, 10){$j_2$}
	\put(430, 20){$+$}
	\put(407, 30){$m_1$}
	\put(407, 40){$+$}
	\end{picture}
	\caption{Similarly to Figure~\ref{fig:crossing-cut-1}, this is an example of a function $\bh_i \colon \{0, 1\}^n \to \{0, 1\}$ with $\bh_i(x) = x_{j_1} \oplus x_{j_2} \oplus x_{m_1}$ variables $j_1$ (which ought to be anti-monotone), $j_2$ (which ought to be monotone), and $m_1$ (which is always monotone) being ``fixed'' into a function $\bh_i' \colon \{0, 1\}^n \to \{0, 1\}$ defined on the right-hand side.}
\end{figure}

\begin{figure}\label{fig:crossing-cut-2-2}
	\centering
	\begin{picture}(400, 160)
	\put(0,0){\includegraphics[width=0.4\linewidth]{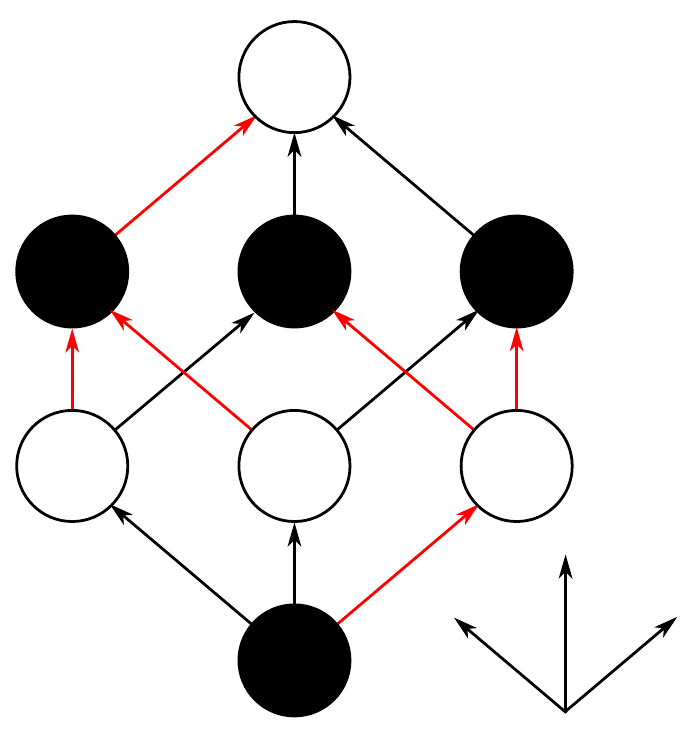}}
	\put(132, 10){$j_1$}
	\put(120, 20){$-$}
	\put(168, 10){$j_2$}
	\put(180, 20){$+$}
	\put(157, 30){$m_2$}
	\put(157, 40){$-$}
	\put(200, 100){$\longrightarrow$}
	\put(250, 0){\includegraphics[width=0.4\linewidth]{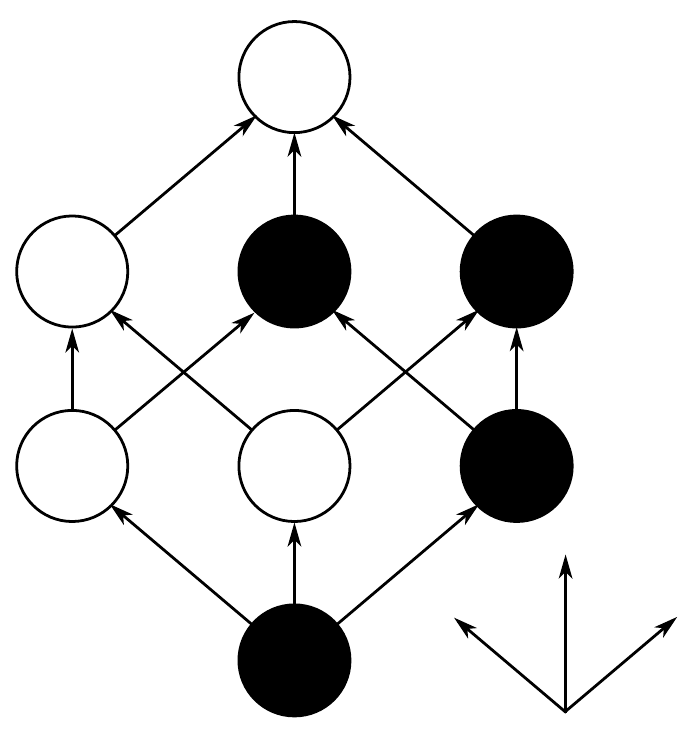}}
	\put(382, 10){$j_1$}
	\put(370, 20){$-$}
	\put(418, 10){$j_2$}
	\put(430, 20){$+$}
	\put(407, 30){$m_2$}
	\put(407, 40){$-$}
	\end{picture}
	\caption{Similarly to Figure~\ref{fig:crossing-cut-1}, this is an example of a function $\bh_i \colon \{0, 1\}^n \to \{0, 1\}$ with $\bh_i(x) = \neg x_{j_1} \oplus x_{j_2} \oplus x_{m_2}$ variables $j_1$ (which ought to be anti-monotone), $j_2$ (which ought to be monotone), and $m_2$ (which is always anti-monotone) being ``fixed'' into a function $\bh_i' \colon \{0, 1\}^n \to \{0, 1\}$ defined on the right-hand side.}
\end{figure}

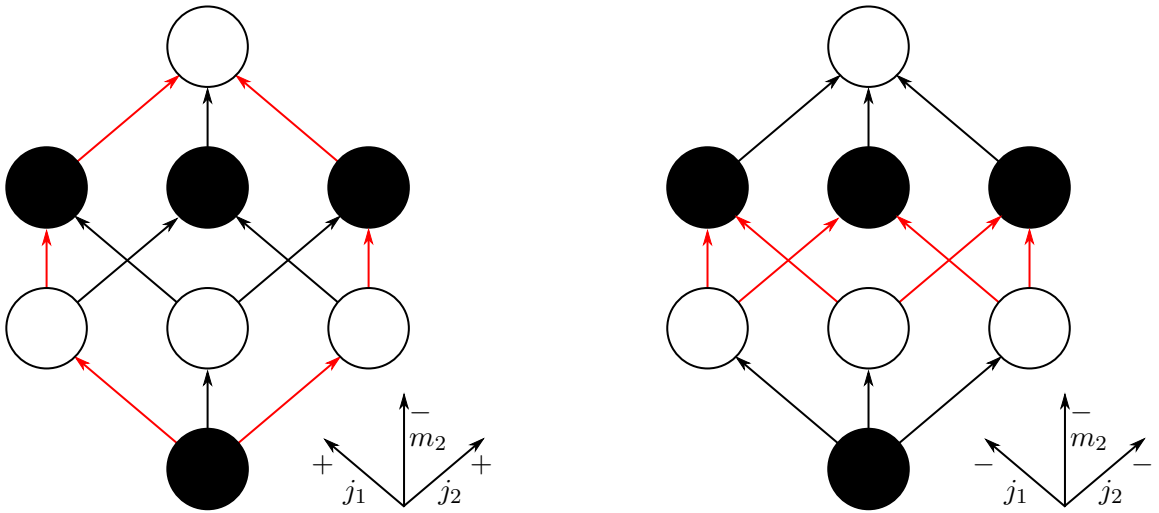
\begin{figure}\label{fig:not-crossing-cut-2}
	\centering
	\begin{picture}(400, 160)
	\put(0,0){\includegraphics[width=0.4\linewidth]{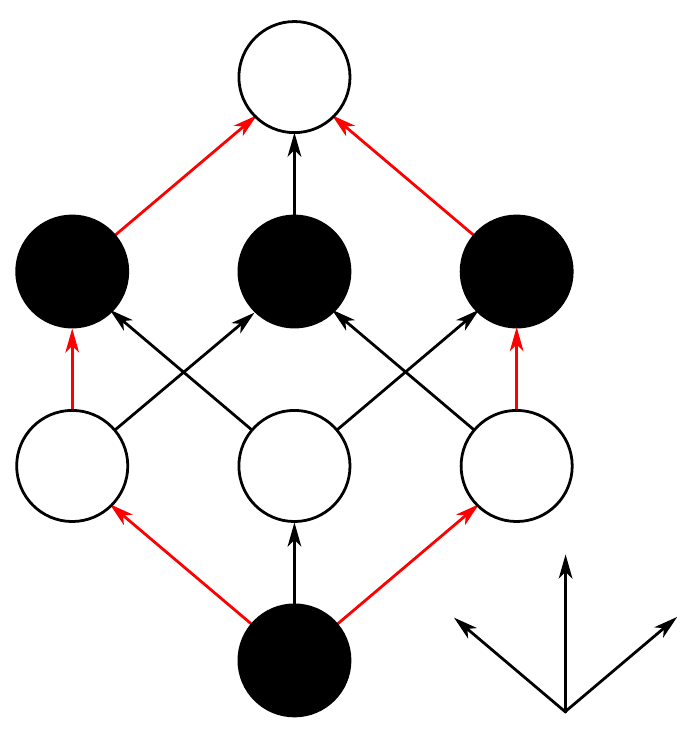}}
	\put(132, 10){$j_1$}
	\put(120, 20){$+$}
	\put(168, 10){$j_2$}
	\put(180, 20){$+$}
	\put(157, 30){$m_2$}
	\put(157, 40){$-$}
	\put(250, 0){\includegraphics[width=0.4\linewidth]{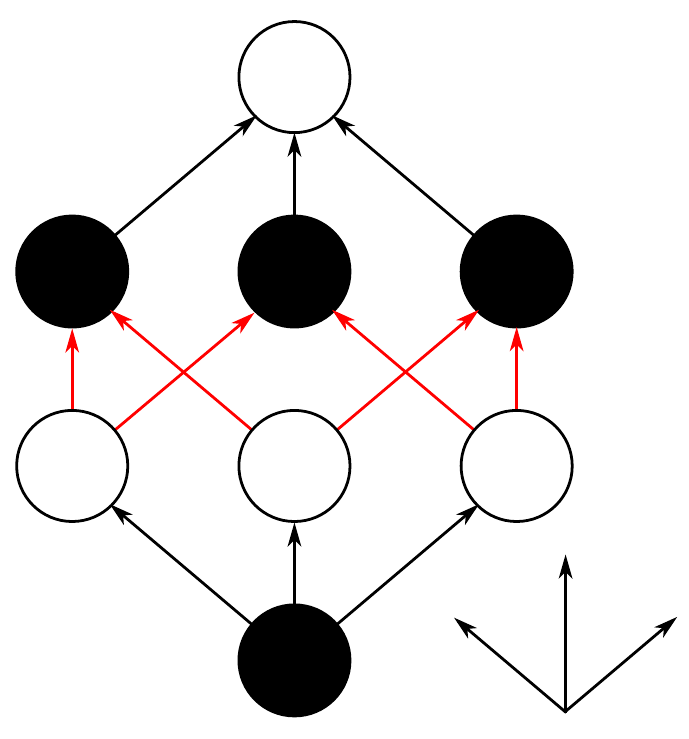}}
	\put(382, 10){$j_1$}
	\put(370, 20){$-$}
	\put(418, 10){$j_2$}
	\put(430, 20){$-$}
	\put(407, 30){$m_2$}
	\put(407, 40){$-$}
	\end{picture}
	\caption{Examples of functions $\bh_i \colon \{0, 1\}^n \to \{0, 1\}$ with orientations on the variables and violating edges. On the left-hand side, $\bh_i(x) = \neg x_{j_1} \oplus x_{j_2} \oplus x_{m_2}$ with variables $j_1$ and $j_2$  (which ought to be monotone), and $m_2$ (which is always anti-monotone). On the right-hand side, $\bh_i(x) = \neg x_{j_1} \oplus x_{j_2} \oplus x_{m_2}$ with variables $j_1$ and $j_2$ (which ought to be anti-monotone), and $m_2$ (which is always anti-monotone). We note that the violating edges form a cycle of length $6$, so any unate function whose orientations on $j_1$ and $j_2$ are as indicated (both monotone on the left-hand side, and both anti-monotone on the right-hand side) must disagree on a $\frac{3}{8}$-fraction of the points.}
\end{figure}

We consider the constants 
\[ \eps_0  = \frac{\gamma}{16} \qquad \text{ and }\qquad \eps_1 = \frac{5\gamma}{64}. \]
\begin{corollary}
	We have that $\boldf \sim \Dyes$ has $\dist(\boldf, \Unate) \leq \eps_0 + o(1)$ with high probability, and $\boldf \sim \Dno$ has $\dist(\boldf, \Unate) \geq \eps_1 - o(1)$ with high probability.
\end{corollary}

\begin{proof}
	We simply note that when $\bG = K_{\bA, \obA}$ (as is the case in $\Dyes$), we have $\chi(\bG) = 0$, and when $\bG = K_{\bA} \cup K_{\obA}$, we have $\chi(\bG) \to \frac{1}{2}$ as $n \to \infty$. 
\end{proof}
\newcommand{\obL}{\overline{\bL}}

\subsection{Reducing from Rejection Sampling}

The goal of this section is to prove the following two lemmas.

\begin{lemma} \label{lem:adaptive-reduct}
	Suppose there exists a deterministic algorithm $\Alg$ making $q$ queries to Boolean functions $f \colon \{0, 1\}^{2n} \to \{0, 1\}$. Then, there exists a deterministic non-adaptive algorithm $\Alg'$ making rejection sampling queries to an $n$-vertex graph with $\cost(\Alg') = q n$ such that:
	\begin{align*}
	\Prx_{\boldf \sim \Dyes}[\Alg(\boldf) \text{ ``accepts''}] &= \Prx_{\bG \sim \calG_2}[\Alg'(\bG) \text{ outputs ``$\calG_2$''}], \qquad\text{and} \\
	\Prx_{\boldf \sim\Dno}[\Alg(\boldf) \text{ ``accepts''}] &= \Prx_{\bG \sim \calG_1}[\Alg'(\bG) \text{ outputs ``$\calG_2$''}].
	\end{align*}
\end{lemma}

\begin{lemma}\label{lem:nonadaptive-reduct-unate}
	Suppose there exists a deterministic non-adaptive algorithm $\Alg$ making $q$ queries to Boolean functions $f\colon \{0, 1\}^{2n} \to \{0, 1\}$ where $q \leq \frac{n^{3/2}}{\log^8 n}$. Then, there exists a deterministic non-adaptive algorithm $\Alg'$ making rejection sampling queries to an $n$-vertex graph such that:
	\begin{align*}
	\Prx_{\boldf \sim \Dyes}[\Alg(\boldf) \text{ ``accepts''}] &\approx \Prx_{\bG \sim \calG_2}[\Alg'(\bG) \text{ outputs ``$\calG_2$''}] \pm o(1), \qquad\text{and} \\
	\Prx_{\boldf \sim\Dno}[\Alg(\boldf) \text{ ``accepts''}] &\approx \Prx_{\bG \sim \calG_1}[\Alg'(\bG) \text{ outputs ``$\calG_2$''}] \pm o(1).
	\end{align*}
	and has $\cost(\Alg') \leq q \sqrt{n} \log n$ with probability $1 - o(1)$ over the randomness in $\Alg'$.   
\end{lemma}

Combining Lemma~\ref{lem:adaptive-reduct} with Theorem~\ref{thm:bipartite}, we conclude Theorem~\ref{thm:unate}, and combining Lemma~\ref{lem:nonadaptive-reduct-unate} with Theorem~\ref{thm:bipartite}, we conclude Theorem~\ref{thm:unate-non}.

\subsection{Proof of Lemma~\ref{lem:adaptive-reduct}}

Consider an algorithm $\Alg$ making $q$ queries to a Boolean function which receives access to a Boolean function $\boldf = f_{\bT, \bA, \bH} \colon \{0, 1\}^{2n} \to \{0, 1\}$ (sampled from either $\Dyes$ or $\Dno$).

Since the values of $\bM, \bm_1, \bm_2$ and $\bT$ are distributed in the same way in $\Dyes$ and $\Dno$, a rejection sampling algorithm may generate $\bM, \bm_1, \bm_2$ and $\bT$, and utilize the randomness from rejection sampling to output values of $\bH$. In particular, for each query in $\Alg$, we will query the set $[n]$ in the rejection sampling algorithm. Then, given the edges sampled, as well as the values of $\bM$, $\bm_1$, $\bm_2$ and $\bT$, we will be able to simulate all the randomness in the construction of $\Dyes$ and $\Dno$. We give a formal description of a rejection sampling algorithm $\Alg'$ which assumes access to an algorithm $\Alg$ testing Boolean functions.

\begin{enumerate}
	\item We first sample $\bM \subset [2n]$ of size $n$, and let $\bm_1, \bm_2 \sim \bM$ be two distinct indices. Sample $\bT \sim \calE(\bM \setminus \{ \bm_1, \bm_2\})$. We may now view the hidden graph $\bG$ (from rejection sampling) as a graph on vertex set $\obM$.
	\item For each $t \in [q]$, perform the query $L_t = \obM$, which returns $(j_1^{(t)}, j_2^{(t)}) \in \bG$, we sample $\bj_3^{(t)} \sim \{ \bm_1, \bm_2 \}$ and $\bj^{(t)} \sim \{ \bm_1 , \bm_2 \}$. Intuitively, the values of $(j_1^{(t)}, j_2^{(t)}, \bj_{3}^{(t)})$ will generate the $t$-th accessed subfunction $\bh_i$ with $\Gamma_{\bT}(x) > 3N/4$, and $\bj^{(t)}$ will generate the $t$-th accessed subfunction $\bh_i$ with $\Gamma_{\bT}(x) \leq 3N/4$.
	\item We simulate $\Alg$ by maintaining two $q$-tuples $p_1, p_2 \in (\{0 \} \cup  [N])^q$, which is initially $p_1 = p_2 = (0, 0, \dots 0)$ which will record the indices of the subfunctions accessed. We proceed as follows, where we assume that $\Alg$ makes the query $z \in \{0, 1\}^{2n}$:
	\begin{itemize}
		\item Suppose $|z_{|\bM}| > \frac{n}{2} + \sqrt{2n}$, $|z_{|\bM}| < \frac{n}{2} - \sqrt{2n}$, $\Gamma_\bT(z) = 1^*$, or $\Gamma_\bT(z) = 0^*$, report to $\Alg$ the appropriate value of $\boldf(x)$.
		\item Otherwise, consider $\Gamma_{\bT}(z) = i \in [N]$. 
		\begin{itemize}
			\item Suppose $i \leq \frac{3N}{4}$ and $(p_1)_t = i$ (if $(p_1)_t \neq i$ for all $t$, then find the first $t \in [q]$ with $(p_1)_t = 0$ and write $(p_1)_t = i$). In this case, report $z_{\bj^{(t)}}$ if $\bj^{(t)} = \bm_1$ and $\neg z_{\bj^{(t)}}$ if $\bj^{(t)} = \bm_2$.
			\item If $i > \frac{3N}{4}$ and $(p_2)_t = i$ (again, if $(p_2)_t \neq i$ for all $t$, then find the first $t \in [q]$ with $(p_2)_t = 0$ and write $(p_2)_t = i$). In this case, we report $\neg z_{j_1^{(t)}} \oplus z_{j_2^{(t)}} \oplus z_{\bj_{3}^{(t)}}$ if $\bj_3^{(t)} = \bm_1$ and $z_{j_1^{(t)}} \oplus z_{j_2^{(t)}} \oplus z_{\bj_{3}^{(t)}}$ if $\bj_3^{(t)} = \bm_2$.
		\end{itemize}
	\end{itemize}
	\item If $\Alg$ outputs ``accept'', then $\Alg'$ outputs ``$\calG_2$'', if $\Alg$ outputs ``reject'', then $\Alg'$ outputs ``$\calG_1$''.
\end{enumerate}

Clearly, $\cost(\Alg') = q n$. In addition, we may view $\Alg'(\bG)$ as generating the necessary randomness for answering queries $\boldf(x)$ on the go, where $\bG$ will determine whether $\boldf \sim \Dyes$ or $\boldf \sim \Dno$. When $\bG = K_{\bA, \obA}$ (in the case $\bG \sim \calG_2$, the resulting function $\boldf$ is distributed as a function drawn from $\Dyes$; when $\bG = K_{\bA} \cup K_{\obA}$ (in the case $\bG \sim \calG_1$), the resulting function $\boldf$ is distributed as a function drawn from $\Dno$. Therefore, by the principle of deferred decisions, we have that $\Alg'(\bG)$ perfectly simulates queries to a Boolean function $\boldf \sim \Dyes$ (if $\bG \sim \calG_2$) or $\boldf \sim \Dno$ (if $\bG \sim \calG_1$). We conclude that
\begin{align*} 
\Prx_{\bG \sim \calG_1}[\Alg'(\bG) \text{ outputs ``$\calG_2$''} ] &= \Prx_{\boldf \sim \Dyes}[\Alg(\boldf) \text{ ``accepts''}],  \qquad\text{ and} \\
\Prx_{\bG \sim \calG_2}[\Alg'(\bG) \text{ outputs ``$\calG_2$''} ] &= \Prx_{\boldf \sim \Dno}[\Alg(\boldf) \text{ ``accepts''}].  
\end{align*}

\begin{remark}
	A close inspection of the proof of Lemma~\ref{lem:adaptive-reduct} reveals that the rejection sampling algorithm distinguishing $\calG_1$ and $\calG_2$ always makes queries $L_i = [n]$. This makes the lower bound simpler, as we can focus on proving lower bounds against algorithms which receive random edge samples. 
\end{remark}

\subsection{Proof of Lemma~\ref{lem:nonadaptive-reduct-unate}}

Similarly to the proof of Lemma~\ref{lem:adaptive-reduct}, we will proceed by generating the necessary randomness to generate the functions $\boldf$ from $\Dyes$ or from $\Dno$. However, unlike Lemma~\ref{lem:adaptive-reduct}, this will not be a black box reduction, since we will not be able to simulate $\boldf$ exactly.

Consider a deterministic non-adaptive algorithm $\Alg$ which makes queries to a Boolean function $\boldf \colon \{0, 1\}^{2n} \to \{0, 1\}$ sampled from $\Dyes$ or $\Dno$ and outputs ``accept'' if $\Alg$ believes $\boldf$ was sampled from $\Dyes$, and outputs ``reject'' if $\Alg$ believes $\boldf$ was sampled from $\Dno$. Since $\Alg$ is non-adaptive and deterministic, all queries are determined, so consider the queries $z_1, \dots, z_q \in \{0, 1\}^{2n}$, and let $\Alg \colon \{0, 1\}^{q} \to \{\text{``accept''}, \text{``reject''} \}$ be a function. 

We will now define a non-adaptive algorithm $\Alg'$ which makes rejection sampling queries to an unknown graph $\bG$ on $n$ vertices sampled from $\calG_1$ or from $\calG_2$. The algorithm $\Alg'$ proceeds as follows:
\begin{enumerate}
	\item Using some randomness and answers from rejection sampling queries to an unknown graph $\bG$, we generate a sequence of $r$ bits $(\boldr_1, \dots, \boldr_q)$ satisfying the following two conditions (we give the procedure to generate these random bits after)\footnote{With a slight abuse of notation, we let $\Alg'(\bG)$ correspond to to the output $(\boldr_1, \dots, \boldr_q)$ that $\Alg'$ produces with rejection sampling access to graph $\bG$.}:
	\begin{itemize}
		\item If $\bG \sim \calG_1$, then $(\boldr_1, \dots, \boldr_q)$ will be roughly distributed as $(\boldf(z_1), \dots, \boldf(z_q))$ where $\boldf$ is a Boolean function $\boldf \sim \Dno$.
		\item If $\bG \sim \calG_2$, then $(\boldr_1, \dots, \boldr_q)$ will be roughly distributed as a $(\boldf(z_1), \dots, \boldf(z_q))$ where $\boldf$ is a Boolean function $\boldf \sim \Dyes$. 
	\end{itemize}
	\item Finally, if $\Alg(\boldr_1, \dots, \boldr_q)$ outputs ``accept'', then $\Alg'$ outputs ``$\calG_2$'', and if $\Alg(\boldr_1, \dots, \boldr_q)$ outputs ``reject'', then $\Alg'$ outputs ``$\calG_1$''.
\end{enumerate}

In order to formalize the notion of ``roughly distributed as'' from above, let $\calV_{\text{yes}}$ and $\calV_{\text{no}}$ be the distributions supported on $\{0,1\}^{q}$ given by:
\begin{align*} 
\boldr &\sim \calV_{\text{yes}} \qquad\text{where}\qquad \forall i \in [q], \boldr_i = \boldf(z_i), \text{ and } \boldf \sim \Dyes.  \\
\boldr &\sim \calV_{\text{no}} \qquad \text{where}\qquad \forall i\in[q], \boldr_i = \boldf(z_i), \text{ and } \boldf \sim \Dno.
\end{align*}
Now, given the algorithm $\Alg'$, we let $\calU_{\text{yes}}, \calU_{\text{no}}$ be the distributions supported in $\{0, 1\}^{q}$ given by:
\begin{align*}
\boldr &\sim \calU_{\text{yes}} \qquad\text{where}\qquad \Alg'(\bG) \text{ outputs } (\boldr_1, \dots, \boldr_q) \text{ when } \bG \sim \calG_2 \\
\boldr &\sim \calU_{\text{no}} \qquad\text{where}\qquad \Alg'(\bG) \text{ outputs } (\boldr_1, \dots, \boldr_q) \text{ when } \bG \sim \calG_1
\end{align*}

The following lemma is a simple consequence will allow us to conclude Lemma~\ref{lem:nonadaptive-reduct-unate}.
\begin{lemma}\label{lem:distance-uv}
	Suppose $\calV_{\text{yes}}, \calV_{\text{no}}, \calU_{\text{yes}}$ and $\calU_{\text{no}}$ satisfy:
	\[ d_{TV}(\calV_{\text{yes}}, \calU_{\text{yes}}) = o(1) \qquad\text{and}\qquad d_{TV}(\calV_{\text{no}}, \calU_{\text{no}}) = o(1). \]
	Then, we have that:
	\begin{align*}
	\Prx_{\bG \sim\calG_1}[\Alg'(\bG) \text{ outputs ``$\calG_1$''}] &\approx \Prx_{\boldf \sim \Dno}[\Alg(\boldf) \text{ ``rejects''}] \pm o(1). \\
	\Prx_{\bG \sim \calG_2}[\Alg'(\bG) \text{ outputs ``$\calG_2$''} ] &\approx \Prx_{\boldf \sim \Dyes}[\Alg(\boldf) \text{ ``accepts''}] \pm o(1).
	\end{align*}
\end{lemma}

\begin{proof}
	We show the first inequality in the conclusion, as the argument is the same for the second inequality. Consider the set $R = \{ r \in \{0,1\}^{q} : \Alg(r) = \text{``reject''} \}$. Then, we have:
	\begin{align*}
	\Prx_{\boldf \sim \Dno}[\Alg(\boldf) \text{ ``rejects''}] &= \Prx_{\boldr \sim \calV_{\text{no}}}[ \boldr \in R] \\
	&\approx \Prx_{\boldr \sim \calU_{\text{no}}}[ \boldr \in R] \pm o(1) \\
	&\approx \Prx_{\bG\sim\calG_1}[\Alg'(\bG) \text{ outputs ``accept''}] \pm o(1).
	\end{align*}
\end{proof}

Given Lemma~\ref{lem:distance-uv}, it remains to describe the randomized procedure $\Alg'$ which given rejection sampling access to an unknown $n$-vertex graph $\bG$ from $\calG_1$ or $\calG_2$ outputs a bit-string of length $q$ such that:
\[ d_{TV}(\calV_{\text{yes}}, \calU_{\text{yes}}) = o(1) \qquad\text{and}\qquad d_{TV}(\calV_{\text{no}}, \calU_{\text{no}}) = o(1). \]
The procedure will work as follows:
\begin{enumerate}
	\item First, sample a random subset $\bM \subset [2n]$ of size $n$, and let $\bm_1, \bm_2 \sim \bM$ be two distinct random indices, and let $\bT \sim \calE(\bM \setminus\{\bm_1,\bm_2\})$. This defines an indexing function\footnote{Note that now, $N = 2^{\sqrt{2n}}$ since we are considering Boolean functions with $2n$ variables.} $\Gamma_{\bT} \colon \{0, 1\}^{2n} \to [N]$. We may view the unknown graph $\bG$ as being defined over vertices in $\obM$ \footnote{We may assume this by picking an arbitrary mapping of the indices in $\obM$ to $[n]$.}.
	\item We now consider partitioning the queries $z_1, \dots, z_q \in \{0, 1\}^{2n}$ into at most $t+4$ sets (where we will have $t \leq q$) $\bQ_{\bM}^{(+)}, \bQ_{\bM}^{(-)}, \bQ_{*}^{(0)}, \bQ_{*}{(1)}$ and $\bQ_{\ell_1}, \dots, \bQ_{\ell_t}$ non-empty sets where $\ell_{1}, \dots, \ell_{t} \subset [N]$:
	\begin{align*}
	\bQ_{\bM}^{(-)} &= \left\{ z_i : |(z_i)_{|\bM}| < \frac{n}{2} - \sqrt{2n} \right\}, \\\bQ_{\bM}^{(+)} &= \left\{z_i : |(z_i)_{|\bM}| > \frac{n}{2} + \sqrt{2n} \right\}, \\
	\bQ_{*}^{(0)} &= \left\{ z_i : \Gamma_{\bT}(z_i) = 0^* \wedge z_i \notin \bQ_{\bM}^{(-)} \cup \bQ_{\bM}^{(+)} \right\}, \\
	\bQ_{*}^{(1)} &= \left\{ z_i : \Gamma_{\bT}(z_i) = 1^* \wedge z_i \notin \bQ_{\bM}^{(-)} \cup \bQ_{\bM}^{(+)} \right\}, \\
	\bQ_{\ell} &= \left\{ z_i : \Gamma_{\bT}(z_i) = \ell \wedge z_i \notin \bQ_{\bM}^{(-)} \cup \bQ_{\bM}^{(+)} \right\}.
	\end{align*}
	\item If $z_i \in \bQ_M^{(-)}$, we let $\boldr_i = 0$, if $z_i \in \bQ_{M}^{(+)}$, we let $\boldr_i = 1$. If $z_i \in \bQ_{*}^{(0)}$, we let $\boldr_i = 0$, and if $z_i \in \bQ_{*}^{(1)}$, we let $\boldr_i = 1$. We may thus only consider the queries in $\bQ_{\ell_1}, \dots, \bQ_{\ell_t}$, and for simplicity in the notation, we re-index the queries to let:
	\[ \bQ_{\ell_i} = \left\{ z^{(i)}_{1}, z_{2}^{(i)}, \dots, z_{|\bQ_{\ell_i}|}^{(i)} \right\}  \]
	for each $i \in [t]$, and the corresponding bits $\boldr_{1}^{(i)}, \boldr_{2}^{(i)}, \dots, \boldr_{|\bQ_{\ell_i}|}^{(i)}$.
	\item We thus consider each $i \in [t]$ and independently set the values of $\boldr_{1}^{(i)}, \dots, \boldr_{|\bQ_{\ell_i}|}^{(i)}$ as follows:
	\begin{enumerate}
		\item If $\ell_i \leq 3N/4$, sample some $\bj \sim \{\bm_1, \bm_2\}$, and for every $\alpha \in [|\bQ_{\ell_i}|]$, let:
		\[ \boldr_{\alpha}^{(i)} = \left\{ \begin{array}{cc} (z_{\alpha}^{(i)})_{\bj} & \bj = \bm_1 \\
		\neg (z_{\alpha}^{(i)})_{\bj} & \bj = \bm_2 \end{array} \right. .\]
		\item Otherwise, if $\ell_i > 3N/4$, consider the following sets
		\[ \bL_i = \left\{ k \in \obM : \exists \alpha, \beta \in [|\bQ_{\ell_i}|], (z_{\alpha}^{(i)})_{k} \neq (z_{\beta}^{(i)})_{k} \right\}, \]
		and,
		\[ \obL_i^{(0)} = \left\{ k \in \obM \setminus \bL_i : z \in \bQ_{\ell_i}, z_{k} = 0 \right\}\qquad \obL_{i}^{(1)} = \left\{ k \in \obM \setminus \bL_i : z \in \bQ_{\ell_i}, z_k = 1 \right\} \;.\]
		
		We make the query $\bL_i$ if $|\bL_i| \leq \frac{n}{\log n}$ and $\obM$ otherwise to the rejection sampling oracle and obtain a response $\bv \in (\obM \times \obM) \cup \obM \cup \{\emptyset\}$. In addition, sample $\bj_3 \sim \{ \bm_1, \bm_2 \}$. We now consider three cases:
		\begin{enumerate}
			\item If $\bv = (\bj_1, \bj_2) \in \obM \times \obM$ is an edge, then for each $\alpha \in [|\bQ_{\ell_i}|]$, we let:
			\[ \boldr_{\alpha}^{(i)} = \left\{ \begin{array}{cc} (z_{\alpha}^{(i)})_{\bj_1} \oplus (z_{\alpha}^{(i)})_{\bj_2} \oplus (z_{\alpha}^{(i)})_{\bj_3} & \bj_3 = \bm_1 \\
			\neg (z_{\alpha}^{(i)})_{\bj_1} \oplus (z_{\alpha}^{(i)})_{\bj_2} \oplus (z_{\alpha}^{(i)})_{\bj_3} & \bj_3 = \bm_2 \end{array} \right. .\]
			\item If $\bv = \bj_2 \in \obM$ is a lone vertex, then let $w = \neg (z_1^{(i)})_{\bj_2}$ and $p_v(\obL_{i}^{(w)}) = \frac{|\obL_i^{(w)}|}{|\obL_i|}$, we sample $\bb \sim \Ber(p_v(\obL_i^{(w)}))$ and for each $\alpha \in [|\bQ_{\ell_i}|]$, we let:
			\[ \boldr_{\alpha}^{(i)} \oplus (z_{\alpha}^{(i)})_{\bj_2} \oplus (z_{\alpha}^{(i)})_{\bj_3} = \left\{ \begin{array}{cc} \bb \oplus (z_{1}^{(i)})_{\bj_2} & \bj_3 = \bm_1 \\
			\neg \bb \oplus (z_1^{(i)})_{\bj_2} & \bj_3 = \bm_2 \end{array} \right. .\]
			\item Lastly, if $\bv = \emptyset$ is the empty set, then let $p_{\emptyset}(\obL_i) = \frac{2|\obL_i^{(0)}| |\obL_i^{(1)}|}{|\obL_i|^2}$ and sample $\bb \sim \Ber(p(\obL_i))$ and for each $\alpha \in [|\bQ_{\ell_i}|]$, we let:
			\[ \boldr_{\alpha}^{(i)} \oplus (z_{\alpha}^{(i)})_{\bj_3} = \left\{ \begin{array}{cc} \bb & \bj_3 = \bm_1 \\
			\neg \bb & \bj_3 = \bm_2 \end{array} \right. .\]
		\end{enumerate}
	\end{enumerate} 
\end{enumerate}

\begin{remark}
	The procedure described above does not exactly simulate queries to a $\boldf \sim \Dyes$ or $\Dno$ (in the case of $\bG \sim \calG_2$ or $\bG \sim \calG_1$, respectively) as in the reductions of Lemma~\ref{lem:adaptive-reduct} and Lemma~\ref{lem:juntarejustion}). Let us briefly explain why this happens by giving an illuminating example. Consider a one-query algorithm which makes query $z \in \{0, 1\}^n$ and suppose $|z_{\bM}| \approx \frac{n}{2} \pm \sqrt{2n}$ and $\Gamma_{\bT}(z) = i > \frac{3N}{4}$ with $z_{m_1} = 0$ and $z_{m_2} = 1$. Then, the value $\boldf(z) = \bh_i(z)$ will be $0$ if $z_{\bj_1} = z_{\bj_2}$, and $1$ if $z_{\bj_1} \neq z_{\bj_2}$, where $(\bj_1, \bj_2) \sim \bG$ is the edge sampled for subfunction $\bh_i$.
	
	We note that this probability is slightly different for $\bG \sim \calG_1$ and $\bG \sim \calG_2$ and depends on how $\bA$ partitions the $0$-variables and $1$-variables of $z$. Despite this difference, $\Alg'$ always observes $\emptyset$ from the rejection sampling oracle, so the output bit $\boldr \in \{0, 1\}$ which $\Alg'$ produces will not simulate $\boldf(z)$ exactly. The bulk of the argument shows that $\Alg'$ can sample a random bit whose distribution is close to $\boldf(z)$ in total variation distance, so that $\Alg$ cannot exploit the fact that the simulation is not exact.
\end{remark}

We first note the following lemma.
\begin{lemma}\label{lem:cost-blowup}
	With probability $1 - o(1)$ over the draw of $\bM \subset [n]$, $\bm_1, \bm_2$ and $\bT \sim \calE(\bM \setminus \{ \bm_1, \bm_2\})$, we have that for all $i \in [t]$, 
	\[ |\bL_i| \leq |\bQ_{\ell_i}| \cdot 90 \sqrt{n} \log n. \]
\end{lemma}

\begin{proof}
	We will prove this by showing that for any two $z, z' \in \bQ_{\ell_i}$, $\|z - z'\|_{1} \leq 90\sqrt{n} \log n$ with probability $1 - \frac{1}{n^{10}}$, so that we may union bound over all possible pairs. More specifically, consider two queries $z, z' \in \{0, 1\}^{2n}$ which differ by more than $90\sqrt{n} \log n$ indices. Note that the distribution of the random variable $\|(z-z')|_\bM\|_1\sim \mathrm{HG}(2n,|z-z'|,n)$.   Then using Theorem~\ref{thm:HG-Chernoff} we have that with probability at least $1 - \frac{1}{n^{10}}$ over the draw of $\bM$,  $\|z_{|\bM} - z_{|\bM}'\|_{1} \geq 30{\sqrt{n}\log n}$. 
	
	Next, if $|z_{|\bM}| \approx \frac{n}{2} \pm \sqrt{2n}$ and $|z'_{|\bM}| \approx \frac{n}{2} \pm \sqrt{2n}$ (if either of these conditions do not hold, then we know the strings are not in $\bQ_{\ell_i}$ for any $i$), then there exists a set $\bP \subset \bM$ with $|\bP| = 15{\sqrt{n} \log n}$ such that for all $k \in \bP$, $z_{k} = 1$ and $z'_{k} = 0$. Thus, we have that:
	\[ \Prx_{\bT}[\exists i \in [t], z, z' \in \bQ_{\ell_i}] \leq \Prx_{\bT}[ z' \in \bQ_{\ell_i} \mid z \in \bQ_{\ell_i}] \leq \Prx_{\bT_{\ell_i}}[\bT_{\ell_i} \cap \bP = \emptyset] \leq \left( 1 - \frac{15\log n}{\sqrt{n}} \right)^{\sqrt{n}} \ll \frac{1}{n^{10}}. \]
	So we may union bound over all pairs of queries to conclude that if $z, z' \in \bQ_{\ell_i}$, then $\|z - z'\|_1 \leq 90\sqrt{n}\log n$ with high probability, which gives the desired claim.
\end{proof}

Thus, given Lemma~\ref{lem:cost-blowup} as well as the fact that we query $[n]$ when $|\bL_i| \geq \frac{n}{\log n}$, we conclude that if $\Alg$ makes $q$ queries, then $\Alg'$ has complexity at most $q\cdot O(\sqrt{n} \log^2 n)$ in the rejection sampling model.

\begin{lemma}\label{lem:good-split}
	If $q\le \frac{n^{3/2}}{\log^{8}n }$, then with probability $1 - o(1)$ over the draw of $\bM \subset [n], \bm_1, \bm_2$, $\bT$, and $\bA \subset \obM$, we have that for every $i \in [t]$ where $|\bL_i| \leq \frac{n}{\log n}$, the sets $|\obL_i^{(0)}|, |\obL_i^{(1)}| $ 
	satisfy the following
	\[|\obL_i^{(0)}|, |\obL_i^{(1)}| = \Omega(n)\;, \]
	\[ \left|\bA \cap \obL_{i}^{(0)}\right| \approx \dfrac{\left|\obL_{i}^{(0)}\right|}{2} \pm \sqrt{n} \log n \qquad\text{and}\qquad \left| \bA \cap \obL_{i}^{(1)}\right| \approx \dfrac{\left| \obL_{i}^{(1)} \right|}{2} \pm \sqrt{n}\log n. \]
\end{lemma}

\begin{proof}
	We first claim that with probability $1-o(1)$ over the choice of $\bM$, all the  queries $z\in\{0,1\}^{2n}$ that are mapped to some $\bQ_{\ell_i}$ are such  that $|z|\approx n\pm 50\sqrt{2n}\log n$. Assume $z\in\{0,1\}^{2n}$ is such that $|z|> n+50\sqrt{2n}\log n$, and consider the random variable $|z|_\bM|$. Note that the distribution of $|z|_\bM|$ is hyper-geometric with parameters $(2n,|z|,n)$. By using Theorem~\ref{thm:HG-Chernoff} on the tail bounds for hyper-geometric random variable, we get that for any $t>0$ \[ \Prx_{\bM}\left[|z|_{\bM}|<\left(\frac{|z|}{2n}-t\right)n\right] \le e^{-2t^2n} \;. \]
	By choosing  $t=\frac{50\log n}{\sqrt{2n}}-\frac{\sqrt{2}}{\sqrt{n}}$, and considering the complement event, we have that 
	\[ \Prx_{\bM}\left[ |z|_{\bM}|\ge \frac{|z|}{2} - \frac{50\sqrt{n}\log n}{\sqrt{2}}+\sqrt{2n} \right] \ge 1-\frac{1}{n^{50}}\;.\] 
	Combining this with the fact that $|z|>n+50\sqrt{2n}\log n$, we get that the probability that $|z|_\bM|>n/2+\sqrt{2n}$ is at least $1-1/n^{50}$. 
	
	Similarly, we get that when $|z|< n-50\sqrt{2n}\log n$ , we have that with probability $1-1/n^{50}$ over the choice of $\bM$, $|z|_\bM|<n/2-\sqrt{2n}$. By using a union bound on the number of queries we get that with probability $1-o(1)$ over the choice of $\bM$, all the  queries $z\in\{0,1\}^{2n}$ that are mapped to some $\bQ_{\ell_i}$  are such that $|z|\approx n\pm 50\sqrt{2n}\log n$. 
	
	We henceforth condition on such $\bM=M$. Consider any $T\sim\calE(M)$ and all the indices $i\in[t]$ such that $|L_i|\le\frac{n}{\log n}$. By definition, if $z\in\{0,1\}^{2n}$ is mapped to some $Q_{\ell_i}$, then $|z|_M| \approx n/2\pm\sqrt{2n}$, which implies that $|z|_{\oM}|\approx n/2\pm 49\sqrt{2n}\log n$. Therefore, by the fact that all queries in $Q_{\ell_i}$ must agree on all of the coordinates in $\overline L_{i}$, we can conclude that $|\overline L_{i}^{(0)}|$ and $|\overline L_{i}^{(1)}|$ are $\Omega(n)$. 
	
	Next, consider the random variable $|\bA\cap \overline L_{i}^{(1)}|$, and note that its distribution is hyper-geometric with parameters $(n,|\overline{L}_i^{(1)}|,n/2)$. By using tail bounds for hyper-geometric random variable, we get that with probability at least $1-o(1)$ over the choice of $\bA$ \[|\bA\cap \overline L_i^{(1)}|\approx \frac{|L_i^{(1)}|}{2}\pm \sqrt{n}\log n \;. \]
	Using the same argument, we also get that with probability $1-o(1)$ over the choice of $\bA$ we have that \[|\bA\cap \overline L_i^{(0)}|\approx \frac{|L_i^{(0)}|}{2}\pm \sqrt{n}\log n \;. \]
	
	By applying a union bound over all indices $i\in[t]$ the lemma follows.
\end{proof}

\begin{lemma}\label{lem:few-lone}
	If $\cost(\Alg') \leq \frac{n^{2}}{\log^{6} n}$ which occurs with high probability over $\bM$, with probability $1 - o(1)$ over the draw of $\bv$ in Step 4(b), there are at most $\frac{n}{\log^4 n}$ responses $\bv \in \obM$ which are lone vertices of case (ii).
\end{lemma}

As discussed earlier, the proof of the above lemma is given in the lower bound for distinguishing $\calG_1$ and  $\calG_2$ in Section~\ref{sec:LowerBound} (Lemma~\ref{lem:calE_F}). We assume its correctness for the rest of this section.
\medskip

We note that since $\bM, \bm_1, \bm_2$ and $\bT$ are distributed in the same way in $\boldf \sim \Dyes$ and in Step 1 of $\Alg$, we may consider the distribution $\calV_{\text{yes}}(M, m_1, m_2, T)$ denoting $\calV_{\text{yes}}$ conditioned on $\bM = M, \bm_1 = m_1, \bm_2 = m_2$ and $\bT = T$, and we analogously define $\calU_{\text{yes}}(M, m_1, m_2, T)$, $\calV_{\text{no}}(M, m_1, m_2, T)$ and $\calU_{\text{no}}(M, m_1, m_2, T)$. In addition, we may denote the event $\calbE_{A}$ to denote the event that the hidden subset $\bA$ sampled in $\boldf$ or in the graph $\bG$ satisfies the conditions of Lemma~\ref{lem:good-split}, and the event $\calbE_{V}$ to be the event that there are at most $\frac{n}{\log^4 n}$ responses which are lone vertices from Lemma~\ref{lem:few-lone}. We thus consider a fixed set $M, m_1, m_2,$ and $T$ satisfying the following conditions of Lemma~\ref{lem:cost-blowup} and consider the distribution $\calV_{\text{yes}}'$ to be the distribution given by sampling $\boldr \sim \calV_{\text{yes}}(M, m_1, m_2, T)$ conditioned on events $\calbE_{A}$ and $\calbE_{V}$. We analogously define $\calV_{\text{no}}'$, $\calU_{\text{yes}}'$ and $\calU_{\text{no}}'$. We note it suffices to show $d_{TV}(\calV_{\text{yes}}', \calU_{\text{yes}}') = o(1)$ and $d_{TV}(\calV_{\text{no}}', \calU_{\text{no}}') = o(1)$. 

We now note that conditioned on $M, m_1, m_2$ and $T$, the sets $\bQ_{\bM}^{(+)}, \bQ_{\bM}^{(-)}, \bQ_{*}^{(1)}$ and $\bQ_{*}^{(0)}$, as well as all $\bQ_{\ell_1}, \dots, \bQ_{\ell_t}$ are no longer random. Furthermore, when $z \in \bQ_{\bM}^{(+)} \cup \bQ_{\bM}^{(-)} \cup \bQ_{*}^{(1)} \cup \bQ_{*}^{(0)}$ the values of $f_{T, \bA, \bH}(z)$ from $\Dyes$ (and from $\Dno$) are fixed to their corresponding values according to (\ref{eq:unate-dist}), which match their settings in $\calU_{\text{yes}}'$ and $\calU_{\text{no}}'$. Likewise, when $z \in \bQ_{\ell_i}$ with $\ell_i \leq \frac{3N}{4}$, $f_{T, \bA, \bH}(z)$ is determined by a dictator or anti-dictator in $\{ m_1, m_2\}$; by the principle of deferred decisions, the values of $f_{T, \bA, \bH}(z)$ can be simulated exactly. Therefore, it remains to consider the values of $\boldr_{\alpha}^{(i)}$ corresponding to $f_{T, \bA, \bH}(z_{\alpha}^{(i)})$ for each $i \in [t]$, where $\ell_i > \frac{3N}{4}$, so for simplicity, assume that every $\ell_i > \frac{3N}{4}$. 

Consider a function $v \colon [t] \to \{ \text{``edge''}, \text{``lone vertex''}, \text{``empty set''} \}$ which indicates whether the response of the $i$th rejection sampling query sampled in Step 4(b) falls into case (i) (when $\bv_i$ is an edge), or case (ii) (when $\bv_i$ is a lone vertex), or case (iii) (when $\bv_i$ is $\emptyset$). In other words, 
\[ v(i) =\left\{\begin{array}{cc} \text{``edge''} & \bv_i \in \obM \times \obM \\
\text{``lone vertex''} & \bv_i \in \obM \\
\text{``empty set''} & \bv_i = \emptyset \end{array} \right. \]

We thus consider one fixed function $v \colon [t] \to \{\text{``edge''}, \text{``lone vertex''}, \text{``empty set''} \}$ and condition on the fact that $v$ specifies the three cases of Step 4(b) (in the case of $\calU_{\text{yes}}$ and $\calU_{\text{no}}$) and whether the edge sampled $(\bj_1, \bj_2) \sim \bG$ in the fourth step of generating $\Dyes$ and $\Dno$ for $\bh_{\ell_i}$ either intersects $\bL_i$ fully (in the case of an edge), or partially (in the case of a lone vertex), or it does not intersect at all (in the case of the empty set). Thus, again, we may consider the distributions conditioned on the edges sampled are specified correctly by $v$.

The following three lemmas give the distribution of $\boldr_1^{(i)} \sim \calV_{\text{yes}}'$ and $\boldr_1^{(i)} \sim \calV_{\text{no}}'$ in the cases when $\bv_i$ is an edge, or a lone vertex, or the empty set. We note that the three lemmas indicate how to generate the bits $\boldr_{\alpha}^{(i)}$ in Step 4(b) of $\Alg'$.
\begin{lemma}
	For every $i \in [t]$ with $v(i) = \text{``edge''}$, we have that every $\alpha \in [|\bQ_{\ell_i}|]$ has $\boldr_\alpha^{(i)}$ generated from $\Alg'$ is distributed exactly as $\boldf(z_{\alpha}^{(i)})$.
\end{lemma}

\begin{proof}
	This simply follows from the principle of deferred decisions, since $\Alg'$ generates all the necessary randomness to simulate a query to a function $\boldf \sim \Dyes$ or $\boldf \sim \Dno$ which indexes to the sub-function $\bh_{\ell_i}$.
\end{proof}

\begin{lemma}\label{cl:empty-set-val-dist}
	For every $i \in [t]$ with $v(i) = \text{``empty set''}$, there exists $|\gamma_{\text{yes}}|, |\gamma_{\text{no}}| = O(\tfrac{\log^2 n}{n})$ such that for $\boldr \sim \calV_{\text{yes}}'$ satisfies 
	\[ \boldr_{1}^{(i)} \oplus (z_1^{(i)})_{\bj_3} \sim \left\{ \begin{array}{cc} \Ber\left(p_{\emptyset}(\obL_i) + \gamma_{\text{yes}}\right) & \bj_{3} = m_1 \\
	\Ber(1 - p_{\emptyset}(\obL_i) - \gamma_{\text{yes}}) & \bj_3 = m_2 \end{array}\right. ,\] 
	and $\boldr \sim \calV_{\text{no}}'$ satisfies 
	\[ \boldr_{1}^{(i)} \oplus (z_{1}^{(i)})_{\bj_3} \sim \left\{ \begin{array}{cc} \Ber\left(p_{\emptyset}(\obL_i) + \gamma_{\text{no}}\right) & \bj_3 = m_1 \\
	\Ber\left(1 - p_{\emptyset}(\obL_i) - \gamma_{\text{no}}\right) & \bj_3 = m_2 \end{array} \right. . \]
\end{lemma}

\begin{proof}
	We recall that $\bh_{\ell_i}$ is determined by $(\bj_1, \bj_2) \sim \bG$ and $\bj_3 \sim \{m_1, m_2\}$ in the fourth step of generating $\boldf \sim \Dyes$ or $\Dno$. Consider the case when $\bj_3 = m_1$, and the case when $(z_{1}^{(i)})_{m_1} = 0$ (since the case $(z_1^{(i)})_{m_1} = 1$ is symmetric, except we flip the answer). 
	
	Recall that we condition on the fact that the edge $(\bj_1, \bj_2) \sim \bG$ satisfies $\bL_i \cap \{ \bj_1, \bj_2\} = \emptyset$, as well as the conclusions from Lemma~\ref{lem:good-split}, so we may write:
	\begin{align}
	&\Prx_{\boldr \sim \calV_{\text{yes}} }\left[\boldr_{1}^{(i)} = 1 \mid v(i) = \text{``empty set''}\right] \nonumber \\
	&= \Prx_{\substack{\bG \sim \Dno \\ (\bj_1, \bj_2)}}\left[ \left(\bj_1 \in \bA \cap \obL_i^{(0)} \wedge \bj_2 \in \obA \cap \obL_i^{(1)} \right) \vee \left(\bj_1 \in \bA \cap \obL_i^{(1)} \wedge \bj_2 \in \obA \cap \obL_i^{(0)} \right) \mid v(i) = \text{``empty set''} \right], \nonumber \\
	&= \dfrac{1}{|\bA \cap \obL_i| \cdot |\obA \cap \obL_i|} \cdot \left(|\bA \cap \obL_i^{(0)}| \cdot |\obA \cap \obL_i^{(1)}| + |\bA \cap \obL_i^{(1)}| \cdot |\obA \cap \obL_i^{(0)}| \right) \label{eq:yes-empty-prob} 
	\end{align}
	since the value of $\boldf(z_{1}^{(i)})$ in the case of $\bj_3 = m_1$ will be a parity of the end points, so this parity will be 1 when the values of the variables $\bj_1$ and $\bj_2$ under $z_{1}^{(i)}$ disagree. In order to see this, we recall that $\bG$ is the complete bipartite graph (in the case when $\boldr \sim \calV_{\text{yes}}$) with sides $\bA$ and $\obA$, so the edge $(\bj_1, \bj_2) \in \bA \times \obA$ must have $(z_{1}^{(i)})_{\bj_1} \neq (z_{1}^{(i)})_{\bj_2}$, and $\bj_1, \bj_2 \in \obL_i$. 
	
	Since $v(i) = \text{``empty set''}$, we note that $|\bL_i| \leq \frac{n}{\log n}$, so $|\obL_i| = \Omega(n)$. In addition, by Lemma~\ref{lem:good-split}, let:
	\begin{align} 
	|\bA \cap \obL_{i}^{(0)}| = \frac{|\obL_i^{(0)}|}{2} + \xi_{0} \qquad\text{and}\qquad |\bA \cap \obL_i^{(1)}| = \frac{|\obL_i^{(1)}|}{2} + \xi_1, \label{eq:A-L-bound}
	\end{align}
	where $|\xi_0|, |\xi_1| \leq \sqrt{n} \log n$, which in turn, implies:
	\begin{align} |\obA \cap \obL_i^{(0)}| = \frac{|\obL_i^{(0)}|}{2} - \xi_0 \qquad\text{and}\qquad |\obA \cap \obL_i^{(1)}| = \frac{|\obL_i^{(1)}|}{2} - \xi_1. \label{eq:barA-L-bound}
	\end{align}
	Therefore, combining (\ref{eq:yes-empty-prob}) with (\ref{eq:A-L-bound}) and (\ref{eq:barA-L-bound}),
	\begin{align*}
	&\Prx_{\boldr \sim \calV_{\text{yes}} }\left[\boldr_{1}^{(i)} = 1 \mid v(i) = \text{``empty set''}\right]\\
	& = \dfrac{1}{\left(\frac{|\obL_i|}{2} + \xi_0 + \xi_1 \right) \left(\frac{|\obL_i|}{2} - \xi_0 - \xi_1 \right)} \left( \left(\frac{|\obL_i^{(0)}|}{2} + \xi_0 \right) \left(\frac{|\obL_{i}^{(1)}|}{2} - \xi_1 \right) + \left(\frac{|\obL_{i}^{(1)}|}{2} + \xi_1 \right)\left(\frac{|\obL_i^{(0)}|}{2} - \xi_0 \right)\right) \\
	&= \dfrac{2|\obL_i^{(0)}| \cdot |\obL_i^{(1)}| - 8\xi_0 \xi_1}{|\obL_i|^2 - 4\xi_0^2 - 4\xi_1^2 - 8\xi_0 \xi_1} = \dfrac{2|\obL_i^{(0)}|\cdot |\obL_i^{(1)}|}{|\obL_i|^2} + \gamma_{\text{yes}},
	\end{align*}
	where $|\gamma_{\text{yes}}| \leq O(\tfrac{\log^2 n}{n})$, since $|\obL_i|, |\obL_{i}^{(0)}|, |\obL_i^{(1)}| = \Omega(n)$.
	
	The case when $\boldr \sim \calV_{\text{no}}$ is analogous, except that now the underlying graph is the union of two cliques at $\bA$ and $\obA$, so:
	\begin{align*}
	&\Prx_{\boldr \sim \calV_{\text{no}} }\left[\boldr_{1}^{(i)} = 1 \mid v(i) = \text{``empty set''}\right] \\
	&= \Prx_{\substack{\bG \sim \calG_1 \\ (\bj_1, \bj_2)}}\left[ \left(\bj_1 \in \bA \cap \obL_i^{(0)} \wedge \bj_2 \in \bA \cap \obL_i^{(1)} \right) \vee \left(\bj_1 \in \obA \cap \obL_i^{(0)} \wedge \bj_2 \in \obA \cap \obL_i^{(1)} \right) \mid v(i) = \text{``empty set''} \right], \\
	&= \dfrac{1}{\binom{|\bA \cap \obL_i|}{2} + \binom{|\obA \cap \obL_i|}{2}} \cdot \left(|\bA \cap \obL_i^{(0)}| \cdot |\bA \cap \obL_i^{(1)}| + |\obA \cap \obL_i^{(0)}| \cdot |\obA \cap \obL_i^{(1)}| \right) \\
	&= \dfrac{1}{\binom{\frac{|\obL_i|}{2} + \xi_0 + \xi_1}{2} + \binom{\frac{|\obL_i|}{2} -\xi_0 - \xi_1}{2}} \left(\left(\frac{|\obL_i^{(0)}|}{2} + \xi_0 \right)\left(\frac{|\obL_i^{(1)}|}{2} + \xi_1 \right) + \left(\frac{|\obL_i^{(0)}|}{2} - \xi_0 \right)\left(\frac{|\obL_i^{(1)}|}{2} - \xi_1 \right) \right) \\
	&= \dfrac{2|\obL_i^{(0)}| \cdot |\obL_i^{(1)}|}{|\obL_i|^2} + \gamma_{\text{no}},
	\end{align*}
	were again, $|\gamma_{\text{no}}| \leq O(\tfrac{\log^2 n}{n})$.
\end{proof}

\begin{lemma}\label{cl:lone-vertex-val-dist}
	For every $i \in [t]$ with $v(i) = \text{``lone vertex''}$, let $\bj_2 \in \oM$ be the lone vertex observed and let $w = \neg (z_1^{(i)})_{\bj_2}$. There exists $|\gamma_{\text{yes}}'|, |\gamma_{\text{no}}'| \leq O(\tfrac{\log n}{\sqrt{n}})$ such that for $\boldr \sim \calV'_{\text{yes}}$ satisfies
	\[ \boldr_1^{(i)} \oplus (z_1^{(i)})_{\bj_3} \sim\left\{ \begin{array}{cc} \Ber(p_v(\obL_i^{(w)}) + \gamma_{\text{yes}}') & \bj_3 = m_1 \\
	\Ber(1 - p_v(\obL_i^{(w)}) - \gamma_{\text{yes}}') & \bj_3 = m_2 \end{array} \right. ,\]
	and $\boldr \sim \calV'_{\text{no}}$ satisfies
	\[ \boldr_1^{(i)} \oplus (z_1^{(i)})_{\bj_3} \sim\left\{ \begin{array}{cc} \Ber(p_v(\obL_i^{(w)}) + \gamma_{\text{no}}') & \bj_3 = m_1 \\
	\Ber(1 - p_v(\obL_i^{(w)}) - \gamma_{\text{no}}') & \bj_3 = m_2 \end{array}\right.. \]
\end{lemma}

\begin{proof}
	We follow a similar strategy to Lemma~\ref{cl:empty-set-val-dist}, where we know that we sample an edge $(\bj_1, \bj_2) \sim \bG$ whose value of $\bj_2 \in \bL_i$, and $\bj_1 \notin \bL_i$. Consider for simplicity the case when $\bG$ is a complete bipartite graph with sides $\bA$ and $\obA$, and $\bj_3 = m_1$ and $(z_1^{(i)})_{m_1} = 0$. 
	
	Similarly to (\ref{eq:yes-empty-prob}), we have that in order for $\boldr_1^{(i)} = 1$, we must have $(z_{1}^{(i)})_{\bj_1} \neq (z_{1}^{(i)})_{\bj_2}$. Suppose that $\bj_2 \in \bA$ and $w = \neg (z_{1}^{(i)})_{\bj_2}$, then in order for $\boldr_1^{(i)} = 1$,  $\bj_1$ must have been sampled from $\obA \cap \obL_{i}^{(w)}$. Using Lemma~\ref{lem:good-split}, we have that there exists $|\xi_0|, |\xi_1| \leq \sqrt{n} \log n$ so:
	\begin{align*} 
	\Prx_{\boldr \sim \calV_{\text{yes}}'}[\boldr_{1}^{(i)} = 1 \mid v(i) = \text{``lone vertex''}] = \dfrac{|\obA \cap \obL_i^{(w)}|}{|\obA \cap \obL|} = \dfrac{|\obL_i^{(w)}|/2 - \xi_{w}}{|\obL_i|/2 - \xi_0 - \xi_1} \approx \frac{|\obL_i^{(w)}|}{|\obL_i|} \pm O(\tfrac{\log n}{\sqrt{n}}),
	\end{align*}
	where we used the fact that $|\bL_i|, |\bL_i^{(w)}| = \Omega(n)$. If $\bj_2 \in \obA$, then
	\begin{align*} 
	\Prx_{\boldr \sim \calV_{\text{yes}}'}[\boldr_{1}^{(i)} = 1 \mid v(i) = \text{``lone vertex''}] = \dfrac{|\bA \cap \obL_i^{(w)}|}{|\bA \cap \obL|} = \dfrac{|\obL_i^{(w)}|/2 + \xi_{w}}{|\obL_i|/2 + \xi_0 + \xi_1} \approx \frac{|\obL_i^{(w)}|}{|\obL_i|} \pm O(\tfrac{\log n}{\sqrt{n}}).
	\end{align*}
	In both cases, we have that $\boldr_{1}^{(i)} \sim \Ber(p_{v}(\obL_i^{(w)}) \pm O(\tfrac{\log n}{\sqrt{n}}))$, and when we have $(z_1^{(i)})_{m_1} = 1$, we simply flip the answer. Likewise, when $\bj_3 = m_2$, we flip the answer once more.
	
	In the case of $\bG$ being the union of two cliques at $\bA$ and $\obA$, when $\bj_3 = m_1$ and $(z_1^{(i)})_{m_1} = 0$, we have that when $\bj_2 \in \bA$, 
	\begin{align*}
	\Prx_{\boldr \sim \calV_{\text{no}}'}[\boldr_{1}^{(i)} = 1 \mid v(i) = \text{``lone vertex''}] = \dfrac{|\bA \cap \obL_i^{(w)}|}{|\bA \cap \obL|} = \dfrac{|\obL_i^{(w)}|/2 + \xi_{w}}{|\obL_i|/2 + \xi_0 + \xi_1} \approx \frac{|\obL_i^{(w)}|}{|\obL_i|} \pm O(\tfrac{\log n}{\sqrt{n}}),
	\end{align*}
	and when $\bj_2 \in \obA$, 
	\begin{align*}
	\Prx_{\boldr \sim \calV_{\text{no}}'}[\boldr_{1}^{(i)} = 1 \mid v(i) = \text{``lone vertex''}] = \dfrac{|\obA \cap \obL_i^{(w)}|}{|\obA \cap \obL|} = \dfrac{|\obL_i^{(w)}|/2 - \xi_{w}}{|\obL_i|/2 - \xi_0 - \xi_1} \approx \frac{|\obL_i^{(w)}|}{|\obL_i|} \pm O(\tfrac{\log n}{\sqrt{n}}),
	\end{align*}
	so we obtain the analogous conclusion.
\end{proof}

We note that after defining $\boldr_1^{(i)}$ in the cases with $v(i) = \text{``empty set''}$, we have that all values $\boldr_{\alpha}^{(i)}$ are determined by flipping the answer when $(z_{\alpha}^{(i)})_{\bj_3} \neq (z_{1}^{(i)})_{\bj_3}$. Likewise, after defining $\boldr_1^{(i)}$ in the cases with $v(i) = \text{``lone vertex''}$, we have that all values $\boldr_{\alpha}^{(i)}$ are determined by flipping the answer when $(z_{\alpha}^{(i)})_{\bj_3} \neq (z_1^{(i)})_{\bj_3}$ and when $(z_{\alpha}^{(i)})_{\bj_2} \neq (z_{1}^{(i)})_{\bj_2}$. 

Finally, consider the indices $i \in [t]$ of responses $\boldr_\alpha^{(i)}$ with $v(i) = \text{``empty set''}$, and call these $E$. We have that for all $i \in E$, $\calU_{\text{yes}}'$ and $\calU_{\text{no}}'$ outputs bits which equal 1 with probability $\tau_i$ where $\tau_i = \Omega(1)$, and $\calV_{\text{yes}}'$ and $\calV_{\text{no}}'$ outputs bits which equal 1 with probability $\tau_i \pm O(\tfrac{\log^2 n}{n})$. Since these groups are independent and there at at most $q \ll n^{1.5}$ groups, we have that the bits $(\boldr_1^{(i)})_{i\in E} \sim \calU_{\text{yes}}'$ (and also $\calU_{\text{no}}'$) satisfy:
\[ (\boldr_1^{(i)})_{i \in E} \sim \prod_{i \in E} \Ber(\tau_i), \]
and for each $i \in E$, there exists $\gamma_{i,\text{yes}}$ and $\gamma_{i, \text{no}}$ with $|\gamma_{i, \text{yes}}|, |\gamma_{i, \text{no}}| = O(\frac{\log^2 n}{n})$ such that $(\boldr_1^{(i)})_{i \in E} \sim \calV_{\text{yes}}'$ satisfies
\[ (\boldr_1^{(i)})_{i \in E} \sim \prod_{i \in E} \Ber(\tau_i + \gamma_{i, \text{yes}}), \]
and if $(\boldr_1^{(i)})_{i \in E} \sim \calV_{\text{no}}'$ satisfies
\[ (\boldr_1^{(i)})_{i \in E} \sim \prod_{i \in E} \Ber(\tau_i + \gamma_{i, \text{no}}). \]
Thus, by \cite{Roos01}, we have that the distance in total variation between these two distributions is at most $o(1)$. 

Similarly, we consider the indices $i \in [t]$ with $v(i) = \text{``lone vertex''}$, and call these $V$. By Lemma~\ref{lem:few-lone}, we have that $|V| \leq \frac{n}{\log^4 n}$ with probability $1 - o(1)$ if the cost of the rejection sampling algorithm is less than $\frac{n^2}{\log^{6} n}$. So similarly to the case with the groups in $E$, these can only incur at most $o(1)$ in distance in total variation.

\section{A lower bound for distinguishing $\calG_1$ and $\calG_2$ with rejection samples}\label{sec:LowerBound}

In this section, we derive a lower bound for distinguishing $\calG_1$ and $\calG_2$ with rejection samples. 

\begin{lemma}\label{lem:lb}
	Any deterministic non-adaptive algorithm $\Alg$ with $\cost(\Alg) \leq \frac{n^{2}}{\log^6 n}$, has:
	\[ \Prx_{\bG \sim \calG_1}[\Alg \text{ outputs ``$\calG_1$''}] \leq (1 + o(1)) \Prx_{\bG \sim \calG_2}[\Alg\text{ outputs ``$\calG_1$''}] + o(1). \]
\end{lemma}

We assume $\Alg$ is a deterministic non-adaptive algorithm with $\cost(\Alg) \leq \frac{n^2}{\log^6 n}$. $\Alg$ makes queries $L_1, \dots, L_t \subset [n]$ and the oracle returns $\bv_1, \dots, \bv_t$, some of which are edges, some are lone vertices, and some are $\emptyset$. Let $\bG_o \subset \bG$ be the graph observed by the algorithm by considering all edges in $\bv_1, \dots, \bv_t$. We let $|\bG_o|$ be the number of edges.

Before going on to prove the lower bound, we use the following simplification. First, we assume that any algorithm $\Alg$ has all its queries $L_1, \dots, L_t$ satisfying that either $|L_i| \leq \frac{n}{\log n}$, or $L_i = [n]$. Thus, it suffices to show for this restricted class of algorithms, the cost must be at least $\frac{n^2}{\log^5 n}$.

\subsection{High Level Overview}

In this subsection, we will give a high level overview of the proof of Lemma~\ref{lem:lb}.

The idea is that we will argue outcome-by-outcome; i.e., we consider the possible ways the algorithm can act, which depends on the responses to the queries the algorithm gets. Consider some responses $v_1, \dots, v_t \in [n] \cup \left([n]\times[n]\right) \cup \{\emptyset\}$, where each $v_i$ may be either a lone vertex, an edge, or $\emptyset$. Suppose that upon observing this outcome, the algorithm outputs ``$\calG_1$''. There will be two cases:
\begin{itemize}
	\item The first case is when the probability of observing this outcome from $\calG_2$ is not too much lower than the probability of observing this outcome from $\calG_1$. In these outcomes, we will not get too much advantage in distinguishing $\calG_1$ and $\calG_2$.
	\item The other case is when the probability of observing this outcome from $\calG_2$ is substantially lower than the probability of observing this outcome from $\calG_1$. These cases do help us distinguish between $\calG_1$ and $\calG_2$; thus, we will want to show that collectively, the probability that we observe these outcomes from $\calG_1$ is $o(1)$. 
\end{itemize}
We will be able to characterize the outcomes which fall into the first case and the second case by considering a sequence of events. In particular we define five events which depend on $v_1, \dots, v_t$, as well as the random choice of $\bA$. Consider the outcome $v_1,\dots, v_t$ which together form components $C_1, \dots, C_{\alpha}$. The events are the following\footnote{We note that the first two event are not random and depends on the values $v_1,\ldots,v_t$, and the rest are random variables depending on the partition $\bA$ and the oracle responses $\bv_1,\ldots,\bv_t$. }:
\begin{enumerate}
	\item  $\calE_{T}$ (Observe small trees): this is the event where the values of $v_1, \dots, v_t$ form components $C_1,\dots, C_{\alpha}$ which are all trees of size at most $\log n$.  
	\item $\calE_{F}$ (Observe few non-empty responses): this is the event where the values of $v_1, \dots, v_t$ have at most $\frac{n}{\log^4 n}$ non-$\emptyset$ responses. This event implies that the total number of vertices in the responses $v_1, \dots, v_t$ is at most $\frac{n}{\log^4 n}$. 
	\item $\calbE_{C, \text{yes}}$ and $\calbE_{C,\text{no}}$ (Consistency condition of the components observed): these are the events where $\bA \subset [n]$ partitions the components $C_1, \dots, C_{\alpha}$ in a manner consistent with $\calG_1$ in $\calbE_{C,\text{yes}}$ or $\calG_2$ in $\calbE_{C, \text{no}}$. See Definition~\ref{def:event-cons} for a formal definition of this event. These events are random variables that depend only on $\bA$. It will become clear that in order to observe the outcome $v_1, \dots, v_t$ in $\calG_1$, event $\calbE_{C, \text{yes}}$ must be triggered, and in $\calG_2$, event $\calbE_{C, \text{no}}$ must be triggered. See Figure~\ref{Fig:consistent} for an illustration.
	
	\begin{figure}[ht!] 
		\centering
		\begin{picture}(400, 160)
		\put(10,-170){\includegraphics[width=0.8\linewidth]{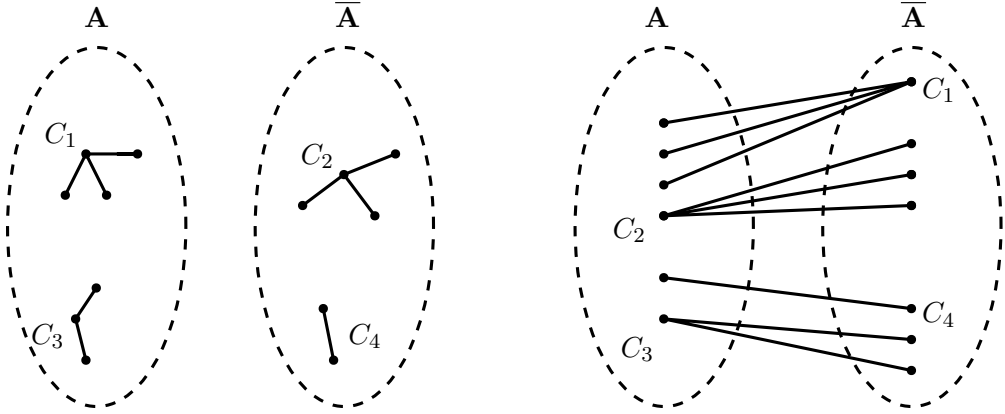}}
		\put(40, 150){$\bA$}
		\put(135, 150){$\obA$}
		\put(252, 150){$\bA$}
		\put(349, 150){$\obA$}
		\put(25, 106){$C_1$}
		\put(20, 30){$C_3$}
		\put(122, 98){$C_2$}
		\put(140, 30){$C_4$}
		\put(357, 123){$C_1$} 
		\put(240, 70){$C_2$}
		\put(243, 25){$C_3$}
		\put(357, 38){$C_4$}
		\end{picture}
		\caption{\label{Fig:consistent} $\bA$ consistently partition of the components $C_1, C_2, C_3$ and $C_4$ according to $\calG_1$ (on the left) and $\calG_2$ (on the right). }
	\end{figure}
	
	\item $\calbE_O$ (Observe specific responses) : this event is over the randomness in $\bA$, as well as the randomness in the responses of the oracle $\bv_1, \dots, \bv_t$. The event is triggered when the responses of the oracle are exactly those dictated by $v_1,\dots, v_t$; i.e., for all $i \in [t]$, $\bv_i = v_i$.
	\item $\calbE_{B}$ (Balanced lone vertices condition) : this event is over the randomness in $\bA$, as well as the responses $\bv_1, \dots, \bv_t$. The event occurs when a particular quantity which depends on $\bA$ and $\bv_1, \dots, \bv_t$ is bounded by some predetermined value. See Definition~\ref{def:event-l} for a formal definition.
\end{enumerate}
\newcommand{\Ysample}{\substack{\bG\sim\calG_1 \\ \bv_1, \dots, \bv_t}}
\newcommand{\Nsample}{\substack{\bG\sim\calG_2 \\ \bv_1, \dots, \bv_t}}
Having defined these events, the lower bound follows by the following three lemmas. The first lemma says that for any outcomes satisfying $\calE_T$ and $\calE_F$, the probability over $\bA$ of being consistent in $\calG_1$ cannot be much higher than in $\calG_2$. The second lemma says that the outcomes satisfying the events described above do not help in distinguishing $\calG_1$ and $\calG_2$. The third lemma says that good outcomes occur with high probability over $\calG_1$.
\begin{lemma}[Consistency Lemma]\label{lem:consistency}
	Consider a fixed $v_1, \dots, v_t \in [n] \cup \left([n]\times[n]\right) \cup \{\emptyset\}$ forming components $C_1,\dots, C_{\alpha}$ where events $\calE_{T}$ and $\calE_F$ are satisfied. Then, we have:
	\[ \Prx_{\Ysample}[\calbE_{C,\text{yes}}] \leq (1 + o(1)) \Prx_{\Nsample}[\calbE_{C,\text{no}}]. \]
\end{lemma}

\begin{lemma}[Good Outcomes Lemma]\label{lem:good}
	Consider a fixed $v_1, \dots ,v_{t} \in [n] \cup\left( [n] \times [n]\right) \cup \{ \emptyset\}$ forming components $C_1,\dots, C_{\alpha}$ where events $\calE_T$ and $\calE_F$ are satisfied. Then, we have:
	\[ \Prx_{\Ysample}[\calbE_O \wedge \calbE_B \mid \calbE_{C,\text{yes}}] \leq (1 + o(1)) \Prx_{\Nsample}[\calbE_O \mid \calbE_{C,\text{no}}]. \]
\end{lemma}

\begin{lemma}[Bad Outcomes Lemma]\label{lem:bad}
	We have that:
	\[ \Prx_{\Ysample}[\neg \calbE_{T} \vee \neg \calbE_{F} \vee \neg \calbE_{B} ] = o(1). \]
\end{lemma}

Assuming the above three lemmas, we may prove Lemma~\ref{lem:lb}.

\begin{proof}
	Let $\Lambda$ be the set of outcomes of the algorithm which output ``$\calG_1$.'' Each outcome is a collection of responses $v_1, \dots, v_t$. We let 
	\[ \Lambda_{G} = \{ \ell \in \Lambda : \text{ responses } v_1, \dots, v_t \text{ satisfy } \calE_{T} \wedge \calE_{F}\},  \]
	and $\calbE_{O,\ell}$ be the event that responses $\bv_1, \dots, \bv_t$ result in outcome $\ell$.
	We have:
	\begin{align*}
	\Prx_{\Ysample}[\Alg \text{ outputs ``$\calG_1$''}] &\leq \sum_{\ell \in \Lambda_G}  \Prx_{\Ysample}[\ell \text{ is observed by } \Alg \mid \calbE_{C,\text{yes}}] \Prx_{\Ysample}[\calbE_{C, \text{yes}}] + \Prx_{\Ysample}[\neg \calbE_{T} \vee \neg \calbE_{F}] \\
	&\leq \sum_{\ell \in \Lambda_{G}}  \Prx_{\Ysample}\left[ \calbE_{O,\ell} \wedge \calbE_B \mid \calbE_{C,\text{yes}}\right] \Prx_{\Ysample}[\calbE_{C, \text{yes}}] + \Prx_{\Ysample}[\neg \calbE_{T} \vee \neg\calbE_{F} \vee \neg\calbE_{B}] \\
	&\leq (1 + o(1)) \sum_{\ell \in \Lambda_{G}} \Prx_{\Nsample}\left[\calbE_{O,\ell} \mid \calbE_{C,\text{no}} \right] \Prx_{\Nsample}[\calbE_{C, \text{no}}] + o(1) \\
	&\leq (1 + o(1)) \Prx_{\Nsample}[\Alg\text{ outputs ``$\calG_1$''}] + o(1),
	\end{align*}
	where we used Lemma~\ref{lem:consistency}, Lemma~\ref{lem:good}, and Lemma~\ref{lem:bad} from the second to third line. 
\end{proof}

\subsection{Proof of the Consistency Lemma: Lemma~\ref{lem:consistency}}

We now turn to proving Lemma~\ref{lem:consistency}. We first give the formal definitions of events $\calbE_{C,\text{yes}}$ and $\calbE_{C,\text{no}}$. Next, we set up some definitions necessary for the proof and give two claims which imply the lemma. For the remainder of the section, we consider fixing the responses $v_1, \dots, v_{t}\in [n] \cup \left([n]\times[n]\right)\cup \{ \emptyset\}$. We assume the responses form the components $C_1, \dots, C_{\alpha}$ which satisfy events $\calE_{T}$ and $\calE_{F}$. For each $i \in [\alpha]$, let $u_i$ be the minimum vertex in $C_i$ with respect to the natural ordering of $[n]$, and consider rooting the trees $C_i$ at $u_i$, forming a layered tree with at most $\log n$ layers. Namely, $u_i$ will be in the first layer, all its neighbors in $C_i$ will be in the second layer, and so on.  We let $C_i(\text{even})$ be the set of vertices in even layers, and $C_i(\text{odd})$ be the set of vertices in odd layers.

\begin{definition}\label{def:event-cons}
	We let $\calbE_{C, \text{yes}}$ be the event that $\bA \subset [n]$ is consistent with the observations $v_1, \dots, v_{t}$ when $\bG = K_{\bA} \cup K_{\obA}$, and $\calbE_{C, \text{no}}$ be the event that $\bA \subset [n]$ is consistent with the observations $v_1, \dots, v_t$ when $\bG = K_{\bA, \obA}$. In other words,
	\begin{itemize}
		\item In $\calbE_{C, \text{yes}}$: for all $i \in [\alpha]$, either $C_{i} \subset \bA$ or $C_i \subset \obA$. 
		\item In $\calbE_{C, \text{no}}$: for all $i \in [\alpha]$, either $C_{i}(\text{odd}) \subset \bA$ and $C_{i}(\text{even}) \subset \obA$, or $C_{i}(\text{odd}) \subset \obA$ and $C_i(\text{even}) \subset \bA$.
	\end{itemize}
\end{definition}

For each $i \in [\alpha]$, let $\bY_i$ be the indicator random variable for $u_i \in \bA$. Let:
\[ \bW_{A, \text{yes}} = \sum_{i=1}^{\alpha} \bY_i\cdot |C_i| \qquad \bW_{A, \text{no}} = \sum_{i=1}^{\alpha} \left(\bY_i \cdot |C_i(\text{odd})| + (1 - \bY_i) \cdot |C_i(\text{even})|\right) \qquad V = \sum_{i=1}^{\alpha} |C_i|.\] 

\begin{definition}
	We let $\calbE_{W}$ be the event where:
	\[ \frac{V}{2} - \sqrt{V} \log n \leq \bW_{A, \text{no}} \leq \frac{V}{2} + \sqrt{V}\log n. \]
\end{definition}

Lemma~\ref{lem:consistency} follows from the next two claims.

\begin{claim}\label{cl:when_ew}
	For $v_1, \dots, v_{t}$ satisfying events $\calE_T$ and $\calE_F$, we have:
	\[ \Prx_{\substack{\bG \sim \calG_1 \\ \bv_1, \dots, \bv_t}}[\calbE_{C, \text{yes}} \wedge \calbE_{W}] \leq (1 + o(1)) \Prx_{\substack{\bG \sim \calG_2 \\ \bv_1, \dots, \bv_t}}[\calbE_{C, \text{no}}]. \]
\end{claim}

\begin{claim}\label{cl:no_ew}
	For $v_1, \dots, v_{t}$ satisfying events $\calE_T$ and $\calE_F$, we have:
	\[ \Prx_{\substack{\bG\sim \calG_1 \\ \bv_1, \dots ,\bv_t}}[\neg \calbE_{W} \mid \calbE_{C, \text{yes}}] = o(1). \]
\end{claim}

Given Claim~\ref{cl:when_ew} and Claim~\ref{cl:no_ew}, we proceed to proving Lemma~\ref{lem:consistency}.

\begin{proofof}{Lemma~\ref{lem:consistency}}
	We simply compute the respective probabilities.
	\begin{align}
	\Prx_{\Ysample}[\calbE_{C,\text{yes}}] &= \Prx_{\Ysample}[\calbE_{C, \text{yes}} \wedge \calbE_{W}] + \Prx_{\Ysample}[\neg \calbE_W \mid \calbE_{C, \text{yes}}] \Prx_{\Ysample}[\calbE_{C, \text{yes}}] \nonumber \\
	&\leq (1 + o(1)) \Prx_{\Nsample}[\calbE_{C,\text{no}}] + o(1) \Prx_{\Ysample}[\calbE_{C,\text{yes}}] \label{eq:hehe},
	\end{align}
	Where we applied both Claim~\ref{cl:when_ew} and Claim~\ref{cl:no_ew} in Line~(\ref{eq:hehe}). Finally, this implies:
	\[ (1 - o(1)) \Prx_{\Ysample}[\calbE_{C,\text{yes}}] \leq (1+o(1)) \Prx_{\Nsample}[\calbE_{C,\text{no}}], \]
	which finishes the proof.
\end{proofof}

We now proceed to proving Claim~\ref{cl:when_ew}, followed by the proof of Claim~\ref{cl:no_ew}.

\begin{proofof}{Claim~\ref{cl:when_ew}}
	Note that $V \leq \frac{n}{\log^4 n}$ since event $\calE_F$ is satisfied. Let $y \in \{0, 1\}^\alpha$ be an assignment of $u_1, \dots ,u_{\alpha}$ to $\bA$; more formally, for a fixed $y \in \{0, 1\}^{\alpha}$, we let $\calbE_{y}$ be the event that for each $i \in [\alpha]$, $u_i \in \bA$ if $y_i = 1$, and $u_{i} \in \obA$ if $y_i = 0$. Additionally, let 
	\[ Y_G = \{ y \in \{0, 1\}^{\alpha} : \text{ if $\bA$ satisfies $\calbE_y$, then $\calbE_W$ is satisfied} \}. \]
	Then,
	\[ \Prx_{\Ysample}[\calbE_{C,\text{yes}} \wedge \calbE_{W}] = \sum_{y \in Y_{G}} \Prx_{\Ysample}[\calbE_{C,\text{yes}} \wedge \calbE_y]. \]
	It suffices to show that for $y \in Y_G$:
	\[ \Prx_{\Ysample}[\calbE_{C,\text{yes}} \wedge \calbE_y ] \leq (1 + o(1)) \Prx_{\Nsample}[\calbE_{C,\text{no}} \wedge \calbE_y]. \]
	Note that if $\bA$ satisfies $\calbE_y$ and $\calbE_{C,\text{yes}}$ is satisfied, there is precisely one choice for assigning each vertex in $C_1, \dots, C_{\alpha}$ to $\bA$ or $\obA$. Likewise, if $\bA$ satisfied $\calbE_y$ and $\calbE_{C,\text{no}}$, there is precisely one choice for assigning each vertex in $C_1, \dots, C_{\alpha}$ to $\bA$ or $\obA$. The remaining vertices may be placed in $\bA$ or $\obA$ so the resulting set $\bA$ contains half of all vertices, therefore, we have:
	\[ \Prx_{\Ysample}[\calbE_{C,\text{yes}} \wedge \calbE_y] \leq \dfrac{\binom{n-V}{\frac{n}{2} - \frac{V}{2}}}{\binom{n}{n/2}} \qquad \Prx_{\Nsample}[\calbE_{C,\text{no}} \wedge \calbE_y] \geq \dfrac{\binom{n-V}{\frac{n}{2} - \frac{V}{2} - \sqrt{V} \log n}}{\binom{n}{n/2}}. \]
	Taking the ratio, we have:
	\begin{align*}
	\dfrac{\Prx[\calbE_{C, \text{yes}} \wedge \calbE_y]}{\Prx[\calbE_{C,\text{no}} \wedge \calbE_y]} &\leq \dfrac{\binom{n - V}{\frac{n}{2} - \frac{V}{2}}}{\binom{n}{n/2}} \cdot \dfrac{\binom{n}{n/2}}{\binom{n - V}{\frac{n}{2} - \frac{V}{2} - \sqrt{V}\log n}} \leq \left( \dfrac{\frac{n}{2} - \frac{V}{2} + \sqrt{V} \log n}{\frac{n}{2} - \frac{V}{2} - \sqrt{V}\log n}\right)^{\sqrt{V}\log n} \\
	&\leq \left( 1 + O\left(\frac{1}{\sqrt{n} \log n}\right)\right)^{\sqrt{n}/\log n} = 1 + o(1). 
	\end{align*}
\end{proofof}

\begin{proofof}{Claim~\ref{cl:no_ew}}
	We let:
	\[ \bW_{A,\text{no}}^{(O)} = \sum_{i=1}^{\alpha} \bY_i \cdot |C_i(\text{odd})| \qquad \text{and} \qquad \bW_{A,\text{no}}^{(E)} = \sum_{i=1}^{\alpha} (1 - \bY_i) \cdot |C_i(\text{even})|. \]
	where $\bW_{A,\text{no}}^{(O)} + \bW_{A,\text{no}}^{(E)} = \bW_{A, \text{no}}$ specifies the number of vertices in $\cup_{i\in[\alpha]} C_i$ assigned to $\bA$. Conditioning on event $\calbE_{C,\text{yes}}$, $\bA$ and $\obA$ can be interchanged, so
	\[ \Prx_{\Ysample}[\bY_i = 1 \mid \calbE_{C,\text{yes}} ] = \Prx_{\Ysample}[\bY_i = 0 \mid \calbE_{C,\text{yes}}] = \frac{1}{2}.\]
	So, 
	\[ \Ex_{\Ysample }[\bW^O_{A,\text{no}}\mid \calbE_{C,\text{yes}}] =\frac{1}{2}\sum_{i\in[\alpha]}|C_i(\text{odd})| \qquad\text{and}\qquad \Ex_{\Ysample }[\bW^E_{A,\text{no}}\mid \calbE_{C,\text{yes}}]=\frac{1}{2}\sum_{i\in[\alpha]}|C_i(\text{even})|.\]
	Additionally, for any set of indices $I \subset [\alpha]$, 
	\[ \Prx_{\Ysample}[\forall i \in I, \bY_i = 1 \mid \calbE_{C,\text{yes}}] \leq \frac{1}{2^{|I|}} \qquad \text{and}\qquad \Prx_{\Ysample}[\forall i \in I, \bY_i = 0 \mid \calbE_{C, \text{no}} ] \leq \frac{1}{2^{|I|}},\]
	which implies that the variables $\bY_i$, as well as the variables in $1 -\bY_i$ are negatively correlated. We may apply Chernoff bounds (for negatively correlated variables) to obtain deviation bounds for $\bW_{A,\text{no}}^{(O)}$ and $\bW_{A,\text{no}}^{(E)}$. Then, a union bound gives the desired result for $\bW_{A,\text{no}}$.
\end{proofof}

\subsection{Proof of the Bad Outcomes Lemma: Lemma~\ref{lem:bad}}

In this section, we give a proof of Lemma~\ref{lem:bad}, which says that the probability over $\bG \sim \calG_1$ and $\bv_1, \dots ,\bv_t$ of not satisfying events $\calbE_{T}$, $\calbE_{F}$, as well as $\calbE_{B}$ is $o(1)$. In order to prove this, we will show that individually, the probability of not satisfying each event is $o(1)$ and conclude with a union bound.

\subsubsection{$\calbE_{T}$: components observed are small trees}

The goal of this section is to show that with high probability, the algorithm only sees edges which form various components of small trees. 

\begin{definition}
	We let $\calbE_{T}$ be the event that observed responses $\bv_1, \dots, \bv_t$ generate components $\bC_1, \dots, \bC_{\alpha}$ which are all trees of size less than $\log n$. 
\end{definition}

\begin{lemma}\label{lem:calE_T}
	We have that:
	\[ \Prx_{\substack{\bG\sim\calG_1 \\ \bv_1, \dots, \bv_t}}[\calbE_T] \geq 1 - o(1). \]
\end{lemma}

We prove the above lemma by showing the following two claims.

\begin{claim}\label{lem:no-cycles}
	With probability $1 - o(1)$ over the draw of $\bG \sim \calG_1$ and the draw of $\bv_1, \dots, \bv_t$, $\bG_o$ has no cycles.
\end{claim}

\begin{proof}
	Recall that $L_1, \dots, L_t$ are the set queries made, and let $\calbE_{\circ, \ell}$ be the event that $\bG_o$ has a cycle of length $\ell$. We have:
	\begin{align}
	\Prx_{\substack{\bG \sim \calG_1 \\ \bv_1,\dots, \bv_t}}[\calbE_{\circ, \ell}] &\leq\sum_{\substack{S \subset [t] \\ S = \{ i_1, \dots, i_\ell\}}} \Prx_{\substack{\bG \sim \calG_1 \\ \bv_1,\dots, \bv_t}}[\bv_{i_1}, \dots, \bv_{i_\ell} \text{ form cycle}] \nonumber \\
	&\leq \sum_{\substack{S \subset [t] \\ S = \{ i_1, \dots, i_\ell\}}} \sum_{\substack{U \subset [n] \\ U = \{ u_1, \dots, u_{\ell}\} \\
			u_j \in L_{i_j} \cap L_{i_{j+1}}}} \Prx_{\substack{\bG \sim \calG_1 \\ \bv_1, \dots, \bv_t}}[\forall j \in [\ell], \bv_{i_j} = (u_{j}, u_{j+1})], \label{eq:haha}
	\end{align}
	where we think $j+1 = 1$ when $j = \ell$. The above restriction of $u_j \in L_{i_j} \cap L_{i_{j+1}}$ is necessary if edges $\bv_{i_j}$ and $\bv_{i_{j+1}}$ will be the edges of the cycle incident on node $u_j$. Additionally, we may upper bound (\ref{eq:haha}) by disregarding the effect of the partition $\bA$ and $\obA$; in fact, the presence of $\bA$ and $\obA$ make it harder to achieve a cycle, since if $u_j \in \bA$ and $u_{j+1} \in \obA$, the probability of $\bv_{i_j} = (u_j, u_{j+1})$ is 0. For any $S = \{ i_1, \dots, i_{\ell} \}$, once we fix a set $U = \{ u_1, \dots, u_{\ell} \}$ where $u_j \in L_{i_j} \cap L_{i_j+1}$, 
	\[ \Prx_{\Ysample}[\forall j \in [\ell], \bv_{i_j} = (u_j, u_{j+1})] \leq \left( \dfrac{1}{2\binom{n/2}{2}}\right)^{\ell}. \]
	Thus, we have:
	\begin{align*}
	\Prx_{\substack{\bG \sim \calG_1 \\ \bv_1,\dots, \bv_t}}[\calbE_{\circ, \ell}] &\leq \sum_{\substack{S \subset [t] \\ S = \{ i_1, \dots, i_\ell\}}} \left(  \prod_{j=1}^{\ell} |L_{i_{j}} \cap L_{i_{j+1}}| \right)
	\left(\dfrac{1}{2\binom{\frac{n}{2}}{2}}\right)^{\ell} \\
	&\leq \left(\dfrac{1}{\Omega(n)}\right)^{2\ell} \sum_{\substack{S \subset [t] \\ S = \{ i_1, \dots, i_\ell\}}} \prod_{j=1}^{\ell} |L_{i_{j}}| \\
	&\leq \left( \frac{1}{\Omega(n)}\right)^{2\ell} \left( \sum_{i=1}^t |L_i|\right)^\ell \left( \frac{1}{t}\right)^{\ell}\binom{t}{\ell} \leq \left( O\left(\frac{1}{\log^5 n} \right)\right)^{\ell}.
	\end{align*}
	where we used the fact that $\sum_S \prod_{j=1}^{\ell} |L_{i_j}|$ is the elementary symmetric polynomial of degree $\ell$, and $\sum_{i=1}^t |L_{i}| \leq \frac{n^2}{\log^5 n}$. Thus, we obtain:
	\[ \Prx_{\substack{\bG\sim\calG_1 \\ \bv_1, \dots, \bv_t}}[\bG_o \text{ contains a cycle}] \leq \sum_{\ell=1}^{t} \left(  O\left( \frac{1}{\log^5 n}\right)\right)^{\ell} = o(1). \]
\end{proof}

\begin{claim}\label{lem:no-large}
	With probability $1 - o(1)$ over the draw of $\bG \sim \calG_1$ and the draw of $\bv_1, \dots, \bv_t$, we have $\bG_o$ has all components of size at most $\log n$. 
\end{claim}

\begin{proof}
	This proof is very similar to the one above. Let $\calbE_{T, \ell}$ be the event there exists a tree of $\ell$ edges. We note that there are at most $\exp(O(\ell))$ rooted trees of $\ell$ edges and $\ell+1$ vertices. We consider first picking a rooted tree, and we pick an arbitrary vertex to be the root of the tree. We then pick the $\ell$ edges of the tree to some responses, $\bv_{i_1}, \dots, \bv_{i_{\ell}}$. We select the vertex on query of the edge going away from the root; this leaves the root, which we choose arbitrarily from $[n]$.
	
	So we have:
	\begin{align*}
	\Prx_{\substack{\bG\sim\calG_1 \\ \bv_1, \dots, \bv_t} } \left[ \calbE_{T, \ell}\right] &\leq \exp(O(\ell)) \sum_{\substack{S \subset [t] \\ S = \{i_1, \dots, i_{\ell}\}}} \left( n  \prod_{j=1}^{\ell} |L_{i_j}| \right) \left( \dfrac{1}{2\binom{n/2}{2}}\right)^{\ell} \\
	&\leq n \cdot \left( O\left(\frac{1}{\log^5 n}\right) \right)^{\ell} = \left( O\left( \frac{1}{\log^5 n}\right)\right)^{\ell},  
	\end{align*}
	when $\ell \geq \log n$. Thus, we sum over all $\ell \geq \log $ to get that there exists a tree of size $\log n$ or greater with probability $o(1)$. 
\end{proof}

\subsubsection{$\calbE_{F}$: few vertices are observed}

The goal of this section is to show that the algorithm does not observe too many vertices from the responses $\bv_1, \dots \bv_t$ with high probability.
\begin{definition}
	We let $\calbE_{F}$ be the event that the responses $\bv_1, \dots ,\bv_t$ contain at most $\frac{n}{\log^4 n}$ values which are not $\emptyset$.
\end{definition}

\begin{lemma}\label{lem:calE_F}
	We have:
	\[ \Prx_{\substack{\bG \sim \calG_1 \\ \bv_1, \dots, \bv_t}}[\calbE_F] \geq 1 - o(1) \qquad \text{and}\qquad \Prx_{\substack{\bG \sim \calG_2 \\ \bv_1, \dots, \bv_t}}[\calbE_{F}] \geq 1 -o(1). \]
	In other words, any rejection sampling algorithm with cost less than $\frac{n^2}{\log^6 n}$ will observe at most $\frac{n}{\log^4 n}$ non-$\emptyset$ responses in both $\calG_1$ and $\calG_2$ with high probability.
\end{lemma}

\begin{proof}
	Simply note that for a query $L_i$, and any $G \in \calG_1$, the probability of observing a response which is not $\emptyset$ is at most $\dfrac{|L_i| \cdot \frac{n}{2}}{2\binom{n/2}{2}} = O(|L_i|/n)$ (in the case of $\calG_1$, and $\frac{|L_i| \cdot \frac{n}{2}}{n^2/4}$ in the case of $\calG_2$). Therefore, the expected number of responses which are not $\emptyset$ is at most $O(n/\log^5 n)$, and via a Markov bound, we have the desired result. 
\end{proof}

\subsubsection{$\calbE_{B}$: vertices observed do not prefer any side too much}

We now formally define the event $\calbE_{B}$, and prove the event occurs with high probability over the draw of $\bG \sim \calG_1$ and $\bv_1,\dots, \bv_t$. 
\begin{definition}\label{def:event-l}
	Let $\bV_L \subset [t]$ be the random variable corresponding to the set of indices of responses $\bv_1, \dots, \bv_t$ which correspond to observing lone vertices, and for $i \in \bV_L$, we let $\by_i$ be the indicator random variable for $\bv_i \in \bA$. Let $\calbE_{B}$ be the event where:
	\[ \bB = \sum_{i \in \bV_L} (-1)^{\by_i} \left(|L_i \cap \bA| - |L_i \cap \obA| \right) = O\left( \frac{n}{\log n}\right). \]
\end{definition}

We start by giving some intuition. Fix some query $L_i$ such that $|L_i|\le \frac{n}{\log n}$. By using Chernoff bound we have that $||L_i\cap\bA|-|L_i\cap\obA||=O(\sqrt{|L_i|}\log n)$ with high probability. Now assume that \emph{every} query we make is skewed toward $\obA$. This bad event will create a gap in the probabilities to see a lone vertex between the two distributions, and the algorithm might use it in order to distinguish $\calG_1$ and $\calG_2$. 
Hence, we would like to claim that \emph{collectively} the probability of observing such bad events is extremely small. More precise details follows.

\begin{definition}
	Let $\calbE_{Q}$ be the event that all queries $L_1, \dots, L_{t}$ satisfy:
	\[ \left| |L_i \cap \bA| - |L_i \cap \obA| \right| = O\left(\sqrt{|L_i|} \log n \right). \]
\end{definition}

\begin{claim}\label{cl:tbalance}
	We have:
	\[ \Prx_{\substack{\bG\sim\calG_1}}[\calbE_Q] \geq 1- o(1).\]
\end{claim}

\begin{proof}
	This simply follows from a union bound over $2t$ applications of the Chernoff bound for negatively correlated random variables. In particular, for all $k \in [n]$, let $\bY_k$ be the indicator random variable for $k \in \bA$. Then we note that for each $i \in [t]$,
	\[ |L_i \cap \bA| = \sum_{k \in L_i} \bY_k \qquad\text{and}\qquad |L_i \cap \obA| = \sum_{k \in L_i} (1 - \bY_k). \]
	In a similar way to the proof of Claim~\ref{cl:no_ew}, we note that all $\bY_k$ are negatively correlated, and all $(1 - \bY_i)$ are negatively correlated, thus, we have that with probability at least $1 - n^{-10}$, 
	\[ |L_i \cap \bA| \leq \frac{|L_i|}{2} + \sqrt{|L_i|}\log n \qquad\text{and}\qquad |L_i \cap \obA| \leq \frac{|L_i|}{2} + \sqrt{|L_i|}\log n. \]
	Thus, we may union bound over all $2t$ events, for the desired result.
\end{proof}

\begin{lemma}\label{lem:calE_Q_L}
	We have that:
	\[ \Prx_{\Ysample}[\neg \calbE_{B} \wedge \calbE_F] = o(1). \]
\end{lemma}

\begin{proof}
	We first note that because of Claim~\ref{cl:tbalance}, we have:
	\begin{align*}
	\Prx_{\Ysample}[\neg \calbE_{B} \wedge \calbE_F] &= \sum_{\substack{A \subset [n] \\ \calE_Q \text{ satisfied}}} \Prx_{\Ysample}[\bA = A]\Prx_{\Ysample}[\neg \calbE_{B} \wedge \calbE_F \mid \bA = A] + o(1).
	\end{align*}
	So consider a fixed set $A \subset [n]$ which satisfies event $\calE_Q$. Additionally, we have:
	\[ \Prx_{\Ysample}[\neg \calbE_{B} \wedge \calbE_F \mid \bA = A] = \sum_{\substack{V_L \subset [t] \\ |V_L| \leq \frac{n}{\log^4 n}}} \Prx_{\Ysample}[\bV_L = V_L \mid \bA = A]\Prx_{\Ysample}[\neg \calbE_{B} \mid \bA = A, \bV_L = V_L] \]
	Thus, it suffices to prove that for all $A \subset [n]$ which satisfy $\calE_{Q}$ and $V_L \subset [t]$ of size at most $\frac{n}{\log^4 n}$, $\Pr[\neg \calbE_{B} \mid \bA = A, \bV_L = V_L] = o(1)$. In fact, once we condition on $\bA = A$ and $\bV_L = V_L$, we have:
	\[ \bB = \sum_{i \in V_L} (-1)^{\by_i} \left(|L_i \cap A| - |L_i \cap \oA| \right) ,\]
	which is a sum of independent random variables. Additionally, since $\by_i$ is the indicator random variable for $\bv_i \in A$ conditioned on $\bv_i$ being a lone vertex, we have each $\by_i$ is independent and is $1$ with probability $p_i$, where:
	\[ p_i = \dfrac{|L_i \cap A| \left(\frac{n}{2} - |L_i \cap A|\right)}{|L_i|\cdot \frac{n}{2} - |L_i \cap A|^2 - |L_i \cap \oA|^2} = \frac{1}{2} \pm O\left( \frac{\log n}{\sqrt{n}} \right). \]
	Thus, we have:
	\[ \Ex_{\Ysample}[\bB \mid \bA = A, \bV_L = V_L] = |V_L| \cdot O(\log^2 n) = O\left(\frac{n}{\log^2 n}\right). \]
	Additionally, each variable can contribute $O(\sqrt{|L_i|} \log n)$ to the sum, so via a standard Chernoff bound, noting the fact that $\sum_{i \in V_L} |L_i| \log^2 n \leq \frac{n^2}{\log^3 n}$, we have that $\calbE_{B}$ is satisfied with high probability.
\end{proof}

\subsection{Proof of the Good Outcomes Lemma: Lemma~\ref{lem:good}}

We may divide $v_1, \dots, v_t$ into three sets: 1) $V_E$ contain the indices $i \in [t]$ whose responses $v_i$ which are edges, 2) $V_{L}$ contain the indices $i \in [t]$ whose responses $v_i$ are vertices, and 3) $V_T$ contain the indices $i \in [t]$ whose responses $v_i$ are $\emptyset$. We let:
\[ \Prx_{\Ysample}\left[ \calbE_{O} \wedge \calbE_B  \mid \calbE_{C,\text{yes}}\right] = \calY \qquad \Prx_{\Nsample}\left[\calbE_{O}  \mid \calbE_{C,\text{no}} \right] = \calN. \]

We note that for a fixed $\bA$ the values of $\bv_i$ are independent. Therefore, we may write:
\begin{align*}
\calY &= \Ex_{\bA}\left[ \calY_E \cdot \calY_L \cdot \calY_T \cdot \calbE_{B} \mid \calbE_{C, \text{yes}}\right]\qquad \qquad \qquad   &\calN = \Ex_{\bA}\left[ \calN_E \cdot \calN_L \cdot \calN_T  \mid \calbE_{C,\text{no}}\right]  \\
\calY_E &= \prod_{i \in V_{E}} \Prx_{\bv_i}[\bv_i = v_i \mid Y(\bA)]  \qquad  &\calN_E = \prod_{i \in V_{E}} \Prx_{\bv_i}[\bv_i = v_i \mid N(\bA)] \\
\calY_L &= \prod_{i \in V_{L}} \Prx_{\bv_i}[\bv_i = v_i \mid Y(\bA)] \qquad  &\calN_L = \prod_{i \in V_{L}} \Prx_{\bv_i}[\bv_i = v_i \mid N(\bA)] \\
\calY_{T} &= \prod_{i \in V_T} \Prx_{\bv_i}[\bv_i = \emptyset \mid Y(\bA)] \qquad  &\calN_{T} = \prod_{i \in V_T} \Prx_{\bv_i}[\bv_i = \emptyset \mid N(\bA)]
\end{align*}
where we slightly abused notation to let $\Prx_{\bv_i}[\bv_i = v_i \mid Y(\bA)]$ denote the probability that the sampled response $\bv_i$ is $v_i$ conditioned on the graph $\bG$ being from $\calG_1$ with partition $\bA$; i.e., $\bG = K_{\bA} \cup K_{\obA}$. Likewise, $\Prx_{\bv_i}[\bv_i = v_i \mid N(\bA)]$ denotes the probability that the sampled response $\bv_i$ is $v_i$ conditioned on the graph $\bG$ being from $\calG_2$ with partition $\bA$; i.e., $\bG = K_{\bA, \obA}$. We now simply go through the three products in to show each is at most $1 + o(1)$. 
We shall prove the following claims:
\begin{claim} \label{clm:Y_E}For any $\bA$ for which $\calbE_{C,\text{yes}}$ occurs, we have $\calY_E\le(1+o(1))\calN_E$.
	
\end{claim}

\begin{proof}
	Note that for any choice of $A$ for which $\calbE_{C,\text{yes}}$ occurs, since the $v_i$'s are specific edges:
	\[ \Prx_{\bv_i}[\bv_i = v_i \mid Y(A)] = \frac{1}{2\binom{n/2}{2}} \] 
	and for any choice of $A$ for which $\calbE_{C,\text{no}}$ occurs,
	\[ \Prx_{\bv_i}[\bv_i = v_i \mid N(A)] = \frac{1}{(n/2)^2}. \]
	Thus, 
	\[ \frac{\Prx_{\bv_i}[\bv_i = v_i \mid Y(A)]}{\Prx_{\bv_i}[\bv_i = v_i \mid N(A)]}=\dfrac{n^2}{4} \cdot \dfrac{4}{n^2 - 2n} = 1 + O\left(\frac{1}{n}\right), \]
	and since $|V_E| \leq \frac{n}{\log^4 n}$, we get that $\dfrac{\calY_E}{\calN_E} = 1 + o(1)$.
\end{proof}

\begin{claim} \label{clm:Y_T}For any $\bA$ for which $\calbE_{C,\text{yes}}$ occurs, we have $\calY_T\le\calN_T$.
\end{claim}
\begin{proof}
	Here, we have that for any set $A$ which satisfies $\calbE_{C, \text{yes}}$, we have
	\[ \Prx_{\bv_i}[\bv_i = \emptyset \mid Y(A)] = \dfrac{2\binom{n/2}{2} - |L_i|\frac{n}{2}}{2\binom{n/2}{2}} = 1 - \frac{2|L_i|}{n-2} \]
	and similarly, for any set $A$ which satisfies $\calbE_{C,\text{no}}$, we have
	\[ \Prx_{\bv_i}[\bv_i = \emptyset \mid N(A)] =\dfrac{(n/2)^2 - |L_i|\frac{n}{2}+|A\cap L_i||\oA\cap L_i|}{(n/2)^2} \ge 1 - \frac{2|L_i|}{n}.\]
	Which finishes the proof.
\end{proof}

Thus, by Claims~\ref{clm:Y_E} and~\ref{clm:Y_T} we have:
\begin{align*} 
\dfrac{\Ex_{\bA}\left[ \calY_E \cdot \calY_L \cdot \calY_T \cdot \calbE_{B} \mid \calbE_{C, \text{yes}}\right]}{\Ex_{\bA}\left[ \calN_E \cdot \calN_L \cdot \calN_T \mid \calbE_{C,\text{no}}\right]} 
&\leq (1 + o(1))  \dfrac{\Ex_{\bA}[ \calY_L \cdot \calbE_{B}  \mid \calbE_{C, \text{yes}}]}{\Ex_{\bA}[\calN_L \mid \calbE_{C, \text{no}}]}.
\end{align*}
Therefore, it suffices to prove the following:
\[\dfrac{\Ex_{\bA}[ \calY_L \cdot \calbE_{B}  \mid \calbE_{C, \text{yes}}]}{\Ex_{\bA}[\calN_L \mid \calbE_{C, \text{no}}]}\le 1+o(1).\]

Suppose $\bA \subset [n]$ satisfies $\calbE_{C, \text{yes}}$, then if $v_i$ is a vertex response at query $L_i$. We have:
\begin{align*}
\Prx_{\bv_i}[\bv_i = v_i \mid Y(A)] &= \frac{2}{n-2}\left(1 - \frac{|L_i|}{n} + (-1)^{\bY_i} \left(\dfrac{|L_i \cap \bA| - |L_i \cap \obA|}{n} \right) \right)\\
&= \frac{2}{n-2} \left(1 - \frac{|L_i|}{n} \right) \left(1 + \bZ_i\right),
\end{align*}
where:
\[ \bZ_i =  c_i (-1)^{\bY_i} \left( \dfrac{|L_i \cap \bA| - |L_i \cap \obA|}{n} \right), \]
where $c_i = \dfrac{1}{1 - |L_i|/n} \leq 1 + o(1)$, since $|L_i| \ll \frac{n}{\log n}$, and $\bY_i$ is the indicator random variable for $v_i \in \bA$. Thus, we may simplify:
\begin{align*}
\Ex_{\bA}[\calY_L \cdot \calbE_{B}  \mid \calbE_{C,\text{yes}}] &= \left( \frac{2}{n-2}\right)^{|V_L|} \left( 1 - \frac{|L_i|}{n}\right)^{|V_L|} \Ex_{\bA} \left[ \calbE_{B} \prod_{i \in V_L} (1 + \bZ_i) \mid \calbE_{C,\text{yes}}\right]. 
\end{align*}
Similarly, suppose $\bA \subset [n]$ satisfies $\calbE_{C, \text{no}}$, then if $v_i$ is a vertex response at query $L_i$, we have:
\begin{align*}
\Prx_{\bv_i}[\bv_i = v_i \mid N(A)] &= \dfrac{2}{n} \left(1 - \frac{|L_i|}{n} \right) \left(1 + \bS_i \right),
\end{align*}
where we let $\bS_i$ be the random variable:
\[ \bS_i =  c_i (-1)^{\bY_i} \left( \dfrac{|L_i \cap \obA| - |L_i \cap \bA|}{n} \right), \]
Therefore, we have:
\[ \Ex_{\bA}[\calN_{L} \mid \calbE_{C,\text{no}}] = \left( \frac{2}{n}\right)^{|V_L|} \left(1 - \frac{|L_i|}{n} \right)^{|V_L|} \Ex_{\bA}\left[\prod_{i \in V_L} (1 + \bS_i) \mid \calbE_{C, \text{no}} \right]. \]

We note that since $|V_L|\le \frac{n}{\log^4 n }$, we finish off the proof with the following two claims.
\begin{claim}\label{clm:LoneUpper} \[\Ex_{\bA} \left[ \calbE_{B} \prod_{i \in V_L} (1 + \bZ_i) \mid \calbE_{C,\text{yes}}\right]\le 1+o(1). \]
\end{claim}
\begin{claim}\label{clm:LoneLower}
	\[ \Ex_{\bA}\left[ \prod_{i \in V_L} (1 + \bS_i) \mid \calbE_{C, \text{no}}\right] \geq 1 - o(1) \]
\end{claim}

\begin{proofof} {Claim~\ref{clm:LoneUpper}}
	\begin{align*}
	\Ex_{\bA} \left[ \calbE_{B}  \prod_{i \in V_L} (1 + \bZ_i) \mid \calbE_{C,\text{yes}}\right] &\leq \Ex_{\bA} \left[ \calbE_{B}  \cdot e^{\sum_{i \in V_L} \bZ_i} \mid \calbE_{C, \text{yes}}\right] \\
	&\leq e^{\frac{1}{\log n}} = 1 + o(1).
	\end{align*}
	Where the last inequality follows from the fact that $\calbE_{B} $ occurs.
\end{proofof}

\begin{proofof}{Claim~\ref{clm:LoneLower}}
	Recall that \[ \bS_i =  c_i (-1)^{\bY_i} \left( \dfrac{|L_i \cap \obA| - |L_i \cap \bA|}{n} \right), \]
	therefore, by Chernoff bound (for negative correlations) we have that with probability at least $1-\frac{1}{n^{10}}$, $|\bS_i|\le O\left(\frac{\log n}{\sqrt{n}}\right)$. 
	We let $\bS_i'$ be the random variable which is equal to $\bS_i$ when $|\bS_i| \leq O(\frac{\log n}{\sqrt{n}})$ and $-2n$ otherwise. Via a very similar analysis to Claim~A.1 from \cite{CWX17}, we have:
	\[ \Ex_{\bA}\left[ \prod_{i \in V_L} (1 + \bS_i) \mid \calbE_{C, \text{no}}\right] \geq (1 - o(1)) \left(1 + \sum_{i\in V_L}\Ex_{\bA}[\bS_i' \mid \calbE_{C, \text{no}}] \right). \]
	We now evaluate each $\Ex_{\bA}[\bS_i' \mid \calbE_{C, \text{no}}]$  for $i \in V_L$ individually. We have:
	\begin{align*}
	\Ex_{\bA}[\bS_i' \mid \calbE_{C, \text{no}}] &\ge\Ex_{\bA}\left[\bS_i\mid  \calbE_{C, \text{no}}\right]+(-2n-c_i) \Prx_{\bA}\left[|\bS_i|>O\left(\frac{\log n}{\sqrt{n}}\right)\mid \calbE_{C, \text{no}} \right]\\
	&\ge \Ex_{\bA}\left[\bS_i\mid  \calbE_{C, \text{no}}\right] -O\left(\frac{1}{n^9}\right).
	\end{align*}
	Assume that $v_i$ is in component $C_j$, and note that since $\bA$ and $\obA$ are inter-changeable, 
	\[ \Prx_{\bA}[v_i \in \bA \mid \calbE_{C, \text{no}}] = \Prx_{\bA}[ v_i \in \obA \mid \calbE_{C, \text{no}}] = \frac{1}{2}.\]
	Now we have that,
	\begin{align*}
	\Ex_{\bA}\left[\bS_i\mid  \calbE_{A, \text{no}}\right] &\geq \frac{c_i}{n} \sum_{k\in L_i\setminus C_j}\Ex_{\bA}\left[(-1)^{\bY_i}(-1)^{\bY_k}\mid  \calbE_{C, \text{no}}\right]-O\left(\frac{\log n}{n}\right)\\
	&= \frac{c_i}{n} \sum_{k\in L_i\setminus C_j}\left(2\Prx_{\bA}\left[\bY_k=1\mid \bY_i=1; \calbE_{C, \text{no}}\right]-1\right)-O\left(\frac{\log n}{n}\right),\\
	\end{align*}
	where we used the fact that $|C_i| \leq \log n$, as well as the fact that $\bA$ and $\obA$ are interchangeable. Since $|V_L| \leq \frac{n}{\log^4 n}$ and $|L_i| \leq \frac{n}{\log n}$ for each $i \in V_{L}$ (otherwise, we would have observed an edge), it suffices to prove that $\Prx_{\bA}[\bY_k=1\mid \bY_i=1;\calbE_{C, \text{no}} ]\ge \frac{1}{2}-\frac{\log^4 n}{n}$. This is indeed true, since $\sum_{i=1}^{\alpha} |C_i| \leq \frac{n}{\log^4 n}$ and $|C_i| \leq \log n$ (see Lemma~\ref{lem:split}). 
\end{proofof}

Putting everything together, we have:
\[ \dfrac{\Ex_{\bA}[\calY_L \cdot \calbE_{B} \cdot \calbE_{Q} \mid \calbE_{C,\text{yes}}]}{\Ex_{\bA}[\calN_{L} \mid \calbE_{C, \text{no}}]} \leq \left(\dfrac{n}{n-2}\right)^{|V_L|} \dfrac{1 + o(1)}{1 - o(1)} \leq 1 + o(1). \]

\section*{Acknowledgments}

We thank Eric Blais, Rocco Servedio and Xi Chen for countless discussions and suggestions. We also thank Cl\'ement Canonne, Nathan Harms, Dor Minzer and Sofya Raskhodnikova for useful comments on an earlier version of this manuscript. This work is supported in part by the NSF Graduate Research Fellowship under Grant No. DGE-16-44869, CCF-1703925, CCF-1563155,  CCF-1420349 and the David R. Cheriton Graduate Scholarship.

\begin{flushleft}
\bibliographystyle{alpha}
\bibliography{ms}
\end{flushleft}

\appendix

\section{A Useful Claim}

Consider any set of trees $C_1, \dots, C_{\alpha} \subset [n]$ with roots $u_1, \dots, u_{\alpha}$ satisfying the following conditions:
\begin{itemize}
	\item Each $|C_i| \leq \log n$ for $i \in [\alpha]$,
	\item We have $\sum_{i=1}^\alpha |C_i| \leq \frac{n}{\log^4 n}$.
\end{itemize}
Recall that $\calbE_{C, \text{no}}$ is the event that the components $C_1, \dots, C_{\alpha}$ is consistent with the partition $\bA \subset [n]$. More formally, for each $i \in [\alpha]$, we consider layering the tree $C_i$ with root $u_i$. We let $|C_i(\text{odd})|$ be the odd layers and $|C_i(\text{even})|$ be the even layers. Then, we have event $\calbE_{C, \text{no}}$ is satisfied if for each $i \in [\alpha]$, either $C_i(\text{odd}) \subset \bA$ and $C_i(\text{even}) \subset \obA$ or $C_i(\text{even}) \subset \bA$ and $C_i(\text{odd}) \subset \obA$. 

The following lemma is the last necessary step of Claim~\ref{clm:LoneLower}.

\begin{lemma}\label{lem:split}
	Then, for any two indices $j, k$, which do not lie in the same component, we have:
	\[ \Prx_{\bA}[k \in \bA \mid j \in \bA, \calbE_{C, \text{no}}] \geq \frac{1}{2} - \frac{\log^4 n}{n}. \]
\end{lemma}

\begin{proof}
	The proof is very straight-forward, we simply count the number of possible partitions $\bA$ for which $j \in \bA$ and are consistent with $C_1, \dots, C_{\alpha}$ and divide by the total number of such partitions. For simplicity, assume that $j$ lies in $C_1(\text{odd})$ and $k$ lies in $C_2(\text{odd})$; the other cases, when $j \in C_1(\text{even})$ or $k \in C_2(\text{even})$ follow from very similar arguments.
	
	We let $X$ be the number of partitions $A \subset [n]$ of size $\frac{n}{2}$ which trigger event $\calbE_{C, \text{no}}$ and have $C_1(\text{odd}) \subset A$ and $C_2(\text{odd}) \subset A$. In order to count these, we first choose which roots $u_3, \dots, u_{\alpha}$ will be included in $A$, and then we pick from the remaining vertices to include in $A$. For a subset $S \subset \{3, \dots, \alpha \}$, we define the quantities:
	\begin{itemize}
		\item $Q = \sum_{i=3}^{\alpha} |C_i|$ is the total vertices assigned from components.
		\item $S_{A} = \sum_{i \in S} |C_i(\text{odd})| + \sum_{i \in [\alpha]\setminus S} |C_i(\text{even})|$ is the total vertices assigned from components to $A$ if we included the roots of components in $S$ in $A$.
		\item $S_{\oA} = Q - S_A$.
	\end{itemize}
	Note that for all subsets $S \subset \{3, \dots, \alpha\}$, we have $S_A \leq \frac{n}{\log^4 n}$. 
	
	Then we have:
	\[ X = \sum_{\ell = 0}^{\alpha - 2} \sum_{\substack{S \subset [3;\alpha] \\ |S|=\ell}} \dbinom{n - Q - |C_1| - |C_2|}{\frac{n}{2} - S_{A} - |C_1(\text{odd})| - |C_2(\text{odd})| }.  \]
	Let $Y$ be the number of partitions $A \subset [n]$ of size $\frac{n}{2}$ which trigger event $\calbE_{C, \text{no}}$ and have $C_1(\text{odd}) \subset A$ and $C_2(\text{even}) \subset A$. Similarly, we have:
	\[ Y = \sum_{\ell = 0}^{\alpha - 2} \sum_{\substack{S \subset [3;\alpha] \\ |S| = \ell}} \dbinom{n - Q - |C_1| - |C_2|}{\frac{n}{2} - S_A - |C_1(\text{odd})| - |C_2(\text{even})|}.\]
	For a particular fixed $S \subset [3;\alpha]$ of size $\ell$, we consider the ratio of the summand in $X$ and in $Y$:
	\begin{align*}
	\dfrac{\dbinom{n - Q - |C_1| - |C_2|}{\frac{n}{2} - S_A - |C_1(\text{odd})| - |C_2(\text{odd})| }}{\dbinom{n - Q - |C_1| - |C_2|}{\frac{n}{2} - S_A - |C_1(\text{odd})| - |C_2(\text{even})|}} &= \dfrac{\left(\frac{n}{2} - S_A - |C_1(\text{odd})| - |C_2(\text{even})|\right)!}{\left( \frac{n}{2} - S_A - |C_1(\text{odd})| - |C_2(\text{odd})|\right)!} \\
	&\qquad\times \dfrac{\left(\frac{n}{2} - S_{\oA} - |C_1(\text{even})| - |C_2(\text{odd})| \right)!}{\left( \frac{n}{2} - S_{\oA} - |C_1(\text{even})| - |C_2(\text{even})|\right)!} \\
	&= \left(1 \pm O\left(\frac{\log n}{n} \right)\right)^{\log n} \left(1 \pm O\left(\frac{\log n}{n} \right) \right)^{\log n} \\
	&= 1 \pm O\left(\frac{\log^2 n}{n} \right),
	\end{align*}
	where we used the fact that $|C_2(\text{even})|, |C_2(\text{odd})| \leq \log n$, and $\frac{n}{2} - S_A - |C_1(\text{odd})| = \Omega(n)$ and $\frac{n}{2} - S_{\oA} - |C_1(\text{odd})| = \Omega(n)$. 
	Thus, we have:
	\[ \dfrac{X}{Y} = 1 \pm O\left( \frac{\log^2 n}{n}\right), \]
	and since:
	\[ \Prx_{\bA}[k \in \bA \mid j \in \bA, \calbE_{C, \text{no}}] = \dfrac{X}{X + Y},\]
	we get the desired claim.
\end{proof}

\section{Reducing to the case $k = \frac{3}{4} n$}\label{app:reduct}

\begin{claim}\label{cl:hehe}
For $\eps < \frac{1}{2}$, let $f \colon \{0, 1\}^n \to \{0, 1\}$ have $\dist(f, \text{$k$-}\Junta) = \eps < \frac{1}{2}$. Then, $g \colon \{0, 1\}^n \times \{0, 1\} \to \{0, 1\}$ given by $g(x, y) = f(x) \oplus y$ has $\dist(g, \text{$(k+1)$-}\Junta) = \eps$.
\end{claim}

\begin{proof}
For the upper bound, suppose $h \colon \{0, 1\}^n \to \{0, 1\}$ had $\dist(f, h) = \eps$. Then, we have that $h' \colon \{0, 1\}^n \times\{0, 1\} \to \{0, 1\}$ given by $h'(x, y) = h(x) \oplus y$ has $\dist(h', g) = \eps$. Thus, we have $\dist(g, \text{$(k+1)$-}\Junta) \leq \dist(f, \text{$k$-}\Junta)$.

For the lower bound, suppose for the sake of contradiction that $h' \colon \{0, 1\}^n \times \{0, 1\} \to \{0, 1\}$ is a $(k+1)$-junta with $\dist(g, h') = \dist(g, \text{$(k+1)$-}\Junta) < \dist(f, \text{$k$-}\Junta)$. We note that since $\eps < \frac{1}{2}$, the last variable must be influential in $h'$. Then, consider the functions $h_0, h_1 \colon \{0, 1\}^n \to \{0, 1\}$ given by $h_0(x) = h'(x, 0)$ and $h_1(x) = h(x, 1)$. Since $y$ is influential in $h'$, $h_0$ and $h_1$ are both $k$-juntas, and therefore 
\[ \dist(h', g) = \dfrac{\dist(h_0, f) + \dist(h_1, \neg f)}{2} \geq \dist(f, \text{$k$-}\Junta), \] 
which is a contradiction.
\end{proof}

\begin{claim}\label{cl:hehe2}
Let $f \colon \{0, 1\}^n \to \{0, 1\}$ have $\dist(f, \text{$k$-}\Junta) = \eps$. Then $g \colon \{0, 1\}^n \times \{0, 1\} \to \{0, 1\}$ given by $g(x, y) = f(x)$ has $\dist(g, \text{$k$-}\Junta) = \eps$. 
\end{claim}

\begin{proof}
For the upper bound, we have that if $h \colon \{0, 1\}^n \to \{0, 1\}$ has $\dist(f, h) = \eps$, then if $h' \colon \{0, 1\}^n \times \{0, 1\} \to \{0, 1\}$ is given by $h(x, y) = h(x)$, then $\dist(h', g) = \eps$. Thus, we have $\dist(g, \text{$(k+1)$-}\Junta) \leq \dist(f, \text{$k$-}\Junta)$. 

For the lower bound, suppose for the sake of contradiction that $h' \colon \{0, 1\}^n \times \{0, 1\} \to \{0, 1\}$ is a $k$-junta with $\dist(g, h') = \dist(g, \text{$k$-}\Junta) < \dist(f, \text{$k$-}\Junta)$. Then, similarly to above, the functions $h_0, h_1 \colon \{0, 1\}^n \to \{0, 1\}$ given by $h_0(x) = h'(x, 0)$ and $h_1(x) = h'(x, 1)$ are $k$-juntas with
\[ \dist(g, \text{$k$-Junta}) = \dist(g, h') = \dfrac{\dist(f, h_0) + \dist(f, h_1)}{2} \geq \eps, \]
which is a contradiction.
\end{proof}

\begin{lemma}
For $0 < \eps_0 < \eps_1 < \frac{1}{2}$, let $B$ be a non-adaptive $(\eps_0, \eps_1)$-tolerant $k$-junta tester for $n(k)$ variable functions making $q(k)$ queries, where $k \leq \alpha n(k)$. Then, there exists a non-adaptive $(\eps_0, \eps_1)$-tolerant $\frac{3n}{4}$-junta tester making $q(O(n))$ queries.
\end{lemma}

\begin{proof}
We give an algorithm which on input $f \colon \{0, 1\}^n \to \{0, 1\}$, determines whether $f$ is $\eps_0$-close from being a $\frac{3n}{4}$-junta or is $\eps_1$-far from being a $\frac{3n}{4}$-junta. The algorithm works as follows: on input $f\colon \{0, 1\}^n \to \{0, 1\}$, we let $g \colon \{0, 1\}^{n} \times \{0, 1\}^{n'} \to \{0, 1\}$ be given by:
\[ g(x, y) = f(x) \oplus \bigoplus_{j=1}^{n'} y_j, \]
where $n' = \max \{ \frac{(4\alpha - 3)n}{4(1-\alpha)}, 0 \}$. Note that if we let $m = n + n'$ (the number of variables in $g$), by Claim~\ref{cl:hehe}, if $f$ is $\eps_0$-close from being a $\frac{3n}{4}$-junta, then $g$ is $\eps_0$-close to being an $\alpha m$-junta, and if $f$ is $\eps_1$-far from being a $\frac{3n}{4}$-junta, then $g$ is $\eps_1$-far from being an $\alpha m$-junta. Finally, we run the tester $B$ with $k = \alpha m$ on $f$, where we add $n(k) - m$ dummy variables. 

The query complexity is given by $q(O(n))$, since $k = O(n)$ when $\alpha < 1$ is a constant.
\end{proof}

\end{document}